\documentclass[reqno]{amsart}
\usepackage[margin=1.6in,bottom=1.5in]{geometry}		% for tighter margins (80 CPL)

\usepackage{amsmath}		% for AMS macros
\usepackage{amssymb}		% for AMS symbols
\usepackage{amsfonts}		% for AMS fonts
\usepackage{amsthm}		% for theorems
\usepackage[foot]{amsaddr}		% for author footnotes

\usepackage[centercolon=true]{mathtools}		% for advanced math

\mathtoolsset{%
}

\usepackage[utf8]{inputenc}		% for source encoding
\usepackage[T1]{fontenc}		% for font encoding

\usepackage[%		% for math font selection
cal=cm,
]
{mathalfa}

\usepackage{dsfont}		% for blackboard bold font
\usepackage[sf,mono=false]{libertine}
\usepackage{acronym}		% for acronyms
\renewcommand*{\aclabelfont}[1]{\acsfont{#1}}		% for acronym label font
\newcommand{\acli}[1]{\emph{\acl{#1}}}		% for italicized acro
\newcommand{\acdef}[1]{\emph{\acl{#1}} \textup{(\acs{#1})}\acused{#1}}		% for acro def
\newcommand{\acdefp}[1]{\emph{\aclp{#1}} \textup{(\acsp{#1})}\acused{#1}}	% for acro def (plural)

\usepackage[labelfont={bf,small},labelsep=colon,font=small]{caption}	% for caption control
\usepackage{subcaption}		% for subfigures
\usepackage[dvipsnames,svgnames]{xcolor}		% for color
\colorlet{MyRed}{FireBrick!50!Crimson}
\colorlet{MyBlue}{DodgerBlue!75!black}
\colorlet{MyGreen}{DarkGreen!85!black}
\colorlet{MyViolet}{DarkMagenta}

\colorlet{MyLightBlue}{DodgerBlue!20}
\colorlet{MyLightGreen}{MyGreen!20}

\colorlet{PrimalColor}{MyBlue}
\colorlet{PrimalFill}{MyLightBlue}
\colorlet{DualColor}{MyRed}

\colorlet{RevColor}{BlueViolet}
\colorlet{LinkColor}{MediumBlue}

\newcommand{\afterhead}{.\;}		% for changing headings
\newcommand{\para}[1]{\smallskip\paragraph{\textbf{#1\afterhead}}}

\usepackage{cancel}		% for cancelling terms
\usepackage{latexsym}		% for symbols

\usepackage{pifont}		% for dingbats
\usepackage{tikz}		% for figures
\usepackage{tikz-cd}		% for figures
\usetikzlibrary{calc,patterns}		% for basic tikz figures
\usepackage{array}		% for flexible tables
\usepackage{booktabs}		% for better tables
\usepackage[inline,shortlabels]{enumitem}		% for lists
\setenumerate{itemsep=\smallskipamount,topsep=\smallskipamount,left=\parindent}
\setitemize{itemsep=\smallskipamount,topsep=\smallskipamount,left=\parindent}

\usepackage[kerning=true]{microtype}		% for microtypography

\usepackage{tabto}		% for spacing
\usepackage{xspace}		% for flexible spaces

\usepackage[sort&compress]{natbib}		% for citations

\bibpunct[, ]{[}{]}{,}{n}{,}{,}

\usepackage{hyperref}
\hypersetup{
final,
colorlinks=true,
linktocpage=true,
pdfstartview=FitH,
breaklinks=true,
pdfpagemode=UseNone,
pageanchor=true,
pdfpagemode=UseOutlines,
plainpages=false,
bookmarksnumbered,
bookmarksopen=false,
bookmarksopenlevel=1,
hypertexnames=false,
pdfhighlight=/O,
urlcolor=LinkColor,linkcolor=LinkColor,citecolor=LinkColor,	% for on-screen
pdftitle={},
pdfauthor={},
pdfsubject={},
pdfkeywords={},
pdfcreator={pdfLaTeX},
pdfproducer={LaTeX with hyperref}
}

\newcommand{\EMAIL}[1]{\email{\href{mailto:#1}{#1}}}

\usepackage[sort&compress,capitalize,nameinlink]{cleveref}		% for cleveref formatting
\crefname{algo}{Algorithm}{Algorithms}
\crefname{assumption}{Assumption}{Assumptions}
\crefname{case}{Case}{Cases}
\usepackage{algpseudocode}		% for algorithm macros
\usepackage{thmtools}		% for theorem tools
\usepackage{thm-restate}		% for restating theorems

\theoremstyle{plain}
\newtheorem{corollary}{Corollary}		% for corollaries
\newtheorem{lemma}{Lemma}		% for lemmas
\newtheorem{proposition}{Proposition}		% for propositions
\newtheorem{proofstep}{Step}		% for proof steps

\newtheorem*{theorem*}{Theorem}		% for theorems
\newtheorem*{corollary*}{Corollary}		% for corollaries (unnumbered)

\theoremstyle{definition}
\newtheorem{definition}{Definition}		% for definitions
\newtheorem{algo}{Algorithm}		% for algorithms
\newtheorem{example}{Example}		% for examples

\newtheorem*{definition*}{Definition}		% for definitions (unnumbered)
\newtheorem*{assumption*}{Assumptions}		% for assumptions (unnumbered)
\newtheorem*{example*}{Example}		% for examples (unnumbered)

\theoremstyle{remark}
\newtheorem*{remark*}{Remark}		% for remarks (unnumbered)
\newtheorem*{notation*}{Notation}		% for notation (unnumbered)

\def\endrem{\hfill{\small$\lozenge$}\smallskip}		% for ending environments
\def\endenv{\hfill{\small$\lozenge$}\smallskip}		% for ending environments
\newcounter{proofpart}

\numberwithin{example}{section}		% for example numbering

\usepackage[showdeletions]{color-edits}		% for editing macros / use [suppress] for final
\newcommand{\debug}[1]{#1}		% for removing macro coloring

\newcommand{\explain}[1]{\tag*{{\sf\#}\:#1}}

\newcommand{\newmacro}[2]{\newcommand{#1}{\debug{#2}}}		% for shorthand definitions
\newcommand{\newop}[2]{\DeclareMathOperator{#1}{\debug{#2}}}		% for shorthand definitions

\DeclarePairedDelimiter{\braces}{\{}{\}}		% for braces
\DeclarePairedDelimiter{\bracks}{[}{]}		% for brackets
\DeclarePairedDelimiter{\parens}{(}{)}		% for parentheses

\DeclarePairedDelimiter{\abs}{\lvert}{\rvert}		% for absolute value
\DeclarePairedDelimiter{\clip}{[}{]}		% for clipping
\DeclarePairedDelimiter{\pospart}{[}{]_{+}}		% for positive part

\DeclarePairedDelimiter{\setof}{\{}{\}}		% for sets
\DeclarePairedDelimiterX{\setdef}[2]{\{}{\}}{#1:#2}		% for set builder notation
\DeclarePairedDelimiterXPP{\exclude}[1]{\mathopen{}\setminus}{\{}{\}}{}{#1}		% for exclusion

\newcommand{\Q}{\mathbb{Q}}		% for rationals
\newcommand{\R}{\mathbb{R}}		% for reals
\DeclareMathOperator*{\argmax}{arg\,max}		% for argmax
\DeclareMathOperator*{\argmin}{arg\,min}		% for argmin
\DeclareMathOperator*{\intersect}{\bigcap}		% for intersections
\DeclareMathOperator*{\union}{\bigcup}		% for unions

\DeclareMathOperator{\littleoh}{o}		% for Landau o
\DeclareMathOperator{\bigoh}{\mathcal{O}}		% for Landau O
\DeclareMathOperator{\cl}{cl}		% for closure
\DeclareMathOperator{\dist}{dist}		% for distance
\DeclareMathOperator{\dom}{dom}		% for domain
\DeclareMathOperator{\im}{im}		% for image
\DeclareMathOperator{\one}{\mathds{1}}		% for indicator
\DeclareMathOperator{\relint}{ri}		% for relative interior
\DeclareMathOperator{\supp}{supp}		% for support
\DeclareMathOperator{\unif}{unif}		% for uniform distribution
\newcommand{\cf}{cf.\xspace}		% for consistency
\newcommand{\eg}{e.g.,\xspace}		% for consistency
\newcommand{\ie}{i.e.,\xspace}		% for consistency
\newcommand{\viz}{viz.\xspace}		% for consistency

\newcommand{\textpar}[1]{\textup(#1\textup)}		% for upshape parentheses

\newcommand{\txs}{\textstyle}		% for forcing inline style

\newcommand{\alt}[1]{#1'}		% for variant version
\newcommand{\altalt}[1]{#1''}		% for second variant
\newmacro{\dd}{\:d}		% for integration
\newcommand{\eps}{\varepsilon}		% for better epsilon
\newcommand{\insum}{\sum\nolimits}		% for compact sums
\newmacro{\const}{c}		% for generic constant
\newmacro{\Const}{C}		% for generic constant
\newmacro{\coefalt}{\mu}		% for generic coefficient
\usepackage{xparse}
\NewDocumentCommand{\coef}{O{\lambda}}{\debug{#1}}
\newmacro{\param}{\theta}		% for parameter
\newmacro{\params}{\Theta}		% for set of parameters

\newmacro{\pexp}{p}		% for first exponent
\newmacro{\qexp}{q}		% for second exponent
\newmacro{\rexp}{r}		% for third exponent

\newmacro{\stepexp}{\ell_{\step}}		% for step-size exponent
\newmacro{\mixexp}{\ell_{\mix}}		% for mixing exponent
\newmacro{\biasexp}{\ell_{\bias}}		% for bias exponent
\newmacro{\noisexp}{\ell_{\sdev}}		% for noise exponent

\newmacro{\radius}{r}

\newmacro{\beforestart}{0}		% for before start index
\newmacro{\start}{1}		% for start index
\newmacro{\afterstart}{2}		% for second index
\newmacro{\running}{\start,\afterstart,\dotsc}		% for running
\newmacro{\halfrunning}{1,3/2,2\dotsc}		% for running

\newmacro{\run}{n}		% for main sequence index
\newmacro{\runalt}{k}		% for variant index
\newmacro{\runaltalt}{m}		% for second variant
\newmacro{\nRuns}{T}		% for total number of runs
\newmacro{\runs}{\mathcal{\nRuns}}		% for set of indices

\newmacro{\state}{X}		% for main iterate
\newmacro{\statealt}{y}		% for variant state
\newmacro{\statealtalt}{z}		% for second variant

\newcommand{\init}[1][\state]{\debug{#1}_{\start}}		% for initial iterate (X by default)
\newcommand{\iter}[1][\state]{\debug{#1}_{\runalt}}		% for running iterate (X by default)
\newcommand{\prev}[1][\state]{\debug{#1}_{\run-1}}		% for previous iterate (X by default)
\newcommand{\curr}[1][\state]{\debug{#1}_{\run}}		% for current iterate (X by default)
\newcommand{\curralt}[1][\state]{\debug{#1}_{\runalt}}         % for alternative current iterate (X by default)
\renewcommand{\next}[1][\state]{\debug{#1}_{\run+1}}		% for next iterate (X by default)

\newcommand{\beforelead}[1][\state]{\debug{\tilde#1}_{\run-1}}		% for prev lead iterate (X by default)
\newcommand{\lead}[1][\state]{\debug{\tilde#1}_{\run}}		% for lead iterate (X by default)
\newmacro{\tstart}{0}		% for time start
\newmacro{\timealt}{s}		% for dummy continuous time
\newmacro{\horizon}{T}		% for horizon

\newmacro{\traj}{x}		% for trajectory
\newmacro{\trajalt}{y}		% for variant trajectory
\newmacro{\trajaltalt}{z}		% for second variant

\newmacro{\flow}{\phi}		% for (semi)flow map
\DeclarePairedDelimiterXPP{\flowof}[2]{\flow_{#1}}{(}{)}{}{#2}		% for flow

\newop{\Nash}{NE}		% for Nash equilibria
\newop{\CE}{CE}		% for correlated equilibria
\newop{\CCE}{CCE}		% for Hannan set
\newop{\NI}{NI}		% for Nikaido-Isoda function

\newop{\brep}{\mathtt{br}}		% for best responses
\newop{\btr}{\mathtt{btr}}		% for better responses
\newop{\reg}{Reg}		% for regret
\newop{\preg}{\overline{Reg}}		% for pseudo-regret
\newop{\val}{val}		% for value function

\newmacro{\strat}{x}		% for mixed strategy
\newmacro{\stratalt}{\alt\strat}		% for variant strategy
\newmacro{\strats}{\mathcal{X}}		% for set of mixed strategies
\newmacro{\intstrats}{\relint\strats}		% for set of interior strategies
\newcommand{\eq}{\sol[\strat]}		% for Nash equilibrium
\newmacro{\eqs}{\sol[\strats]}		% for set of Nash equilibria

\newmacro{\corr}{z}		% for mixed strategy
\newmacro{\corralt}{\alt\corr}		% for variant strategy
\newmacro{\corrs}{\mathcal{Z}}		% for set of mixed strategies

\newmacro{\play}{i}		% for player index
\newmacro{\playalt}{j}		% for variant player index
\newmacro{\playaltlalt}{k}		% for second variant
\newmacro{\nPlayers}{N}		% for number of players
\newmacro{\players}{\mathcal{\nPlayers}}		% for set of players

\newmacro{\pure}{\alpha}		% for pure strategy
\newmacro{\purealt}{\alt\pure}		% for variant pure strategy
\newmacro{\purealtalt}{\gamma}		% for second variant
\newmacro{\pureq}{\sol[\pure]}		% for pure strategy
\newmacro{\nPures}{A}		% for number of pure strategies
\newmacro{\pures}{\mathcal{\nPures}}		% for set of pure strategies

\newmacro{\loss}{\ell}		% for loss function
\newmacro{\pay}{u}		% for payoff function
\newmacro{\payv}{v}		% for payoff vector
\newmacro{\pot}{f}		% for potential function

\newmacro{\game}{\mathcal{G}}		% for game
\newmacro{\gamefull}{\game(\players,\points,\pay)}		% for full game

\newmacro{\fingame}{\Gamma}		% for finite game
\newmacro{\fingamefull}{\Gamma(\players,\pures,\pay)}		% for full finite game
\newmacro{\mixgame}{\Delta(\fingame)}		% for mixed extension

\newmacro{\gmat}{g}		% for metric tensor
\newmacro{\gdist}{\dist_{\gmat}}
\newmacro{\mfld}{M}		% for manifold
\newmacro{\form}{\omega}		% for generic form

\newmacro{\tvec}{z}		% for tangent vector
\newmacro{\tvecs}{\mathcal{Z}}		% for tangent vectors
\newmacro{\uvec}{u}		% for unit vector

\newmacro{\ball}{\basin}		% for ball
\newmacro{\sphere}{\mathbb{S}}		% for sphere

\newmacro{\graph}{\mathcal{G}}
\newmacro{\vertices}{\mathcal{V}}
\newmacro{\edges}{\mathcal{E}}

\newmacro{\mat}{M}		% for generic matrix
\newmacro{\jmat}{J}		% for Jacobian matrix
\newmacro{\hmat}{H}		% for Hessian matrix

\newop{\row}{row}		% for row space
\newop{\col}{col}		% for row space

\newmacro{\ones}{\mathbf{1}}		% for matrix of ones
\newmacro{\eye}{I}		% for identity matrix
\newmacro{\zer}{\mathbf{0}}		% for zero matrix

\DeclarePairedDelimiter{\norm}{\lVert}{\rVert}		% for norm
\DeclarePairedDelimiterXPP{\dnorm}[1]{}{\lVert}{\rVert}{_{\infty}}{#1}		% for dual norm

\DeclarePairedDelimiterXPP{\onenorm}[1]{}{\lVert}{\rVert}{_{1}}{#1}		% for dual norm
\DeclarePairedDelimiterXPP{\twonorm}[1]{}{\lVert}{\rVert}{_{2}}{#1}		% for dual norm
\DeclarePairedDelimiterXPP{\supnorm}[1]{}{\lVert}{\rVert}{_{\infty}}{#1}		% for dual norm
\DeclarePairedDelimiterXPP{\pnorm}[1]{}{\lVert}{\rVert}{_{\pexp}}{#1}
\DeclarePairedDelimiterXPP{\qnorm}[1]{}{\lVert}{\rVert}{_{\qexp}}{#1}

\DeclarePairedDelimiterX{\braket}[2]{\langle}{\rangle}{#1,#2}		% for brakets

\DeclarePairedDelimiterX{\inner}[2]{\langle}{\rangle}{#1,#2}		% for scalar product

\newmacro{\vecspace}{\mathcal{V}}		% for generic vector space
\newmacro{\subspace}{\mathcal{W}}		% for vector subspace

\newmacro{\coord}{i}		% for coordinate index
\newmacro{\coordalt}{j}		% for variant coordinate
\newmacro{\coordaltalt}{k}		% for second variant
\newmacro{\nCoords}{n}		% for number of coordinates
\newmacro{\dims}{\nCoords}		% for dimension
\newmacro{\vdim}{\nCoords}		% for dimension (legacy alias)

\newmacro{\pvec}{z}		% for primal vector
\newmacro{\pvecalt}{r}		% for primal vector
\newmacro{\bvec}{e}		% for basis vector
\newmacro{\bvecs}{\mathcal{E}}		% for basis vectors
\newmacro{\cvec}{b}     % for column vector
\newmacro{\cvecalt}{d}     % for column vector

\newmacro{\pspace}{\vecspace}		% for primal space
\newmacro{\dspace}{\vecspace^{\ast}}		% for dual space

\newmacro{\dvec}{\dpoint}		% for dual vector
\newmacro{\dbvec}{\eps}		% for dual basis vectors

\newmacro{\dpoint}{y}		% for generic dual point
\newmacro{\dpointalt}{\alt\dpoint}		% for variant dual point
\newmacro{\dpointaltalt}{\altalt\dpoint}		% for second variant
\newmacro{\dpoints}{\mathcal{Y}}		% for set of dual points

\newmacro{\dstate}{Y}		% for dual state
\newmacro{\dbase}{v}		% for dual base

\newcommand{\defeq}{\coloneqq}		% for direct definition
\newcommand{\from}{\colon}		% for function definition
\newcommand{\too}{\rightrightarrows}		% for correspondences
\newop{\Opt}{Opt}		% for value of problem
\newop{\Sol}{Sol}		% for solution of problem
\newop{\gap}{Gap}		% for gap function

\newmacro{\orcl}{V}		% for oracle input
\newmacro{\err}{Z}		% for error
\newop{\IWE}{IWE}		% for oracle input

\newmacro{\tfun}{f}		% for test function
\newmacro{\obj}{f}		% for objective function
\newmacro{\objalt}{g}		% for variant objective (smooth etc.)
\newmacro{\sobj}{F}		% for stochastic objective

\newmacro{\gvec}{g}		% for gradient vector
\newmacro{\oper}{A}		% for operator
\newmacro{\vecfield}{v}		% for vector field (selection etc.)
\newcommand{\sol}[1][\point]{#1^{\ast}}		% for solution point (x by default)
\newmacro{\solvec}{\vecfield(\sol)}		% for vector at a solution
\newmacro{\solpay}{\eq[\payv]}		% for vector at a solution

\newmacro{\signal}{\est\payv}		% for signal
\newmacro{\step}{\gamma}		% for step-size
\newmacro{\learn}{\eta}		% for learning rate

\newmacro{\vbound}{G}		% for vector bound
\newmacro{\lips}{L}		% for Lipschitz modulus
\newmacro{\strong}{\mu}		% for strong convexity modulus
\newmacro{\smooth}{\beta}		% for strong smoothness modulus

\newop{\tcone}{TC}		% for tangent cone
\newop{\dcone}{\tcone^{\ast}}		% for dual cone
\newop{\ncone}{NC}		% for normal cone
\newop{\pcone}{PC}		% for polar cone
\newop{\hull}{\Delta}		% for hull

\newmacro{\cvx}{\mathcal{C}}		% for generic convex set
\newmacro{\subd}{\partial}		% for subdifferential
\newmacro{\minmax}{\mathcal{L}}		% for minmax objective

\newmacro{\minvar}{\point_{1}}		% for minimization variable
\newmacro{\minvaralt}{\alt\minvar}		% for variant minvar
\newmacro{\minvars}{\points_{1}}		% for set of minvars
\newmacro{\minsol}{\sol[\minvar]}		% for min solution

\newmacro{\maxvar}{\point_{2}}		% for maximization variable
\newmacro{\maxvaaltr}{\alt\maxvar}		% for variant maxvar
\newmacro{\maxvars}{\points_{2}}		% for set of maxvars
\newmacro{\maxsol}{\sol[\maxvar]}		% for max solution

\newop{\Eucl}{\Pi}		% for Euclidean projection
\newop{\logit}{\Lambda}		% for logit map
\newop{\dkl}{KL}		% for Kullback Leibler

\newmacro{\hreg}{h}		% for regularizer
\newmacro{\hconj}{\hreg^{\ast}}		% for regularizer
\newmacro{\breg}{D}		% for Bregman divergence
\newmacro{\mprox}{P}		% for Bregman prox-mapping
\newmacro{\mirror}{Q}		% for mirror map
\newmacro{\fench}{F}		% for Fenchel coupling
\newmacro{\depth}{H}		% for regularizer depth
\newmacro{\hstr}{K}		% for strong convexity constant
\newmacro{\hker}{\theta}		% for regularizer kernel

\newmacro{\proxdom}{\strats_{\hreg}}		% for prox-domain
\newmacro{\proxdomi}{\strats_{\hreg_{\play}}}		% for prox-domain of player i
\newmacro{\zone}{\mathbb{D}}		% for Bregman zone

\DeclarePairedDelimiterXPP{\proxof}[2]{\mprox_{#1}}{(}{)}{}{#2}		% for Bregman prox step

\newmacro{\point}{p}		% for generic point
\newmacro{\pointalt}{q}		% for variant point
\newmacro{\pointaltalt}{\altalt\point}		% for second variant
\newmacro{\points}{\mathcal{X}}		% for set of points
\newmacro{\intpoints}{\relint\points}		%for point set interior

\newmacro{\base}{p}		% for reference point
\newmacro{\basealt}{q}		% for variant reference point
\newmacro{\basealtalt}{u}		% for second variant

\newmacro{\set}{\mathcal{S}}		% for generic set
\newmacro{\open}{\mathcal{U}}		% for open sets
\newmacro{\closed}{\mathcal{C}}		% for closed sets
\newmacro{\cpt}{\mathcal{K}}		% for compact sets
\newmacro{\nhd}{\mathcal{U}}		% for neighborhoods
\newmacro{\nhdalt}{\alt\nhd}		% for other neighborhood

\newop{\ex}{\mathbb{E}}		% for expectations
\newop{\prob}{\mathbb{P}}		% for probability
\newop{\Var}{Var}		% for variance
\newop{\simplex}{\hull}		% for simplices
\renewcommand{\P}{{\mathbb{P}}} % Victor wants to type fast -- don't ware about usual \P which is the paragraph symbol

\providecommand\given{}		% empty command for conditionals

\DeclarePairedDelimiterXPP{\exwrt}[2]{\ex_{#1}}{[}{]}{}{%		% for conditional expectations
\renewcommand\given{\nonscript\:\delimsize\vert\nonscript\:\mathopen{}} #2}

\DeclarePairedDelimiterXPP{\exof}[1]{\ex}{[}{]}{}{%		% for conditional expectations
\renewcommand\given{\nonscript\,\delimsize\vert\nonscript\,\mathopen{}} #1}

\DeclarePairedDelimiterXPP{\probof}[1]{\prob}{(}{)}{}{%		% for conditional probabilities
\renewcommand\given{\nonscript\:\delimsize\vert\nonscript\:\mathopen{}} #1}

\DeclarePairedDelimiterXPP{\oneof}[1]{\one}{\{}{\}}{}{%		% for conditional expectations
\renewcommand\given{\nonscript\,\delimsize\vert\nonscript\,\mathopen{}} #1}

\newmacro{\sample}{\omega}		% for sample
\newmacro{\samples}{\Omega}		% for set of samples

\newmacro{\seed}{\theta}		% for seed
\newmacro{\seeds}{\Theta}		% for seed space

\newmacro{\filter}{\mathcal{F}}		% for filtration
\newmacro{\probspace}{(\samples,\filter,\prob)}		% for probability space

\newmacro{\history}{\mathcal{H}}		% for filtrations

\newcommand{\as}{\debug{\textpar{a.s.}}\xspace}		% for almost surely
\newmacro{\event}{E}       % for event
\newmacro{\eventalt}{H}       % for variant event

\newmacro{\mean}{\mu}		% for mean of distribution
\newmacro{\sdev}{\sigma}		% for mean of distribution
\newmacro{\variance}{\sdev^{2}}		% for mean of distribution
\newmacro{\aggnoise}{\mathrm{\uppercase\expandafter{\romannumeral1}}}		% for aggregate
\newmacro{\supnoise}{\aggnoise_{\infty}}		% for supremum
\newmacro{\maxnoise}{\aggnoise^{\ast}}		% for maximal

\newmacro{\aggbias}{\mathrm{\uppercase\expandafter{\romannumeral2}}}		% for aggregate
\newmacro{\supbias}{\aggbias_{\infty}}		% for supremum
\newmacro{\maxbias}{\aggbias^{\ast}}		% for maximal

\newmacro{\second}{\psi}		% for second moment
\newmacro{\aggsecond}{\mathrm{\uppercase\expandafter{\romannumeral3}}}		% for aggregate
\newmacro{\supsecond}{\aggsecond_{\infty}}		% for supremum
\newmacro{\maxsecond}{\aggsecond^{\ast}}		% for maximal

\newmacro{\submart}{S}      % for submartingales
\newmacro{\toterr}{R}      % for total error

\newmacro{\thres}{\eps}		% for tolerance level
\newmacro{\conf}{\eta}		% for confidence level
\newmacro{\stoptime}{N}		% for stopping time

\newmacro{\mart}{M}  % contraction of martingale
\newmacro{\martbd}{\sigma}     % martingale bound

\newcommand{\est}[1]{\hat #1}		% for estimates

\newmacro{\efftime}{\tau}		% for proper time
\newmacro{\seq}{Z}		% for APT (X by default)
\newcommand{\apt}{\seq}		% for APT (X by default)
\DeclarePairedDelimiterXPP{\aptof}[1]{\apt}{(}{)}{}{#1}		% for flow

\newmacro{\error}{Z}		% for error
\newmacro{\noise}{U}		% for noise
\newmacro{\bias}{b}		% for bias
\newmacro{\brown}{W}		% for Wiener process

\newmacro{\serror}{\theta}		% for scalar error
\newmacro{\snoise}{\xi}		% for scalar noise
\newmacro{\sbias}{\psi}		% for scalar bias

\newmacro{\sbound}{M}		% for signal bound
\newmacro{\bbound}{B}		% for bias bound
\newmacro{\noisepar}{\sdev}		% for noise parameter
\newmacro{\noisevar}{\variance}		% for noise variance

\newmacro{\mix}{\delta}		% for query radius
\newmacro{\unitvec}{w}		% for query direction
\newmacro{\unitvar}{W}		% for query direction
\newmacro{\perturb}{z}		% for perturbation

\newmacro{\purequery}{\est\pure}		% for pure strategy query
\newmacro{\query}{\est\state}		% for query state
\newmacro{\pivot}{\point}		% for pivot point
\newmacro{\querypoint}{\est\point}		% for query point

\addauthor[Victor]{VB}{red}

\addauthor[Pan]{PM}{MediumBlue}

\newmacro{\pdist}{P}		% for seed law

\newmacro{\basin}{\mathcal{B}}		% for basin
\newmacro{\dbasin}{\mathcal{D}}		% for dual basin
\newmacro{\limset}{\mathcal{L}}		% for limit set

\newmacro{\pureset}{\mathcal{C}}		% for subset of pures
\newmacro{\prodset}{\mathcal{P}}		% for product set
\newmacro{\faces}{\mathcal{P}(\strats)}		% for product set

\newmacro{\score}{y}		% for score
\newmacro{\scorealt}{\alt\score}		% for alternate score
\newmacro{\scores}{\mathcal{Y}}		% for set of scores

\newmacro{\pflow}{\chi}		% for primal flow map
\DeclarePairedDelimiterXPP{\pflowof}[2]{\pflow_{#1}}{(}{)}{}{#2}		% for primal flow

\newmacro{\dflow}{\psi}		% for dual flow map
\DeclarePairedDelimiterXPP{\dflowof}[2]{\dflow_{#1}}{(}{)}{}{#2}		% for dual flow

\newmacro{\rate}{\varphi}	% for rate function
\newmacro{\hinv}{w}	% for inverse rate function
\newmacro{\totinv}{W}	% for rate function

\newmacro{\limpoint}{\hat\strat}		% for limit point

\newmacro{\playI}{R}		% for row player
\newmacro{\playII}{C}		% for column player

\newop{\probalt}{\mathbb{Q}}		% for probability
\newmacro{\good}{\event}
\newmacro{\bad}{\mathcal{N}}
\newmacro{\lyap}{\Phi}
\newmacro{\imlyap}{\Psi}
\newmacro{\gauge}{\kappa}

\newmacro{\energy}{E}
\newmacro{\ediff}{\Delta\energy}
\newmacro{\elvl}{a}
\newmacro{\ebound}{2}
\newmacro{\esmooth}{\beta}
\newmacro{\emax}{\energy_{\ast}}
\newmacro{\hien}{\energy_{+}}
\newmacro{\loen}{\energy_{-}}

\newmacro{\texp}{\alpha}		% for growth exponent

\begin{document}

%**********************************************************************
%***    FRONTMATTER AND METADATA
%**********************************************************************

%----------------------------------------------------------------------
%%% TITLE & AUTHORS
%----------------------------------------------------------------------
\title
[Strategic stability under regularized learning]
{The equivalence of dynamic and strategic stability\\under regularized learning in games}

%-------------------------------------------------------------------
\author
[V.~Boone]
{Victor Boone$^{\ast}$}
\address{$^{\ast}$\,%
Univ. Grenoble Alpes, CNRS, Inria, Grenoble INP, LIG, 38000 Grenoble, France.}
\EMAIL{victor.boone@univ-grenoble-alpes.fr}
%-------------------------------------------------------------------
\author
[P.~Mertikopoulos]
{Panayotis Mertikopoulos$^{\ast}$}
\EMAIL{panayotis.mertikopoulos@imag.fr}
%-------------------------------------------------------------------

%----------------------------------------------------------------------
%%% KEYWORDS
%----------------------------------------------------------------------
\subjclass[2020]{%
Primary 91A10, 91A26;
secondary 68Q32, 62L20.}

\keywords{%
Regularized learning;
strategic stability;
asymptotic stability;
closedness under better replies;
resilience.}

%----------------------------------------------------------------------
%%% ACRONYMS
%----------------------------------------------------------------------
\newacro{LHS}{left-hand side}
\newacro{RHS}{right-hand side}

\newacro{BDG}{Burkholder\textendash Davis\textendash Gundy}
\newacro{iid}[i.i.d.]{independent and identically distributed}
\newacro{lsc}[l.s.c.]{lower-semicontinuous}
\newacro{whp}[w.h.p]{with high probability}
\newacro{wp1}[w.p.$1$]{with probability $1$}

\newacro{APT}{asymptotic pseudotrajectory}
\newacroplural{APT}{asymptotic pseudotrajectories}
\newacro{ICT}{internally chain transitive}
\newacroplural{ICT}{internally chain transitive sets}
\newacro{SA}{stochastic approximation}
\newacro{ODE}{ordinary differential equation}

\newacro{CCE}{coarse correlated equilibrium}
\newacroplural{CCE}[CCE]{coarse correlated equilibria}
\newacro{NE}{Nash equilibrium}
\newacroplural{NE}[NE]{Nash equilibria}
\newacro{VI}{variational inequality}
\newacroplural{VI}[VIs]{variational inequalities}

\newacro{FTRL}{``follow-the-regularized-leader''}
\newacro{FTGL}{``follow-the-generalized-leader''}
\newacro{OFTRL}[Opt-FTRL]{optimistic \acs{FTRL}}
\newacro{BFTRL}[B-FTRL]{bandit \acs{FTRL}}
\newacro{EW}[\textsc{Hedge}]{exponential\,/\,multiplicative weights}
\newacro{MWU}{multiplicative weights update}
\newacro{CMW}{clairvoyant multiplicative weights}
%\newacro{OMW}[\textsc{Hedge-Opt}]{optimistic {\scshape Hedge}}
\newacro{OMW}{optimistic multiplicative weights}
\newacro{IWE}{importance-weighted estimator}
\newacro{EXP3}{exponential weights algorithm for exploration and exploitation}
\newacro{INF}{implicitly normalized forecaster}
\newacro{Tsallis}[\textsc{Tsallis-INF}]{Tsallis \acl{INF}}
\newacro{MD}{mirror descent}
\newacro{MP}{mirror-prox}
\newacro{OMD}{optimistic mirror descent}

\newacro{curb}{closed under rational behavior}
\newacro{closed}[club]{closed under better replies}
\newacro{minimal}[m-club]{minimally \acs{closed}}
\newacro{KKT}{Karush\textendash Kuhn\textendash Tucker}
\newacro{GFO}{generalized first-order oracle}
\newacro{SFO}{stochastic first-order oracle}
\newacro{RM}{Robbins\textendash Monro}
\newacro{RLD}{regularized learning dynamics}
\newacro{RSA}{regularized stochastic approximation}
\newacro{RRM}{regularized Robbins\textendash Monro}

\newacro{method}[RL]{regularized learning}

%----------------------------------------------------------------------
%%% ABSTRACT
%----------------------------------------------------------------------
\begin{abstract}
%----------------------------------------------------------------------
%%% ABSTRACT
%----------------------------------------------------------------------
% !TEX root = ./Main.tex
%
%
In this paper, we examine the long-run behavior of regularized, no-regret learning in finite games.
A well-known result in the field states that the empirical frequencies of no-regret play converge to the game's set of coarse correlated equilibria;
however, our understanding of how the players' \emph{actual} strategies evolve over time is much more limited \textendash\ and, in many cases, non-existent.
This issue is exacerbated by a series of recent results showing that \emph{only} strict Nash equilibria are stable and attracting under regularized learning, thus making the relation between learning and pointwise solution concepts particularly elusive.
In lieu of this, we take a more general approach and instead seek to characterize the \emph{setwise} rationality properties of the players' day-to-day play.
%As a first result, we show that the long-run limit of regularized learning is a set that is almost surely \emph{resilient}, \ie every unilateral deviation from said set is deterred by some other element thereof.
To that end, we focus on one of the most stringent criteria of setwise strategic stability, namely that any unilateral deviation from the set in question incurs a cost to the deviator \textendash\ a property known as \emph{closedness under better replies} (club).
In so doing, we obtain a far-reaching equivalence between strategic and dynamic stability:
\emph{a product of pure strategies is closed under better replies if and only if its span is stable and attracting under regularized learning.}
In addition, we estimate the rate of convergence to such sets, and we show that methods based on entropic regularization (like the exponential weights algorithm) converge at a geometric rate, while projection-based methods converge within a \emph{finite} number of iterations, even with bandit, payoff-based feedback.
\end{abstract}

%**********************************************************************
%***    BODY TEXT
%**********************************************************************
\maketitle

%----------------------------------------------------------------------
%%% GLOBAL COMMANDS
%----------------------------------------------------------------------
\allowdisplaybreaks		% for breaking long displays
\acresetall		% for resetting acros after abstract

%----------------------------------------------------------------------
%%% INTRODUCTION
%----------------------------------------------------------------------
\section{Introduction}
\label{sec:introduction}
%----------------------------------------------------------------------
%%% INTRODUCTION
%----------------------------------------------------------------------
% !TEX root = ../Main.tex

%----------------------------------------------------------------------
%% Background
%----------------------------------------------------------------------
%\para{Background}

The question of whether players can learn to emulate rational behavior through repeated interactions has been one of the mainstays of non-cooperative game theory, and it has recently gained increased momentum owing to a surge of breakthrough applications to machine learning and data science, from online ad auctions to multi-agent reinforcement learning.
Informally, this question can be stated as follows:
\begin{center}
\itshape
If every player follows an iterative procedure aiming to increase their individual payoff,\\
does the players' long-run behavior converge to a rationally admissible state?
\end{center}

A natural setting for studying this question is to assume that each player is following a no-regret algorithm, \ie a policy which is asymptotically as good against a given sequence of payoff functions as the best fixed strategy in hindsight.
In this framework, the link between learning and rationality is provided by a folk result which states that, under no-regret learning, the empirical frequency of play converges to the game's set of \acdefp{CCE} \textendash\ also known as the game's \emph{Hannan set} \cite{Han57}.
This result has been of seminal importance to the field because no-regret play can be achieved via a wide class of ``regularized learning'' policies, as exemplified by the \acdef{FTRL} family of algorithms \cite{SSS06,SS11} and its variants \textendash\ optimistic methods \cite{RS13-NIPS,RS13-COLT,SALS15,DP19,HIMM19}, \acs{EW}\,/\,\acs{EXP3} \cite{ACBFS02,ACBFS95,CBL06,BCB12}, implicitly normalized forecasters \cite{AB10,ALT15}, etc.

All these policies have (at least) one thing in common:
they seek to provide the tightest possible guarantees for each player's individual regret, thus accelerating convergence to the game's Hannan set.
As such, in games where the marginalization of \aclp{CCE} coincides with the game's \aclp{NE} (like two-player zero-sum games), we obtain a positive equilibrium convergence guarantee:
the long-run empirical frequency of play evolves ``as if'' the players were rational to begin with \textendash\ \ie as if they had full knowledge of the game,  common knowledge of rationality, the ability to communicate this knowledge, etc.
On the other hand,
%in many other contexts,
%\textendash\ and, in particular, in the context of regularized learning \textendash\ players learn \emph{independently} from one another, with no common correlating device.
%By comparison, the Hannan set consists of \emph{correlated} strategies which, when marginalized,
the marginals of Hannan-consistent correlated strategies
may fail even the weakest axioms of rationalizability (such as the elimination of strictly dominated strategies).
In particular, a well-known example of \citet{VZ13} (which we discuss in detail in \cref{sec:resilience}) shows that it is possible to have \emph{negative regret} for all time, but still employ \emph{only strictly dominated strategies} throughout the entire horizon of play.

The reason for this disconnect is that no-regret play has significant predictive power for the empirical frequency of play \textendash\ that is, the long-run empirical distribution of pure strategy \emph{profiles} \textendash\ but much less so for the players' day-to-day sequence of play \textendash\ \ie the evolution of the players' \emph{actual} mixed strategies over time.
In particular, even when the marginalization of the Hannan set is Nash, the actual trajectory of play may \textendash\ and, in fact, often \emph{does} \textendash\ diverge away from the game's set of equilibria \cite{DISZ18,GBVV+19,MLZF+19,MPP18,MZ19} or exhibits chaotic, unpredictable oscillations \cite{PPP17,CFMP20}.
Thus, especially in the context of regularized learning \textendash\ where players learn \emph{independently} from one another, with no common correlating device \textendash\ the blanket guarantee of no-regret play may quickly become irrelevant,
%and even misleading,
providing the veneer of rational behavior but not the substance.

Motivated by the above, our paper seeks to understand the rationality properties of the players' \emph{actual} sequence of play under regularized learning, as encoded by the following question:
\begin{center}
\centering
\itshape
Which sets of mixed strategies are stable and attracting under regularized learning?\\
%\textpar{as opposed to sets of correlated strategies, like the Hannan set}?
Are these sets robust to strategic deviations?
And, if so, is the converse also true?
\end{center}

%----------------------------------------------------------------------
%% Related
%----------------------------------------------------------------------
\para{Our contributions in the context of related work}

This question has attracted significant interest in the literature, especially in its pointwise version, namely:
Which mixed strategy profiles are stable and attracting under regularized learning?
Are the dynamics' stable states robust to unilateral deviations?
And, if so, are these the only stable states of regularized learning?

In the related setting of population games, the answer to this question is sometimes referred to as the ``folk theorem of evolutionary game theory'' \cite{Cre03,Wei95,HS03}.
Somewhat informally, this theorem states that, under the replicator dynamics (the continuous-time analogue of the \acl{EW} algorithm, itself an archetypal regularized learning method), the following is true for \emph{all} games:
only \aclp{NE} are (Lyapunov) stable,
and
a state is stable and attracting under the replicator dynamics if and only it is a strict \acl{NE} of the underlying game \cite{Wei95,HS03}.%
%\footnote{Strictness here means that each player has a unique best response at equilibrium.}

In the context of regularized learning, \cite{CGM15,MS16,FVGL+20} showed that a similar equivalence holds for the dynamics of \ac{FTRL} in \emph{continuous} time:
a state is stable and attracting under the \ac{FTRL} dynamics if and only if it is a strict \acl{NE}.
Subsequently, \citet{GVM21,GVM21b} extended this equivalence to an entire class of regularized learning schemes, with different types of feedback and/or update structures \textendash\ from optimistic methods to algorithms run with bandit, payoff-based information.
In all these cases, the same principle emerges:
under regularized learning,
\emph{a state is asymptotically stable and attracting if and only if it is a strict \acl{NE}}.

This is an important pointwise prediction but it does not cover cases where regularized learning algorithms do not converge to a point, but to a \emph{set} (such as a limit cycle or other non-trivial attractor).
In this case, the very definition of strategic stability is an intricate affair, and there are several definitions that come into play \cite{BW91,RW95,DR03,FT91}.
The first such notion that we consider is that of ``resilience to strategic deviations'', namely that every unilateral deviation from the set under study is deterred by some other element thereof.
Our first contribution in this direction is a universal guarantee to the effect that, \acl{wp1}, in any game, and from any initial condition,
\emph{the long-run limit of any regularized learning algorithm is a resilient set.}

This result is significant in its universality, but the notion of resilience is not sufficiently strong to disallow irrational behavior \textendash\ and, in fact, it is subject to similar shortcomings as Hannan consistency.
To account for this deficiency, we turn to a much more stringent criterion of setwise strategic stability, that of \emph{closedness under better replies} (\acs{closed}).
This notion, originally due to \citet{RW95}, states that any deviation from a product of pure strategies is costly, and it is one of the strictest setwise refinements in game theory.
In particular, it refines the notion of closedness under rational behavior (\acs{curb}) \cite{BW91}, and it satisfies all the seminal strategic stability requirements of \citet{KM86}, including robustness to strategic payoff perturbations.%
\footnote{Roughly speaking, robustness to strategic payoff perturbations means that the set under study remains stable even if the payoffs of the game are subject to small \textendash\ but possibly adversarial \textendash\ perturbations.}

In this general context, we show that regularized learning enjoys a striking relation with \acs{closed} sets:
\emph{A product of pure strategies is closed under better replies if and only if its span is stable and attracting under regularized learning.}
In fact, we show that this equivalence can be refined to sets that are \emph{minimally} \acl{closed} (in the sense that they do not contain a strictly smaller \ac{closed} set):
a product of puer strategies is \ac{minimal} if and only if its span is irreducibly stable and attracting (in that it does not contain a smaller asymptotically stable span of strategies).
Finally, we also estimate the rate of convergence to \acs{closed} sets, and we establish convergence at a geometric rate for entropically regularized methods \textendash\ like \acs{EW} and \acs{EXP3} \textendash\ and in a \emph{finite number} of iterations under projection-based methods.

In light of the above, our results can be seen both as a far-reaching setwise generalization of the folk theorem of evolutionary game theory, as well as a bona fide algorithmic analogue of a precursor result for the replicator dynamics, originally due to \citet{RW95}.
Importantly, our analysis covers several different update structures \textendash\ ``vanilla'' regularized methods, but also their optimistic variants \textendash\ as well as a wide range of information models \textendash\ from full payoff information to bandit, payoff-based feedback.
%In this regard, the universality of \acs{club} is highly surprising as it remains robust, even when players lose any ``as if'' rationality and can only reason based on their in-game rewards.

%----------------------------------------------------------------------
%%% SETUP
%----------------------------------------------------------------------
\section{Preliminaries}
\label{sec:prelims}
%----------------------------------------------------------------------
%%% PRELIMS
%----------------------------------------------------------------------
% !TEX root = ../Main.tex

We start by recalling some basic facts and definitions from game theory, roughly following the classical textbook of \citet{FT91}.
First, a \emph{finite game in normal form} consists of
\begin{enumerate*}
%[(\itshape i\hspace*{.5pt}\upshape)]
[\itshape a\upshape)]
\item
a finite set of \emph{players} $\play\in\players \equiv \{1,\dotsc,\nPlayers\}$;
\item
a finite set of \emph{actions} \textendash\ or \emph{pure strategies} \textendash\ $\pures_{\play}$ per player $\play\in\players$;
and
\item
an ensemble of \emph{payoff functions} $\pay_{\play}\from\prod_{\playalt}\pures_{\playalt}\to\R$, each determining the reward $\pay_{\play}(\pure)$ of player $\play\in\players$ in a given \emph{action profile} $\pure = (\pure_{1},\dotsc,\pure_{\nPlayers})$.
\end{enumerate*}
Collectively, we will write $\pures = \prod_{\playalt} \pures_{\playalt}$ for the game's \emph{action space} and $\fingame \equiv \fingamefull$ for the game with primitives as above.

%The tuple $\pure = (\pure_{\play})_{\play\in\players}$ where the $\play$-th player chooses $\pure_{\play}\in\pures_{\play}$ will be referred to as the players' \emph{action profile}, and we will
%write $\pures = \prod_{\play} \pures_{\play}$ for the game's \emph{action space}.
%where $\pures = \prod_{\playalt} \pure_{\playalt}$ denotes the that determine the reward $\pay_{\play}(\pure)$ of each player $\play\in\players$ given an action profile $\pure\in\pures$.

During play, each player $\play\in\players$ may randomize their choice of action by playing a \emph{mixed strategy}, \ie a probability distribution $\strat_{\play} \in \strats_{\play} \defeq \simplex(\pures_{\play})$ over $\pures_{\play}$
that selects $\pure_{\play}\in\pures_{\play}$ with probability $\strat_{\play\pure_{\play}}$.
%In this case, we will write
%$\strat_{\play\pure_{\play}}$ for the probability with which player $\play\in\players$ selects $\pure_{\play}\in\pures_{\play}$,
To lighten notation, we identify $\pure_{\play} \in \pures_{\play}$ with the mixed strategy that assigns all weight to $\pure_{\play}$ (thus justifying the terminology ``pure strategies'').
Then, writing
$\strat = (\strat_{\play})_{\play\in\players}$ for the players' \emph{strategy profile}
and
$\strats = \prod_{\play}\strats_{\play}$ for the game's \emph{strategy space},
the players' payoff functions may be extended to all of $\strats$ by setting
\begin{equation}
\label{eq:pay-mix}
%\(
\pay_{\play}(\strat)
	\defeq \exwrt{\pure\sim\strat}{\pay_{\play}(\pure)}
	= \insum_{\pure\in\pures}  \pay_{\play}(\pure) \, \strat_{\pure}
%\)
%		\prod_{\play\in\players} \strat_{\play\pure_{\play}}
\end{equation}
where, in a slight abuse of notation, we write $\strat_{\pure}$ for the joint probability of playing $\pure\in\pures$ under $\strat$, \ie $\strat_{\pure} = \prod_{\play} \strat_{\play\pure_{\play}}$.
This randomized framework will be referred to as the \emph{mixed extension} of $\fingame$ and we will denote it by $\mixgame$.

For concision, we will also write $(\strat_{\play};\strat_{-\play}) = (\strat_{1},\dotsc,\strat_{\play},\dotsc,\strat_{\nPlayers})$ for the strategy profile where player $\play$ plays $\strat_{\play}\in\strats_{\play}$ against the strategy profile $\strat_{-\play} \in \prod_{\playalt\neq\play} \strats_{\playalt}$ of all other players (and likewise for pure strategies).
In this notation, we also define each player's \emph{mixed payoff vector} as
\begin{equation}
\label{eq:payv-mixed}
\payv_{\play}(\strat)
	= (\pay_{\play}(\pure_{\play};\strat_{-\play}))_{\pure_{\play}\in\pures_{\play}}
\end{equation}
%\ie $\payv_{\play\pure_{\play}}(\strat)$ is simply the payoff that player $\play\in\players$ would have gotten by playing $\pure_{\play}$ against $\strat_{-\play}$.
so the payoff to player $\play\in\players$ under $\strat\in\strats$ becomes
\begin{equation}
\label{eq:pay-lin}
%\(
\pay_{\play}(\strat)
	= \insum_{\pure_{\play}\in\pures_{\play}}
		\pay_{\play}(\pure_{\play};\strat_{-\play}) \, \strat_{\play\pure_{\play}}
	= \braket{\payv_{\play}(\strat)}{\strat_{\play}}.
%\)
\end{equation}
Moving forward, the \emph{best-response correspondence} of player $\play\in\players$ is defined as the set-valued mapping $\brep_{\play}\from\strats\too\strats_{\play}$ given by
\begin{equation}
\label{eq:best}
\brep_{\play}(\strat)
	= \argmax\nolimits_{\stratalt_{\play}\in\strats_{\play}} \pay_{\play}(\stratalt_{\play};\strat_{-\play})
	\quad
	\text{for all $\strat\in\strats$}.
\end{equation}
Extending this over all players, we will write $\brep = \prod_{\play} \brep_{\play}$ for the product correspondence $\brep(\strat) = \brep_{1}(\strat) \times \dotsm \times \brep_{\nPlayers}(\strat)$, and we will say that $\eq\in\strats$ is a \acdef{NE} if $\eq \in \brep(\eq)$.
Equivalently, given that $\pay_{\play}(\stratalt_{\play};\strat_{-\play})$ is linear in $\stratalt_{\play}$, we conclude that $\eq$ is a \acl{NE} if and only if
\begin{equation}
\label{eq:Nash}
\tag{NE}
\pay_{\play}(\eq)
	\geq \pay_{\play}(\pure_{\play};\eq_{-\play})
	\quad
	\text{for all $\pure_{\play}\in\pures_{\play}$ and all $\play\in\players$}.
\end{equation}

As a final point of note, if $\eq$ is a \acl{NE} where each player has a unique best response \textendash\ that is, $\brep_{\play}(\eq) = \{\eq_{\play}\}$ for all $\play\in\players$ \textendash\  we will say that $\eq$ is \emph{strict} because, in this case, $\pay_{\play}(\eq) > \pay_{\play}(\strat_{\play};\eq_{-\play})$ for all $\strat_{\play} \neq \eq_{\play}$, $\play\in\players$.
An immediate consequence of this is that strict equilibria are \emph{pure}, \ie each $\eq_{\play}$ is a pure strategy.
Among \aclp{NE}, strict equilibria are the only ones that are ``structurally robust'' (in the sense that they remain invariant to small perturbations of the underlying game), so they play a particularly important role in game theory.

%----------------------------------------------------------------------
%%% LEARNING
%----------------------------------------------------------------------
\section{Regularized learning in games}
\label{sec:learning}
%----------------------------------------------------------------------
%%% LEARNING
%----------------------------------------------------------------------
% !TEX root = ../Main.tex

Throughout our paper, we will consider iterative decision processes that unfold as follows:
\begin{enumerate}
\item
At each stage $\run=\running$, every participating agent selects an action.
\item
Agents receive a reward determined by their chosen actions and their individual payoff functions.
\item
Based on this reward (or other feedback), the agents update their strategies and the process repeats. 
\end{enumerate}
In this online setting, a crucial requirement is the minimization of the players' \emph{regret}, \ie the difference between a player's cumulative payoff over time and the player's best possible strategy in hindsight.
Formally, if the players' actions at each epoch $\run = \running$ are collectively drawn by the probability distribution $\curr[\corr]\in\simplex(\pures)$, the \emph{regret} of each player $\play\in\players$ is defined as
%\(
\begin{equation}
\reg_{\play}(\nRuns)
	= \max\nolimits_{\pure_{\play}\in\pures_{\play}}
		\insum_{\run=\start}^{\nRuns}
		\bracks{\pay_{\play}(\pure_{\play};\corr_{-\play,\run}) - \pay_{\play}(\curr[\corr])},
\end{equation}
%\)
and we will say that player $\play$ has \emph{no regret} if $\reg_{\play}(\nRuns) = o(\nRuns)$.

One of the most widely used policies to achieve no-regret play is the so-called \acdef{FTRL} family of algorithms and its variants \cite{SSS06,SS11}.
To motivate the analysis to come, we begin with an archetypal \ac{FTRL} method, the \acli{EW} algorithm, also known as \acs{EW} \citep{ACBF02,CBL06,ACBFS95}.

%----------------------------------------------------------------------
%% Gentle start
%----------------------------------------------------------------------
\subsection{A gentle start}

We begin our discussion with a ``stimulus\textendash response'' approach in the spirit of \citet{ER98}:
First, at each stage $\run=\running$, every player $\play\in\players$ employs a mixed strategy $\strat_{\play,\run} \in \strats_{\play}$ to select an action $\pure_{\play,\run} \in \pures_{\play}$.
Subsequently, to measure the performance of their pure strategies over time, each player further maintains a score variable which is updated recursively as
\begin{equation}
\label{eq:score}
\score_{\play\pure_{\play},\run+1}
	= \score_{\play\pure_{\play},\run}
		+ \pay_{\play}(\pure_{\play};\pure_{-\play,\run})
	\quad
	\text{for all $\pure_{\play}\in\pures_{\play}$}.
\end{equation}
In words, $\score_{\play\pure_{\play},\run}$ simply tracks the cumulative payoff of the pure strategy $\pure_{\play}\in\pures_{\play}$ up to time $\run$ (inclusive).%
\footnote{Of course, updating these scores requires the knowledge of the ``what if'' pure payoffs $\pay_{\play}(\pure_{\play};\pure_{-\play,\run})$ at each stage $\run$, but we assume for the moment that this information is available (we will relax this assumption later on).}
As such, this score can be treated as a \emph{propensity} to play a given pure strategy at any given stage:
the strategies $\pure_{\play}\in\pures_{\play}$ with the highest propensity scores $\score_{\play\pure_{\play},\run+1}$ should be played with higher probability at stage $\run+1$.

The most widely used instantiation of this stimulus-response mechanism is the \emph{logit choice} rule
\begin{equation}
\label{eq:logit}
\logit_{\play}(\score_{\play})
	\equiv \frac
		{(\exp(\score_{\play\pure_{\play}}))_{\pure_{\play}\in\pures_{\play}}}
		{\sum_{\pure_{\play}\in\pures_{\play}} \exp(\score_{\play\pure_{\play}})}
\end{equation}
which means that each player selects an action with probability that is exponentially proportional to its score.
In this way, we obtain the \acli{EW} \textendash\ or \acs{EW} \textendash\ algorithm
\acused{EW}
\begin{equation}
\label{eq:EW}
\tag{\acs{EW}}
\score_{\play,\run+1}
	= \score_{\play,\run}
		+ \curr[\step] \payv_{\play}(\curr[\pure])
	\qquad
\strat_{\play,\run+1}
	= \logit_{\play}(\score_{\play,\run+1})
	\qquad
\pure_{\play,\run+1}
	\sim \strat_{\play,\run+1}
\end{equation}
where $\curr[\step]$ is the algorithm's ``learning rate''.
For an introduction to the literature on \eqref{eq:EW}, see \cite{CBL06,AHK12,SS11,BCB12,LS20} and references therein.

The rest of the methods we discuss below will vary some \textendash\ or even all \textendash\ of the components of \eqref{eq:EW}:
the information used to update the players' propensity scores,
the way that propensity scores are mapped to mixed strategies,
and/or
even the way that pure actions are selected.
However, all of the methods under study will be characterized by the same ``stimulus-response'' reinforcement mechanism:
actions that seem to be performing better over time are employed with higher probability, up to some ``regularization'' that incentivizes exploration of underperforming actions.

%----------------------------------------------------------------------
%% Method
%----------------------------------------------------------------------
\subsection{The \acl{method} template}

In the rest of our paper, we will work with an abstract \acdef{method} template which builds on the same stimulus-response principle as \eqref{eq:EW}, while allowing us to simultaneously consider different types of feedback, strategy sampling policies, update structures, etc.
To lighten notation below, we will drop the player index $\play\in\players$ when the meaning can be inferred from the context;
also, to stress the distinction between ``strategy-like'' and ``payoff-like'' variables,
%(such as $\curr$ and $\curr[\dstate]$),
we will write throughout $\scores_{\play} \defeq \R^{\pures_{\play}}$ and $\scores \defeq \prod_{\play} \scores_{\play}$ for the game's ``\emph{payoff space}'', in direct analogy to $\strats_{\play}$ and $\strats = \prod_{\play}\strats_{\play}$ for the game's \emph{strategy space}.

With all this in hand, consider the following general class of \acl{method} methods:
%in an iterative, two-stage fashion as follows:
\begin{flalign}
\label{eq:method}
\tag{RL}
&\quad
\begin{aligned}
&\text{Aggregate payoff information (stimulus):}
	&\quad
\dstate_{\play,\run+1}
	&= \dstate_{\play,\run} + \curr[\step] \signal_{\play,\run}
	\\
&\text{Update choice probabilities (response):}
	&\quad
\state_{\play,\run+1}
	&= \mirror_{\play}(\dstate_{\play,\run+1})
\end{aligned}
&
\end{flalign}
In tune with \eqref{eq:EW}, the various elements of \eqref{eq:method} are defined as follows:
\begin{enumerate}
\item
$\state_{\play,\run} \in \strats_{\play}$ denotes the mixed strategy of player $\play$ at time $\run = \running$
%$\pure_{\play}\in\pures_{\play}$.
\item
$\dstate_{\play,\run} \in \scores_{\play}$ is a ``score vector'' that measures the performance of the player's actions over time.
\item
$\mirror_{\play} \from \scores_{\play} \to \strats_{\play}$ is a ``regularized choice map'' that maps score vectors to choice probabilities.
\item
$\signal_{\play,\run}$ is a surrogate\,/\,approximation of the mixed payoff vector $\payv_{\play}(\curr)$ of player $\play$ at time $\run$.
\item
$\curr[\step] > 0$ is a step-size\,/\,sensitivity parameter of the form $\curr[\step] \propto 1/\run^{\stepexp}$ for some $\stepexp\in[0,1]$.
\end{enumerate}
In words, at each stage of the process, every player $\play\in\players$ observes \textendash\ or otherwise estimates \textendash\ a proxy $\signal_{\play,\run}$ of their individual payoff vector;
subsequently,
players augment their actions' scores based on this information,
they select a mixed strategy via the regularized choice map $\mirror_{\play}$,
and the process repeats.
To streamline our presentation, we discuss in detail the precise definition of $\signal$ and $\mirror$ in \cref{sec:signal,sec:choice} below, and we present a series of examples of \eqref{eq:method} in \cref{sec:examples} right after.
\smallskip

%----------------------------------------------------------------------
%% Payoff signal
%----------------------------------------------------------------------
\subsection{Aggregating payoff information}
\label{sec:signal}

As noted above, the main idea of regularized learning is to track the players' payoff vector $\payv(\curr)$.
Importantly, there are several different modeling choices that can be made here:
players may have direct access to their payoff vectors (in the full information setting),
or some noisy approximation obtained by an inner randomization of the algorithm (\eg when they receive information on their pure actions);
they may have to recreate their payoff vectors altogether (as in the bandit setting),
or
their estimates may be based on a strategy other than the one they actually played (as in the case of optimistic algorithms).

In all cases, we will represent the surrogate payoff vector $\curr[\signal]$ as
\begin{equation}
\label{eq:signal}
\curr[\signal]
	= \payv(\curr) + \curr[\noise] + \curr[\bias]
\end{equation}
where
\begin{equation}
\label{eq:bias-noise}
%\(
\curr[\bias]
	= \exof{\curr[\signal] \given \curr[\filter]} - \payv(\curr)
%\)
%and
	\quad
	\text{and}
	\quad
%\(
\curr[\noise]
	= \curr[\signal] - \exof{\curr[\signal] \given \curr[\filter]}
%\)
\end{equation}
respectively denote the offset and the random error of $\curr[\signal]$ relative to $\payv(\curr)$.
To streamline our presentation, we will also assume that $\norm{\curr[\bias]} = \bigoh(1/\run^{\biasexp})$ and $\norm{\curr[\noise]} = \bigoh(\run^{\noisexp})$ for some $\biasexp,\noisexp \geq 0$;
we discuss the specifics of these bounds later in the paper.

%----------------------------------------------------------------------
%% Choice map
%----------------------------------------------------------------------
\subsection{From scores to strategies}
\label{sec:choice}

Regarding the ``scores-to-strategies'' step of \eqref{eq:method}, we will follow the classical approach of \citet{SS11} and assume that each player is employing a \emph{regularized choice map} of the general form
\begin{equation}
\label{eq:mirror}
\mirror_{\play}(\score_{\play})
	= \argmax\nolimits_{\strat_{\play}\in\strats_{\play}}
		\setof{\braket{\score_{\play}}{\strat_{\play}} - \hreg_{\play}(\strat_{\play})}
	\qquad
	\text{for all $\score_{\play} \in \scores_{\play}$}.
\end{equation}
In the above, the \emph{regularizer} $\hreg_{\play}\from\strats_{\play}\to\R$ acts as a penalty that smooths out the ``hard'' argmax correspondence $\score_{\play} \mapsto \argmax_{\strat_{\play}\in\strats_{\play}} \braket{\score_{\play}}{\strat_{\play}}$.
Accordingly, instead of following the ``leader'' (\ie playing the strategy with the highest propensity score), players follow the ``regularized leader'' \textendash\ that is, they allow for a certain degree of uncertainty in their choice of strategy \cite{SSS06,SS11,BCB12,MS16}.

To ease notation, we will work with kernelized regularizers of the form
\begin{equation}
\label{eq:decomposable}
\txs
%\(
\hreg_{\play}(\strat_{\play})
	= \sum_{\pure_{\play}\in\pures_{\play}} \hker(\strat_{\play\pure_{\play}})
%\)
\end{equation}
for some continuous function $\hker\from[0,1]\to\R$ with $\inf_{z\in(0,1]} \hker''(z) > 0$.
We will also say that the players' regularizers are \emph{steep} if $\lim_{z\to0^{+}} \hker'(z) = -\infty$, and \emph{non-steep} otherwise.

\begin{example}
\label{ex:mirror}
A standard family of kernelized regularizers is given by
\begin{equation}
\label{eq:kernel}
%\(
\hker(z)
%	= [\rho(1-\rho)]^{-1} (z - z^{\rho})
	= z^{\rho} / [\rho(\rho-1)]
	\quad
	\text{for $\rho\in(0,2]$},
\end{equation}
%\)
for $\rho\in(0,1)\cup(1,2]$ and $\hker(z) = z \log z$ for $\rho = 1$ \cite{BCB12,MS16,LS20,ZS21}.
This family includes:
\begin{itemize}
\item
For $\rho = 2$, \cref{eq:kernel} boils down to the quadratic regularizer $\hker(z) = z^{2}/2$, which in turn yields the Euclidean projection map
\begin{equation}
\label{eq:Eucl}
\mirror_{\play}(\score_{\play})
	= \Eucl_{\strats_{\play}}(\score_{\play})
	\equiv \argmin\nolimits_{\strat_{\play}\in\strats_{\play}} \twonorm{\score_{\play} - \strat_{\play}}.
\end{equation}
\item
For $\rho = 1$, \cref{eq:kernel} yields the \emph{entropic regularizer} $\hker(z) = z\log z$, which in turn leads the logit choice map \eqref{eq:logit}.
%\begin{equation}
%\label{eq:logit}
%\mirror_{\play}(\score_{\play})
%	= \logit_{\play}(\score_{\play})
%	\equiv \frac
%		{(\exp(\score_{\play\pure_{\play}}))_{\pure_{\play}\in\pures_{\play}}}
%		{\sum_{\pure_{\play}\in\pures_{\play}} \exp(\score_{\play\pure_{\play}})}
%\end{equation}
%and which underlies \cref{alg:EW,alg:OMW,alg:EXP3} below.
\item
For $\rho = 1/2$, we obtain the \emph{fractional power regularizer} $\hker(z) = - 4\sqrt{z}$ that underlies the \acs{Tsallis} algorithm of \cite{ALT15,ZS21} (see also \cref{sec:examples} below).
\endenv
\end{itemize}
\end{example}

%----------------------------------------------------------------------
%% Examples
%----------------------------------------------------------------------
\subsection{Specific algorithms}
\label{sec:examples}

We now proceed to discuss some archetypal examples of \eqref{eq:method}.
%readers that are already familiar with these algorithms can skip ahead to \cref{sec:resilience}.
%[For simplicity, player indices are suppressed when possible.]
\smallskip

%----------------------------------------------------------------------
\begin{algo}[Follow the regularized leader]
\label{alg:FTRL}
The standard \acdef{FTRL} method of \citet{SSS06} is obtained when players observe their full payoff vectors, that is, $\signal_{\play,\run} = \payv_{\play}(\curr)$.
In this case, \eqref{eq:method} boils down to the deterministic update rule
\begin{align}
\dstate_{\play,\run+1}
	&= \dstate_{\play,\run} + \curr[\step] \vecfield_{\play}(\curr)
	\qquad
\state_{\play,\run+1}
	= \mirror_{\play}(\dstate_{\play,\run+1})
	\notag
\shortintertext{or, more explicitly}
\label{eq:FTRL}
\tag{\acs{FTRL}}
\state_{\play,\run+1}
%	= \mirror_{\play}(\dstate_{\play,\run+1})
	&= \argmax\nolimits_{\strat_{\play}\in\strats_{\play}}
		\braces*{\sum\nolimits_{\runalt=\start}^{\run} \iter[\step] \pay_{\play}(\strat_{\play};\state_{-\play,\runalt}) - \hreg_{\play}(\strat_{\play})}
%	= \argmax_{\strat_{\play}\in\strats_{\play}}
%		\sum_{\runalt=\start}^{\run} \iter[\step]
%\next[\dstate]
%	= \curr[\dstate] + \curr[\step] \payv(\curr)
%	\qquad
%\next
%	= \mirror(\next[\dstate])
\end{align}
For a detailed discussion of \eqref{eq:FTRL}, see \cite{SS11,BCB12,LS20}.
We only note here that, as a special case, when \eqref{eq:FTRL} is run with the logit choice setup of  \cref{eq:logit}, a standard calculation yields the \acli{EW} \eqref{eq:EW} \citep{Vov90,LW94,ACBFS95,SS11}.
%\acli{EW} algorithm \textendash\ or \acs{EW} \citep{Vov90,LW94,ACBFS95} \textendash\ namely
%\acused{EW}
%\begin{equation}
%\label{eq:EW}
%\tag{\acs{EW}}
%\state_{\play\pure_{\play},\run+1}
%	= \frac
%		{\state_{\play\pure_{\play},\run} \exp(\curr[\step] \pay_{\play}(\pure_{\play};\state_{-\play,\run}))}
%		{\sum_{\purealt_{\play} \in \pures_{\play}} \state_{\play\purealt_{\play},\run} \exp(\curr[\step] \pay_{\play}(\purealt_{\play};\state_{-\play,\run}))}
%%\dstate_{\play,\run+1}
%%	= \dstate_{\play,\run} + \curr[\step] \payv_{\play}(\curr)
%%	\qquad
%%\state_{\play,\run+1}
%%	= \logit_{\play}(\dstate_{\play,\run+1})
%\end{equation}
%For an appetizer to the literature on \eqref{eq:EW}, see \cite{CBL06,AHK12,SS11,BCB12,LS20} and references therein.
\endenv
\end{algo}
%----------------------------------------------------------------------

%----------------------------------------------------------------------
\smallskip
\begin{algo}[\Acl{OFTRL}]
\label{alg:OFTRL}
A notable variant of \ac{FTRL} \textendash\ originally due to \citet{Pop80} and subsequently popularized by \citet{RS13-COLT,RS13-NIPS} \textendash\ is the so-called \acli{OFTRL} method.
This scheme employs an ``optimistic'' correction intended to anticipate future steps, and it updates as
\begin{equation}
\label{eq:OFTRL}
\tag{\acs{OFTRL}}
\dstate_{\play,\run+1}
	= \dstate_{\play,\run}
		+ \curr[\step] \bracks{2\payv_{\play}(\curr) - \payv_{\play}(\prev)}
%	= \curr[\dstate]
%		+ \curr[\step] \payv(\curr)
%		+ \curr[\step] \bracks{\payv(\curr) - \payv(\prev)}
%	\quad
%\state_{\play,\run+1}
%	= \mirror_{\play}(\dstate_{\play,\run+1})
\end{equation}
with $\state_{\play,\run} = \mirror_{\play}(\dstate_{\play,\run})$.
As a special case, if \eqref{eq:OFTRL} is run with the logit choice map \eqref{eq:logit}, we obtain the familiar update rule known as \acdef{OMW} \citep{RS13-COLT,RS13-NIPS,SALS15,DP19}.

Compared to \eqref{eq:FTRL}, the gain vector $\curr[\signal] = 2\payv(\curr) - \payv(\prev)$ of \eqref{eq:OFTRL} has offset $\curr[\bias] = \payv(\curr) - \payv(\prev)$ relative to $\payv(\curr)$.
Thus, even though \eqref{eq:OFTRL} assumes full access to the players' mixed payoff vectors, it uses this information differently than \eqref{eq:FTRL}:
in particular, the offset of \eqref{eq:OFTRL} is non-zero \emph{by design}, not because of some systematic error in the payoff measurement process.
%Clearly, \eqref{eq:OFTRL} can be recast as an instance of \eqref{eq:method} by letting $\curr[\signal] = 2\payv(\curr) - \payv(\prev)$, which yields in turn $\curr[\bbound] = \bigoh(\curr[\step])$ and $\curr[\noise] = 0$.
\endenv
\end{algo}

Now, up to this point, we have not detailed how players might observe their full, mixed payoff vectors.
This assumption simplifies the analysis immensely, but it is not realistic in applications to \eg online advertising and network science, where players may only be able to observe their realized payoffs, and have no information about the strategies of other players or actions they did not play.
On that account, we describe below a range of \emph{payoff-based} policies where players estimate their counterfactual, ``what-if'' payoffs \emph{indirectly}.

The most common way to achieve this is via the \acli{IWE}
%defined in components as
\begin{equation}
\label{eq:IWE}
\tag{IWE}
\IWE_{\play\pure_{\play}}(\strat)
	= \frac
		{\oneof{\purequery_{\play} = \pure_{\play}}}
		{\strat_{\play\pure_{\play}}}
			\pay_{\play}(\purequery)
%			\pay_{\play}(\purequery_{\play};\purequery_{-\play})
	\quad
	\text{for all $\pure_{\play}\in\pures_{\play}$, $\play\in\players$},
\end{equation}
where $\strat\in\strats$ is the players' strategy profile, and $\purequery\in\pures$ is drawn according to $\strat$.
This estimator is at the heart of the online learning literature \cite{CBL06,SS11,BCB12,LS20} and it leads to the following methods:

%----------------------------------------------------------------------
\smallskip
\begin{algo}[\Acl{BFTRL}]
\label{alg:BFTRL}
Plugging \eqref{eq:IWE} directly into \eqref{eq:method} yields the \acli{BFTRL} policy
\begin{equation}
\label{eq:BFTRL}
\tag{\acs{BFTRL}}
%\state_{\play,\run}
%	= \logit_{\play}(\dstate_{\play,\run})
%	\quad
\dstate_{\play,\run+1}
	= \dstate_{\play,\run} + \curr[\step] \IWE_{\play}(\curr[\query])
	\qquad
\state_{\play,\run+1}
	= \mirror_{\play}(\dstate_{\play,\run+1})	
\end{equation}
where \eqref{eq:IWE} is sampled at the mixed strategy profile
\begin{equation}
\label{eq:explore}
\query_{\play,\run}
	= (1-\curr[\mix]) \state_{\play,\run} + \curr[\mix] \unif_{\pures_{\play}}
\end{equation}
for some ``explicit exploration'' parameter $\curr[\mix] \propto 1/\run^{\mixexp}$, $\mixexp>0$, which specifies the mix between $\state_{\play,\run}$ and the uniform distribution $\unif_{\pures_{\play}}$ on $\pures_{\play}$.
As we discuss in the sequel, this combination of \eqref{eq:IWE} with the explicit exploration mechanism \eqref{eq:explore} means that the surrogate payoff vector $\curr[\signal] = \IWE(\curr[\query])$ used to update \eqref{eq:BFTRL} has offset and noise bounded respectively as $\curr[\bias] = \bigoh(\curr[\mix])$ and $\curr[\noise] = \bigoh(1/\curr[\mix])$.

Two special cases of \eqref{eq:BFTRL} that have attracted significant attention in the literature are:
\begin{enumerate}
\item
The \acdef{EXP3} \cite{ACBFS02,CBL06,LS20}, obtained by running \eqref{eq:BFTRL} with the logit choice map \eqref{eq:logit}.
\item
The \acdef{Tsallis} \cite{ALT15,AB10,ZS19,ZS21} that was proposed as a more efficient alternative to \ac{EXP3}, and which updates as
\begin{equation}
\label{eq:Tsallis}
\tag{\acs{Tsallis}}
%\dstate_{\play,\run+1}
%	=\dstate_{\play,\run} + \curr[\step] \IWE_{\play}(\curr[\query];\curr[\purequery])
%	\quad
\hspace{-.4em}
\state_{\play,\run}
	= \argmax\nolimits_{\strat_{\play}\in\strats_{\play}}
		\braces*{
			\braket{\dstate_{\play,\run}}{\strat_{\play}} 
			+ 4 \sum\nolimits_{\pure_{\play}\in\pures_{\play}} \sqrt{\strat_{\play\pure_{\play}}}
			} %end braces
\end{equation}
\ie as \eqref{eq:BFTRL} with the fractional power regularizer $\hker(z) = -4\sqrt{z}$ of \cref{ex:mirror}.
\endenv
%is the family of \acdefp{INF} which employ the same payoff estimation scheme as \eqref{eq:EXP3}, but with a different choice map.
%Focusing for concreteness on the ``Tsallis $1/2$'' setup of \cite{ALT15,ZS19,ZS21}, we obtain the \acs{Tsallis} algorithm
\end{enumerate}
\end{algo}
%----------------------------------------------------------------------

%Of course, there are numerous variants of the above, and the range of algorithms covered by \eqref{eq:method} cannot be exhausted here.
For illustration purposes, we provide some more examples of \eqref{eq:method} in \cref{app:algorithms}.

%----------------------------------------------------------------------
%%% RESILIENCE
%----------------------------------------------------------------------
\section{First results: resilience to strategic deviations}
\label{sec:resilience}
%----------------------------------------------------------------------
%%% RESILIENCE
%----------------------------------------------------------------------
% !TEX root = ../Main.tex

We are now in a position to begin our analysis of the rationality properties of the players' long-run behavior under \eqref{eq:method}.
To that end, we should first note that no-regret play may \emph{still} lead to counterintuitive and highly non-rationalizable outcomes, \eg with all players selecting dominated strategies for all time.
The example below is adapted from \citet{VZ13}.

\begin{example}
\label{ex:VZ}
Consider the $4\times4$ symmetric $2$-player game with payoff bimatrix
\begin{equation*}
\begin{array}{l|cccc}
	&A	&B	&C	&D	\\
\hline
A	&(1,1)		&(1,2/3)		&(0,0)		&(0,-1/3)	\\
B	&(2/3,1)	&(2/3,2/3)		&(-1/3,0)	&(-1/3,-1/3)	\\
C	&(0,0)		&(0,-1/3)		&(1,1)		&(1,2/3)	\\
D	&(-1/3,0)	&(-1/3,-1/3)	&(2/3,1)	&(2/3,2/3)
\end{array}
%\begin{array}{cccc}
%	(1,1)	&(1,1)	&(0,0)	&(0,0)	\\
%	(2/3,2/3)	&(2/3,2/3)	&(-1/3,-1/3)	&(-1/3,-1/3)	\\
%	(0,0)	&(0,0)	&(1,1)	&(1,1)	\\
%	(-1/3,-1/3)	&(-1/3,-1/3)	&(2/3,2/3)	&(2/3,2/3)	\\
%\end{array}
\end{equation*}
In this game, $B$ and $D$ are strictly dominated for both players by their stronger ``twins'' ($A$ and $C$ respectively).
However, it is easy to check that if both players choose between $(B,B)$ and $(D,D)$ with probability $1/2$ each, the resulting distribution of play $\corr \in \simplex(\pures)$ satisfies $\pay_{\play}(\pure_{\play};\corr_{-\play}) - \pay_{\play}(\corr) \leq -1/6$ for all $\pure_{\play} \in \setof{A,B,C,D}$, $\play = 1,2$.
As a result, the players' regret under $\curr[\corr] \equiv \corr$ is \emph{negative}, even though both players play strictly dominated strategies at all times.
\endenv
\end{example}

The example above shows unequivocally that
\begin{quote}
\centering
\itshape
No-regret play does not suffice to exclude non-rationalizable outcomes.
\end{quote}
In addition, \cref{ex:VZ} also shows that predictions based on correlated play are not always appropriate for describing the players' behavior under \eqref{eq:method}:
the end-state of any regularized learning algorithm will be a closed connected set of mixed strategies, so it is not possible to play \emph{only} $(B,B)$ or $(D,D)$ in the long run.
We are thus led to the following natural questions:
\begin{quote}
\centering
%\begingroup
\itshape
What are the rationality properties of long-run play under \eqref{eq:method}?\\
Is the players' behavior robust to strategic deviations?
%\endgroup
\end{quote}

To study these questions formally, we will focus on the \emph{limit set} $\limset(\state)$ of $\curr$ under \eqref{eq:method}, \viz 
\begin{equation}
\label{eq:limset}
\limset(\state)
	\defeq \intersect\nolimits_{\run} \cl\setdef{\iter}{\runalt \geq \run}
	\equiv \setdef{\limpoint\in\strats}{\state_{\run_{k}}\to\limpoint\ \text{for some subsequence $\state_{\run_{k}}$ of $\curr$}}.
\end{equation}
In words, $\limset(\state)$ is the set of limit points of $\curr$ or, equivalently, the \emph{smallest} subset of $\strats$ to which $\curr$ converges.
Clearly, the simplest instance of a limit set is when $\limset(\state)$ is a singleton, \ie when $\curr$ converges to a point.
This case has attracted significant interest in the literature:
for example, if $\limset(\state) = \{\eq\}$ then, for certain special cases of \eqref{eq:method}, it is known that $\eq$ is a \acl{NE} of $\fingame$ \citep{MZ19}.
However, beyond this relatively simple regime, the structure of the limit sets of \eqref{eq:method} could be arbitrarily complicated and their rationality properties are not well-understood.

With this in mind, as a first attempt to study whether the long-run behavior of \eqref{eq:method} is ``robust to strategic deviations'', we will consider the following notion of \emph{resilience to strategic deviations}:

\begin{definition}
\label{def:resilience}
A closed subset $\set$ of $\strats$ is said to be \emph{resilient to strategic deviations} \textendash\ or simply \emph{resilient} \textendash\ if,
for every deviation $\strat_{\play}\in\strats_{\play}$ of every player $\play\in\players$, we have
\begin{equation}
\label{eq:resilience}
\pay_{\play}(\eq)
	\geq \pay_{\play}(\strat_{\play},\eq_{-\play})
	\quad
	\text{for some $\eq\in\set$}.
\end{equation}
\end{definition}

Informally, $\set$ is resilient if every unilateral deviation from $\set$ is deterred by some (possibly different) element thereof.
In particular, if $\set$ is a singleton, we immediately recover the definition of a \acl{NE};
beyond this case however, other examples include
the set of undominated strategies of a game,
the support face of the equilibria of two-player zero-sum games,
etc.
Importantly, as we show below, the limit sets of \eqref{eq:method} are resilient \emph{in all games:}

\begin{restatable}{theorem}{RESILIENCE}
\label{thm:resilience}
Let $\curr$, $\run=\running$, be the sequence of play generated by \eqref{eq:method} with step-size\,/\,gain parameters
%$\curr[\bbound] \to 0$
%and
%$\curr[\step]\curr[\sbound]^{2} \to 0$.
$\stepexp > 2\noisexp$ and $\biasexp > 0$.
%\eqref{eq:params-apt}.
%\begin{equation}
%\label{eq:params-res}
%\tag{H2}
%%\curr[\step]
%%	\to 0,
%%	\;\;
%\curr[\bbound]
%	\to 0
%	\quad
%	\text{and}
%	\quad
%\curr[\step]\curr[\sbound]^{2}
%	\to 0.
%\end{equation}
Then, \acl{wp1}, the limit set $\limset(\state)$ of $\curr$ is resilient.
\end{restatable}

\begin{corollary}
\label{cor:Nash}
With assumptions as above, if $\limset(\state) = \{\eq\}$, $\eq$ is a \acl{NE} \acs{wp1}.
\end{corollary}

\begin{proof}[Proof sketch]
The proof of \cref{thm:resilience} boils down to two interleaved arguments that we detail in \cref{app:resilience}.
The first hinges on showing that, if $\probof{\limset(\state) = \set} > 0$ for some \emph{non-random} $\set\subseteq\strats$, $\set$ must be resilient.
This is argued by contradiction:
if $\base_{\play}\in\strats_{\play}$ is a unilateral deviation violating \cref{def:resilience}, we must also have $\liminf_{\run\to\infty} \bracks{\pay_{\play}(\base_{\play};\state_{-\play,\run}) - \pay_{\play}(\curr)} > 0$ with positive probability.
However, the existence of a strategy that consistently outperforms $\curr$ runs contrary to
the fact that
%the reinforcement mechanism underlying \eqref{eq:method} \textendash\ namely that
strategies that \eqref{eq:method} selects against underperforming strategies.
%perform well should be eventually preferred over strategies that do not.
We make this intuition precise via an energy argument that leverages a series of results from martingale limit theory (which is where the requirements for $\curr[\step]$, $\curr[\bias]$ and $\curr[\noise]$ come in).
Then, to get the stronger statement that the \emph{random} set $\limset(\state)$ is resilient \acs{wp1}, we show that the above remains true if $\base_{\play}$ is replaced by a deviation $\basealt_{\play}$ which is close enough to $\base_{\play}$ and has \emph{rational} entries.
Since there is a countable number of such profiles, we can use a union bound on an enumeration of the rationals to isolate a deviation witnessing the negation of \cref{def:resilience};
our claim then follows by applying our argument for non-random sets.
\end{proof}

\Cref{thm:resilience} is our first universal guarantee for \eqref{eq:method}, so some remarks are in order.
First, we should point out that the requirements
$\biasexp > 0$ and $2\noisexp < \stepexp$
%$\curr[\bbound]\to0$ and $\curr[\step]\curr[\sbound]^{2} \to 0$
are a priori \emph{implicit} because they depend on the offset and magnitude statistics of the feedback sequence $\curr[\signal]$.
However, in most learning algorithms, these quantities are under the \emph{explicit} control of the players:
for example, as we show in \cref{app:algorithms}, \cref{alg:OFTRL} has
$\biasexp = \stepexp$
%$\curr[\bbound] = \bigoh(\curr[\step])$
while, for \cref{alg:BFTRL}, we have
$\biasexp = \noisexp = \mixexp$.
%$\curr[\bbound] = \bigoh(\curr[\mix])$ and $\curr[\sbound] = \bigoh(1/\curr[\mix])$.
In this way, when instantiated to \crefrange{alg:FTRL}{alg:BFTRL} (and special cases thereof),
\cref{thm:resilience} yields the following corollary:

\begin{corollary}
\label{cor:resilience}
Suppose that \crefrange{alg:FTRL}{alg:BFTRL} are run with
$\stepexp\in(0,1]$ and, for \cref{alg:BFTRL}, $\mixexp\in(0,\stepexp/2)$.
%a step-size sequence of the form $\curr[\step] = \step/\run^{\stepexp}$, $\stepexp\in(0,1]$, and, where applicable, a sampling parameter $\curr[\mix] = \mix/\run^{\mixexp}$ with $0 < 2\mixexp < \stepexp$.
Then, \acl{wp1}, the limit set $\limset(\state)$ of $\curr$ is resilient.
\end{corollary}

Now, since \cref{thm:resilience} applies to all games, it would seem to provide a universally positive answer to whether \eqref{eq:method} is robsut to strategic deviations.
However, this is not so:
a direct calculation shows that the face of $\strats$ that is spanned by the dominated strategies $(B,B)$ and $(D,D)$ of \cref{ex:VZ} \emph{is} resilient, so \cref{thm:resilience} cannot exclude convergence to a set where dominated strategies survive.
Thus, just like no-regret play, the notion of resilience does not suffice by itself to capture the idea of rational behavior.
This is because, albeit natural, resilience is too lax to provide a meaningful link between robustness to unilateral deviations \textendash\ a \emph{game-theoretic} requirement \textendash\ and stability under regularized learning \textendash\ a \emph{dynamic} requirement.
We address this question in detail in the next section.

%----------------------------------------------------------------------
%%% CLOSED
%----------------------------------------------------------------------
\section{A characterization of strategic stability under regularized learning}
\label{sec:closed}
%----------------------------------------------------------------------
%%% SHARP
%----------------------------------------------------------------------
% !TEX root = ../Main.tex

%\VBedit{Just like CCE, resilience is too weak.}
%%In brief, 
Similar to the set of pure strategies that arise from no-regret play, the main limitation of resilience is that a payoff-improving deviation may be countered by an action profile where the deviator also switched to a \emph{different} strategy;
in other words, resilience is not a \emph{self-enforcing} barrier to deviations.
In view of this, we will focus below on a much more stringent criterion of strategic stability, namely that \emph{any} deviation from the set in question incurs a cost to the deviating agent.

%----------------------------------------------------------------------
%% Analysis
%----------------------------------------------------------------------
% \subsection*{Formal analysis}
\subsection*{Club sets}

To make all this precise, define the \emph{better-reply correspondence} of player $\play\in\players$ as
\begin{equation}
\label{eq:better}
%\(
\btr_{\play}(\strat)
	= \setdef
		{\stratalt_{\play}\in\strats_{\play}}
		{\pay_{\play}(\stratalt_{\play};\strat_{-\play}) \geq \pay_{\play}(\strat)}
%		{\pure_{\play}\in\pures_{\play}}
%		{\pay_{\play}(\pure_{\play};\strat_{-\play}) \geq \pay_{\play}(\strat)}
%\),
\end{equation}
and write $\btr = \prod_{\play} \btr_{\play}$ for the product correspondence $\btr(\strat) = \btr_{1}(\strat) \times \dotsm \times \btr_{\nPlayers}(\strat)$.
[In words, $\btr_{\play}$ assigns to each $\strat\in\strats$ those strategies of player $\play$ that are (weakly) better against $\strat$ than $\strat_{\play}$.]
In addition, given a product of pure strategies $\pureset = \prod_{\play\in\players} \pureset_{\play}$ with $\pureset_{\play} \subseteq \pures_{\play}$ for all $\play\in\players$,
let $\set = \simplex(\pureset)$ denote the span of $\pureset$,
and
let $\faces$ denote the collection of all such sets.
%We then say that $\pureset\in\prodset$ is \acli{closed}\acused{closed} \textendash\ a \emph{\acs{closed} set} for short \textendash\ if it is closed under $\btr$, \ie $\btr(\pureset) \subseteq \pureset$.
We then say that $\set\in\faces$ is \acli{closed}\acused{closed} \textendash\ a \emph{\acs{closed} set} for short \textendash\ if it is closed under $\btr$, \ie $\btr(\set) \subseteq \set$;
finally, $\set$ is said to be \acdef{minimal} if it does not admit a proper \ac{closed} subset.
\footnote{Analogously to \ac{closed} sets, $\set\in\faces$ is said to be \acdef{curb} if it is closed under \emph{best replies}, \ie $\brep(\set) \subseteq \set$ \cite{BW91}.
%(and likewise for \emph{minimal} \ac{curb} sets) \cite{BW91}.
Clearly, \ac{closed} sets are also \ac{curb}, but the converse does not hold, \cf \cite{RW95}.
}

Of course, the entire strategy space $\strats$ is \acl{closed} so, a priori, \ac{closed} sets could also contain dominated strategies and\,/\,or other non-rationalizable outcomes.
By contrast, \emph{minimal} \ac{closed} sets are much more rigid in their relation to rational behavior because any unilateral deviation from an \ac{minimal} set is \emph{costly}, and \ac{minimal} sets are \emph{minimal} in this regard.
On that account, \ac{minimal} sets can be seen as \emph{the closest setwise analogue to strict \aclp{NE}.}

This analogy is accentuated further by the following properties of \ac{minimal} sets, all due to \citet{RW95}, who introduced the concept:
\begin{enumerate}
%[(\itshape i\hspace*{.5pt}\upshape)]
\item
Every game admits an \ac{minimal} set;
and if this set is a singleton, then it is a \emph{strict} \acl{NE}.
%\item
%If an \ac{minimal} set is a singleton, it is a \emph{strict} \acl{NE} of the underlying game.
\item
Any \ac{minimal} set $\set$ is \emph{fixed} under better replies, that is, $\btr(\set) = \set$ (implying in turn that $\set$ cannot contain any dominated strategies, including iteratively dominated ones).
\item
Any \ac{minimal} set $\set$ contains an \emph{essential equilibrium component}, \ie a component of \aclp{NE} such that every small perturbation of the game admits a nearby equilibrium;
%\footnote{Formally,
%%recall that a $\delta$-perturbation of $\fingame = \fingamefull$ is a game $\alt\fingame = (\players,\pures,\alt\pay)$ whose payoff functions $\alt\pay$ are $\delta$-close to $\pay$.
%%Then, a closed component $\eqs$ of \aclp{NE} of $\fingame$ is called \emph{essential} if, for all $\eps>0$, there exists some $\delta>0$ such that any $\delta$-perturbation of $\fingame$ admits a \acl{NE} that is $\eps$-close to $\eqs$ \citep{vD87}.
%a \acl{NE} component $\eqs$ of $\fingame$ is called \emph{essential} if, for all $\eps>0$, there exists $\delta>0$ such that any perturbation of the payoffs of $\fingame$ by at most $\delta$ produces a \acl{NE} that is $\eps$-close to $\eqs$ \citep{vD87}.
%}
in addition, this component has \emph{full support} on $\set$, \ie it employs all pure strategy profiles that lie in $\set$.%
\footnote{Formally,
%recall that a $\delta$-perturbation of $\fingame = \fingamefull$ is a game $\alt\fingame = (\players,\pures,\alt\pay)$ whose payoff functions $\alt\pay$ are $\delta$-close to $\pay$.
%Then, a closed component $\eqs$ of \aclp{NE} of $\fingame$ is called \emph{essential} if, for all $\eps>0$, there exists some $\delta>0$ such that any $\delta$-perturbation of $\fingame$ admits a \acl{NE} that is $\eps$-close to $\eqs$ \citep{vD87}.
a component $\eqs$ of \aclp{NE} of $\fingame$ is \emph{essential} if, for all $\eps>0$, there exists $\delta>0$ such that any perturbation of the payoffs of $\fingame$ by at most $\delta$ produces a \acl{NE} that is $\eps$-close to $\eqs$ \citep{vD87}.
This property \textendash\ known as ``\emph{essentiality}'' \textendash\ has a long history as one of the strictest setwise solution refinements in game theory;
in particular, it satisfies all the seminal \emph{strategic stability} requirements of \citet{KM86}, including robustness to strategic payoff perturbations.
%\footnote{At a high level, this means that any game $\alt\fingame$ obtained by replacing the pure payoffs of $\fingame$ with those of a nearby fully mixed strategy in $\fingame$ admits a nearby \acl{NE};
%for a complete discussion, see \cite{KM86}.}
For an in-depth discussion, see \citet{vD87}.}
\end{enumerate}
\smallskip

Going back to our online learning setting, the above leads to the following natural set of questions:
\vspace{-\topsep}
\begin{center}
\emph{Are \ac{closed} sets \textpar{minimal or not} stable under the dynamics of regularized learning?\\
Are they attracting?
And, if so, are they the only such sets?}
\end{center}
\vspace{-\parsep}
Any answer to these questions \textendash\ positive or negative \textendash\ would be an important step in delineating the relation between \emph{strategic stability} (in the above sense) and \emph{dynamic stability} under \eqref{eq:method}.
To that end, we start by formalizing some notions of dynamic stability that will be central in the sequel:

\begin{definition}
\label{def:stable}
Fix some subset $\set$ of $\strats$ and a tolerance level $\thres > 0$.
We then say that $\set$ is:
\begin{enumerate}
\item
\emph{Stochastically stable}
if, for every neighborhood $\nhd$ of $\set$ in $\strats$, there exists a neighborhood $\init[\nhd]$ of $\set$ such that
\begin{equation}
\probof{\curr \in \nhd \; \text{for all $\run=\running$}}
	\geq 1 - \thres
	\quad
	\text{whenever $\init \in \init[\nhd]$}.
\end{equation}
\item
\emph{Stochastically attracting}
if there exists a neighborhood $\init[\nhd]$ of $\set$ such that
\begin{equation}
\probof{\lim\nolimits_{\run\to\infty} \dist(\curr,\set) = 0}
	\geq 1 - \thres
	\quad
	\text{whenever $\init \in \init[\nhd]$}.
\end{equation}
\item
\emph{Stochastically asymptotically stable}
if it is stochastically stable and attracting.
\item
\emph{Irreducibly stable}
if $\set$ is stochastically asymptotically stable and it does not admit a strictly smaller stochastically asymptotically subset $\alt\set$ with $\supp(\alt\set) \subsetneq \supp(\set)$.
\end{enumerate}
\end{definition}

With all this in hand, our main result below provides a sharp characterization of strategic stability in the context of regularized learning:

\begin{restatable}{theorem}{CLOSED}
\label{thm:closed}
Fix some set $\set\in\faces$ and suppose that \eqref{eq:method} is run with
a steep regularizer
%$\curr[\step] = \bigoh(1/\run^{\stepexp})$, $\curr[\bbound] = \bigoh(1/\run^{\biasexp})$, and $\curr[\sdev] = \bigoh(\run^{\noisexp})$ for some
and step-size\,/\,gain parameters $\stepexp\in[0,1]$, $\biasexp>0$, and $\noisexp < 1/2$.
Then:
\begin{enumerate}
\item
$\set$ is stochastically asymptotically stable under \eqref{eq:method} if and only if it is a \ac{closed} set.
\item
$\set$ is irreducibly stable under \eqref{eq:method} if and only if it is an \ac{minimal} set.
\end{enumerate}
\end{restatable}

In addition, we also get the following convergence rate estimates for \ac{closed} sets:

\begin{restatable}{theorem}{RATE}
\label{thm:rate}
Let $\set \in \faces$ be a \ac{closed} set, and let $\curr$, $\run=\running$, be the sequence of play generated by \eqref{eq:method} with parameters $\stepexp\in[0,1]$, $\biasexp>0$, and $\noisexp < 1/2$.
Then, for all $\thres>0$, there exists an \textpar{open, unbounded} initialization domain $\dbasin \subseteq \scores$ such that, with probability at least $1-\thres$, we have
\begin{equation}
\label{eq:rate}
\txs
\dist(\curr,\set)
%	\defeq \min\nolimits_{\strat\in\set} \onenorm{\curr - \strat}
	\leq \Const \rate\parens*{\const_{1} - \const_{2} \sum_{\runalt=\start}^{\run} \iter[\step]}
	\quad
	\text{whenever $\init[\dstate]\in\dbasin$}
\end{equation}
where
%$\dist(\curr,\set) = \min\nolimits_{\stratalt\in\set} \onenorm{\curr - \strat}$ denotes the \textpar{Hausdorff} $L^{1}$-distance between $\curr$ and $\set$,
$\Const,\const_{1},\const_{2}$ are constants ($\Const,\const_{2} > 0$),
%$\curr[\efftime] = \sum_{\run=\start}^{\run} \iter[\step]$,
and
the rate function $\rate$ is given by
$\rate(z) = (\hker')^{-1}(z)$ if $z > \lim_{z\to0^{+}} \hker'(z)$,
and
$\rate(z) = 0$ otherwise.
\end{restatable}

Specifically, if we instantiate \cref{thm:rate} to \crefrange{alg:FTRL}{alg:BFTRL}, we get the explicit estimates:

\begin{restatable}{corollary}{RATECOR}
\label{cor:rate}
%Let $\set\in\faces$ be a \ac{closed} set and
Suppose that \crefrange{alg:FTRL}{alg:BFTRL} are run with
$\stepexp\in[0,1]$ and, for \cref{alg:BFTRL}, $\mixexp \in (0,1/2)$.
Then, with notation as in \cref{thm:rate}, $\curr$ converges to $\set$ at a rate of
\begin{equation}
\label{eq:rate-explicit}
\dist(\curr,\set)
	\leq \Const \cdot
	\left\{\hspace{-.5em}
	\begin{array}{lll}
	\pospart{1 - \const \sum_{\runalt=1}^{\run} \iter[\step]}
		&\quad
		\text{if $\hker(z) = z^{2}/2$}
		&\quad
		\text{\#\,Euclid. proj.}
	\\[\smallskipamount]
	\exp\parens[\big]{-\const \sum_{\runalt=1}^{\run} \iter[\step]}
		&\quad
		\text{if $\hker(z) = z\log z$}
		&\quad
		\text{\#\,logit choice}
	\\[\smallskipamount]
	1\big/\parens[\big]{\const + \sum_{\runalt=\start}^{\run} \iter[\step]}^{2}
		&\quad
		\text{if $\hker(z) = -4\sqrt{z}$}
		&\quad
		\text{\#\,Tsallis maps}
	\end{array}
	\right.
\end{equation}
for positive constants $\Const,\const > 0$.
In particular, the projection-based variants of \crefrange{alg:FTRL}{alg:BFTRL} converge to \ac{minimal} sets in a \textbf{finite} number of steps.
\end{restatable}

%----------------------------------------------------------------------
%% Limset figure begins here

\begin{figure*}[t]
\centering
\footnotesize
%\begin{subfigure}{.24\textwidth}
%\includegraphics[width=\textwidth]{Figures/HetCycle.pdf}
%\end{subfigure}
\begin{subfigure}{.23\textwidth}
\includegraphics[width=\textwidth]{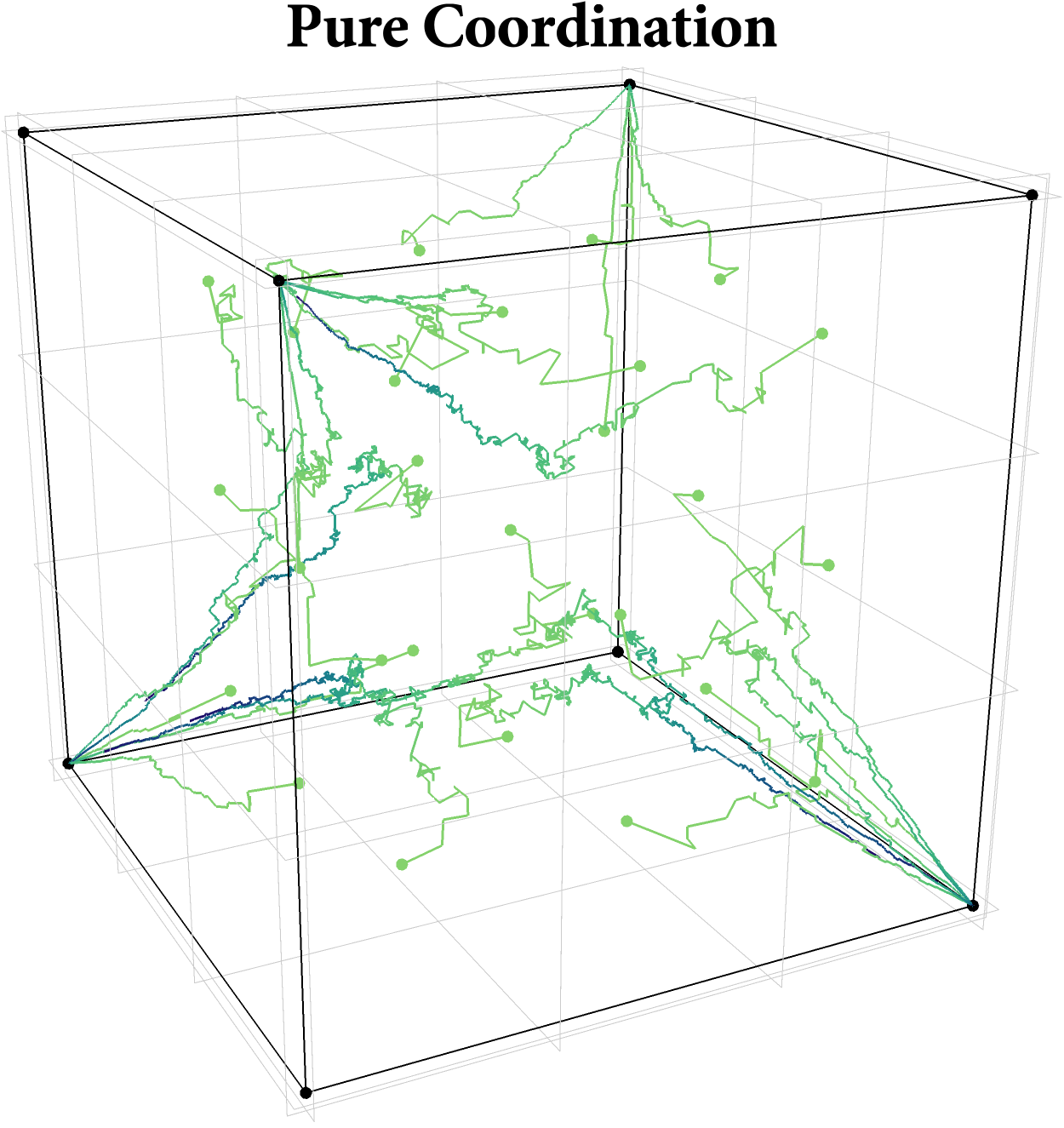}
\end{subfigure}
\hfill
\begin{subfigure}{.23\textwidth}
\includegraphics[width=\textwidth]{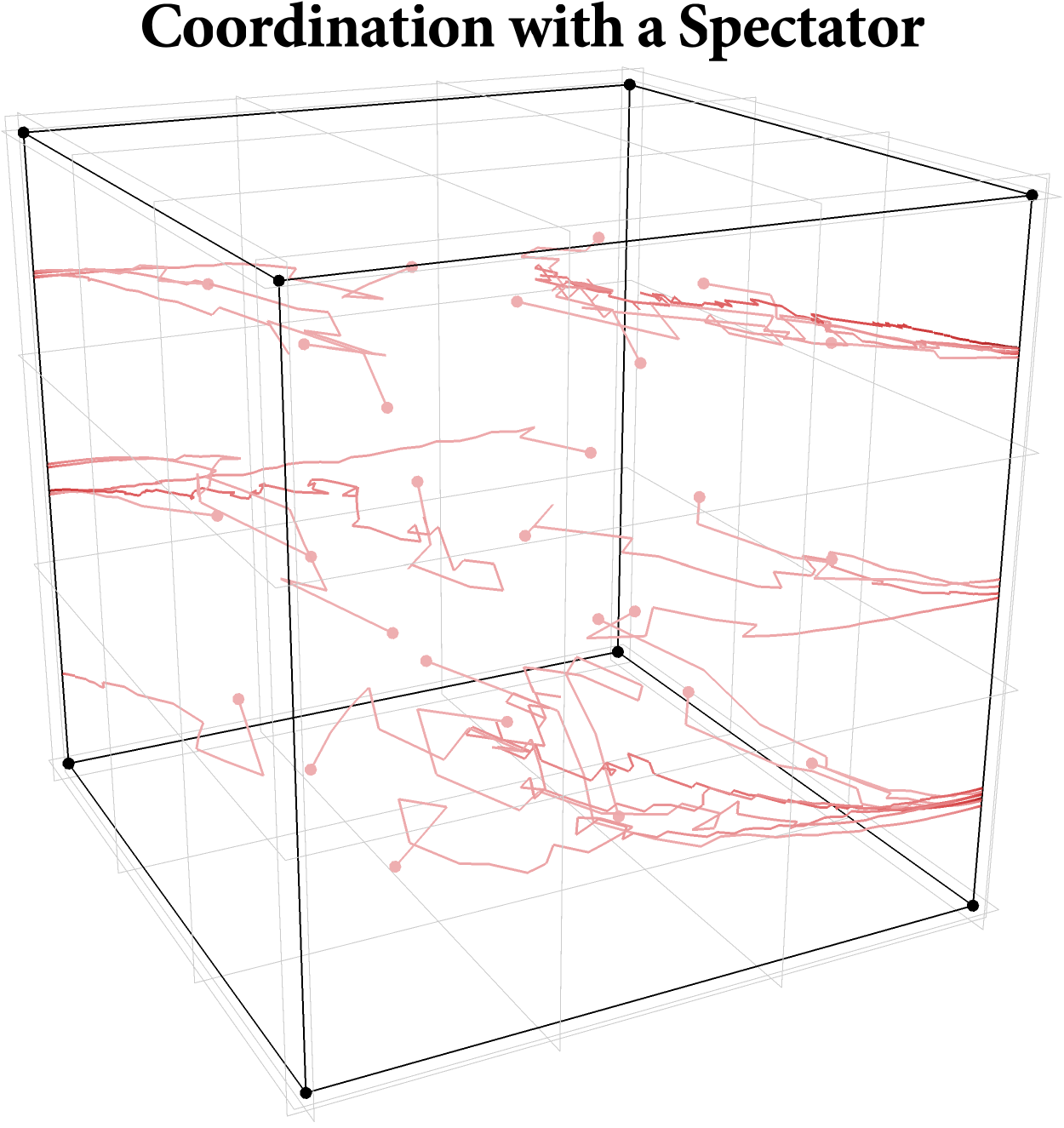}
\end{subfigure}
\hfill
\begin{subfigure}{.23\textwidth}
\includegraphics[width=\textwidth]{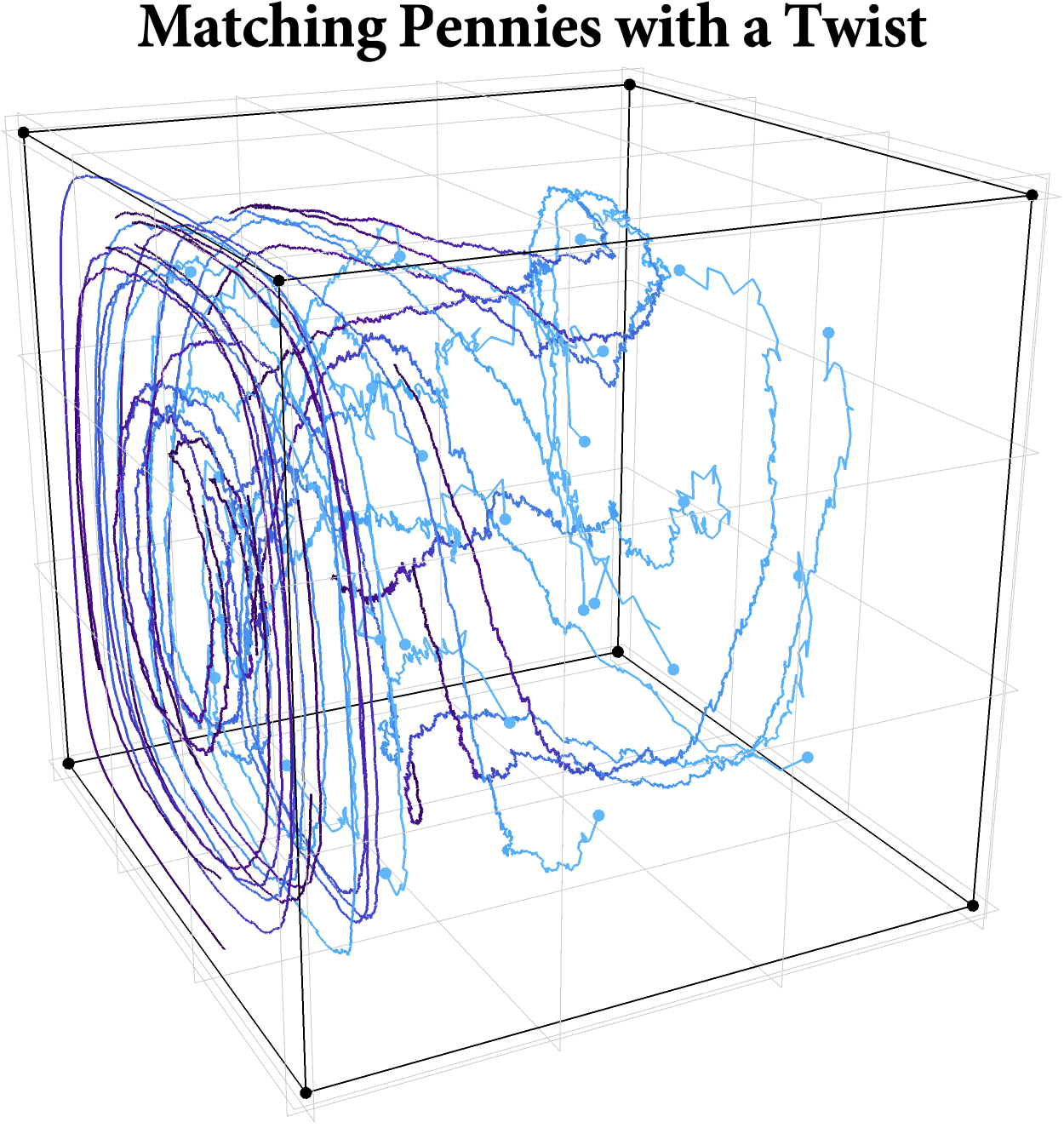}
\end{subfigure}
\hfill
\begin{subfigure}{.23\textwidth}
\includegraphics[width=\textwidth]{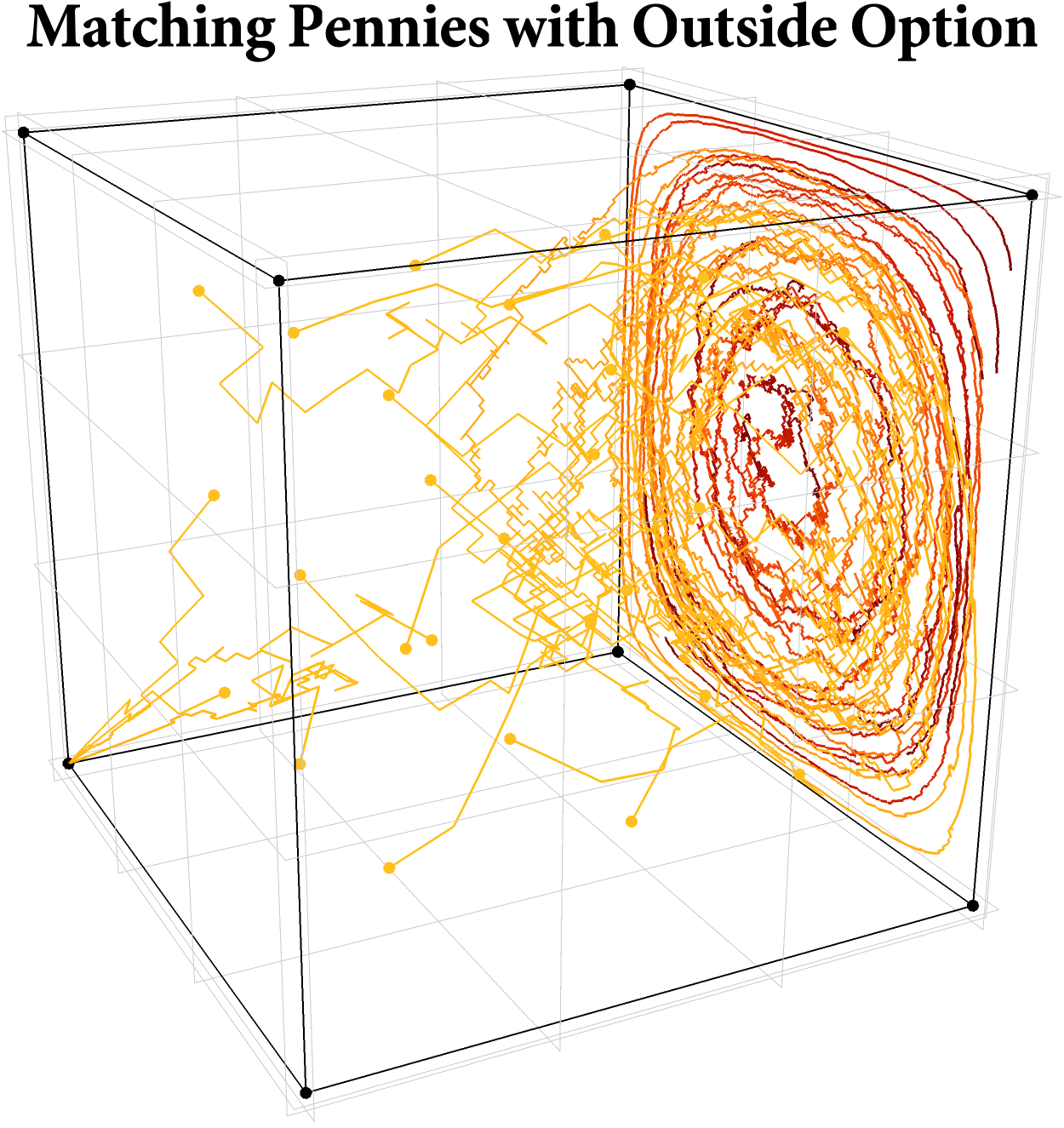}
\end{subfigure}%
%\vspace{-.5\baselineskip}
\caption{The long-run behavior of \ac{EXP3} (\cref{alg:BFTRL}) in four representative $2\times2\times2$ games.
In all cases, the dynamics converge to \ac{minimal} sets, either \emph{strict equilibria} themselves, or spanning an \emph{essential component} of \aclp{NE}.
The details of the numerics and the games being played are provided in the appendix.
%In the first, the limit sets of \ac{EXP3} are (strict) \aclp{NE}, indicating pointwise convergence;
%in the second and third, \ac{EXP3} converges to a heteroclinic limit cycle;
%in the fourth, to a heteroclinic cycle or a periodic orbit, depending on the initialization.
%Despite these very different convergence landscapes,
%\emph{the algorithm's limit sets are always irreducible and resilient,}
%as per \cref{thm:limset}.
%For the details of the numerics, see \cref{app:numerical}.
}
\label{fig:games}
\vspace{-1.5\baselineskip}
\end{figure*}

%% Limset figure ends here
%----------------------------------------------------------------------

\begin{proof}[Proof sketch]
The proof of \cref{thm:closed,thm:rate} is quite involved so we defer it to \cref{app:closed}.
At a high level, it hinges on constructing a family of ``primal-dual'' energy functions, one per pure deviation from the set $\set$ under study.
If unilateral deviations from $\set$ incur a cost to the deviator (that is, if $\set$ is \ac{closed}), these energy functions can be ``bundled together'' to produce a suitable Lyapunov-like function for $\set$.
In more detail, the minimization of each individual energy function implies that the score variable $\curr[\dstate]$ of \eqref{eq:method} diverges along an ``astral direction'' in the payoff space $\scores$ \textendash\ \ie it escapes to infinity along the interior of a certain convex cone of $\scores$ \cite{DST22}.
Because this minimization occurs at infinity, the aggregation of offsets and random errors in \eqref{eq:method} affords some extra ``wiggle room'' in our martingale analysis, so we are able to show that $\curr = \mirror(\curr[\dstate])$ remains close to $\set$ under a much wider range of parameters compared to \cref{thm:resilience}.
Then, a series of convex analysis arguments in the spirit of \cite{MS16} coupled with the definition of $\mirror$ allows us to show that the escape of $\curr[\dstate]$ along the intersection of all these cones implies convergence to $\set$ at the specified rate.

On the converse side, if an asymptotically stable set is not \ac{closed}, we can find a non-costly (and possibly profitable) deviation $\tvec$ from $\set$ which is selected against by \eqref{eq:method}.
However, this extinction runs contrary to the reinforcement of better replies under \eqref{eq:method}, an argument which can be made precise by applying the martingale law of large numbers to $\braket{\curr[\dstate]}{\tvec}$ \cite{HH80}.
The irreducible stability of \ac{minimal} sets then follows by invoking this criterion reductively for any potentially stable subset $\alt\set$ of $\set$.
\end{proof}

%----------------------------------------------------------------------
%%% DISCUSSION
%----------------------------------------------------------------------
\section{Discussion and concluding remarks}
\label{sec:discussion}
%----------------------------------------------------------------------
%%% DISCUSSION
%----------------------------------------------------------------------
% !TEX root = ../Main.tex

\Cref{thm:closed,thm:rate} are our main results linking dynamic and strategic stability, so we conclude with a series of remarks.

First, we should note that \cref{thm:closed} can be summed up as follows:
\emph{a product of pure strategies is \textpar{minimally} \acl{closed} if and only if its span is \textpar{irreducibly} stable under regularized learning.}
%\vspace{-\topsep}
%\begin{center}
%\itshape
%A product of pure strategies is \textpar{minimally} \acl{closed}\\
%if and only if its span is \textpar{irreducibly} stable and attracting under regularized learning.
%\end{center}
%\vspace{-\parsep}
Importantly, this equivalence is based solely on the game's payoff data:
it does not depend on the specific choices underlying \eqref{eq:method}, including the choice map employed by each player,
%\footnote{The choice of regularizer affects the speed of convergence but not the algorithm's end state.}
whether some players are using an optimistic adjustment or not,
if they have access to their full payoff vectors, etc.
As such, this equivalence provides a crisp operational criterion for identifying which pure strategy combinations ultimately persist under regularized learning \textendash\ and, via \cref{thm:rate}, \emph{how fast} this identification takes place.

In this light, \cref{thm:closed} essentially states that the only robust prediction that can be made for the outcome of a regularized learning process is (minimal) closedness under better replies.
%\textendash\ a requirement which is violated by the no-regret outcome of \cref{ex:VZ}.
This interpretation has significant cutting power for the emergence of rational behavior.
To begin, in terms of equilibrium play, it readily implies that a pure strategy profile is stochastically asymptotically stable under \eqref{eq:method} if and only if it is a strict \acl{NE}.
A version of this equivalence was only recently proved in \cite{FVGL+20} and \cite{GVM21} (in continuous and discrete time respectively), so \cref{thm:closed} can be seen as a far-reaching generalization of these recent results.
More to the point, since every \ac{minimal} set $\set$ contains an essential equilibrium component that is fully supported in $\set$, \cref{thm:closed} also provides an important link between dynamic and structural stability:
if an equilibrium \textendash\ or a component of equilibria \textendash\ is not robust to perturbations of the underlying game,
\emph{it cannot be robustly identified by a regularized learning process} (and vice versa).
This remark is of particular importance for extensive-form games as such games often have non-generic equilibrium components that cannot be treated otherwise by the existing theory.

The above also places severe restrictions on which \emph{components} of \aclp{NE} can be stable and attracting under \eqref{eq:method}:
if an equilibrium component \emph{does not} span a \ac{closed} set, it cannot be asymptotically stable (a fact which explains the behavior seen in the Entry Deterrence game in \cref{fig:games}).
This observation goes a long way toward explaining why regularized learning correctly identifies the support of \aclp{NE} in $2$-player zero-sum games (but cannot go further), and also serves to illustrate why the convergence of optimistic methods is destroyed in the presence of randomness and uncertainty \cite{CGFLJ19,HIMM20}.
By this token, \cref{thm:closed} can be seen as a trade-off between how \emph{robust} versus how \emph{informative} a learning prediction is from a strategic perspective.
We find this interpretation of \cref{thm:closed} particularly appealing as it opens the door to several fruitful research directions.

Finally, we should stress that \cref{thm:closed,thm:rate} guarantee convergence even with a constant step-size.
%\textendash\ and, in fact, the fastest convergence rates involve a constant step-size.
Together with the finite-time convergence guarantees of \cref{cor:rate} for projection-based methods, this feature is a testament to the robustness of \ac{closed} sets as, in the presence of uncertainty, convergence invariably requires a vanishing step-size which can slow things down to a crawl.
We find this robust convergence landscape particularly intriguing for future research on the topic.

%**********************************************************************
%***    ACKNOWLEDGMENTS
%**********************************************************************
\section*{Acknowledgments}
\begingroup
\small
%----------------------------------------------------------------------
%%% THANKS
%----------------------------------------------------------------------
% !TEX root = ./Main.tex
%
%
This work has been partially supported by
the French National Research Agency (ANR) in the framework of
the ``Investissements d'avenir'' program (ANR-15-IDEX-02),
the LabEx PERSYVAL (ANR-11-LABX-0025-01),
MIAI@Grenoble Alpes (ANR-19-P3IA-0003),
and
project MIS 5154714 of the National Recovery and Resilience Plan Greece 2.0 funded by the European Union under the NextGenerationEU Program.
PM is also a member of the Archimedes Unit, Athena RC, Department of Mathematics, National \& Kapodistrian University of Athens.
\endgroup

%**********************************************************************
%***    APPENDICES
%**********************************************************************
\appendix
\numberwithin{equation}{section}		% for numbering  in the appendix
\numberwithin{lemma}{section}		% for numbering  in the appendix
\numberwithin{proposition}{section}		% for numbering  in the appendix
\numberwithin{theorem}{section}		% for numbering in the appendix
\numberwithin{corollary}{section}		% for numbering  in the appendix

%----------------------------------------------------------------------
%%% APP: AUX
%----------------------------------------------------------------------
\section{Auxiliary results}
\label{app:aux}
%----------------------------------------------------------------------
%%% APP: AUX
%----------------------------------------------------------------------
% !TEX root = ../Main.tex

In this appendix we collect some basic properties of the regularized choice maps and some results from probability theory that will be useful in the sequel.

%----------------------------------------------------------------------
%% Choice
%----------------------------------------------------------------------
\subsection{Regularized choice maps and their properties}

Thoughout this appendix, we will suppress the player index $\play\in\players$, and we will follow standard conventions in convex analysis \cite{RW98} that treat $\hreg$ as an extended-real-valued function $\hreg\from\vecspace\to\R\cup\{\infty\}$ with $\hreg(\strat) = \infty$ for all $\strat \in \vecspace \setminus \strats$.
With this in mind, the subdifferential of a $\hreg$ at $\strat\in\strats$ is defined as
\begin{equation}
\subd\hreg(\strat)
	\defeq \setdef
		{\score\in\scores}
		{\hreg(\stratalt) \geq \hreg(\strat) + \braket{\score}{\stratalt - \strat} \; \text{for all $\stratalt\in\strats$}},
\end{equation}
where $\scores$ denotes here the algebraic dual $\dspace$ of $\vecspace$.
Accordingly, the \emph{domain of subdifferentiability} of $\hreg$ is $\dom\subd\hreg \defeq \setdef{\strat\in\dom\hreg}{\subd\hreg \neq \varnothing}$,
and
the convex conjugate of $\hreg$ is defined as
\begin{equation}
\label{eq:hconj}
\hconj(\score)
	= \max_{\strat\in\strats} \{ \braket{\score}{\strat} - \hreg(\strat) \}
\end{equation}
for all $\score\in\scores$.
We then have the following basic results.

\begin{lemma}
\label{lem:mirror}
Let $\hreg$ be a regularizer on $\strats$, and let $\mirror\from\scores\to\strats$ be the induced choice map.
Then:
\begin{enumerate}
\item
$\mirror$ is single-valued, and, for all $\strat\in\strats$, $\score\in\scores$, we have $\strat = \mirror(\score) \iff \score \in \subd\hreg(\strat)$.
\item
For all $\strat\in\relint\strats$, we have
\begin{equation}
\label{eq:subdiff}
%\(
\subd\hreg(\strat)
	= \setdef{(\hker'(\strat_{\pure}) + \mu)_{\pure\in\pures}}{\mu\in\R}.
%\)
\end{equation}
\item
The prox-domain $\proxdom \defeq \im\mirror$ of $\hreg$ satisfies $\relint\strats \subseteq \proxdom \subseteq \strats$.
\item
For all $\score\in\scores$, we have $\mirror(\score) = \nabla\hconj(\score)$.
%\item
and
$\mirror$ is $(1/\hstr)$-Lipschitz continuous with $\hstr \defeq \inf_{(0,1]} \hker''(z)$.
In particular, as a special case, the logit choice map $\logit$ is $1$-Lipschitz continuous in the $(L^{1},L^{\infty})$ pair of norms on $\scores$ and $\strats$ respectively.

\item
If $\score_{\pure} - \score_{\purealt} \to -\infty$ for some $\purealt\neq\pure$, then $\mirror_{\pure}(\score) \to 0$.
\end{enumerate}
\end{lemma}

\begin{remark*}
Some of the properties presented in \cref{lem:mirror} are well known in the literature on regularized learning methods (see \eg \cite{MS16} and references therein), but we provide a proof of the entire lemma for completeness.
\endrem
\end{remark*}

\begin{proof}[Proof of \cref{lem:mirror}]
For the first property of $\mirror$, note that the maximum in \eqref{eq:mirror} is attained for all $\score\in\scores$ because $\hreg$ is \ac{lsc} and strongly convex.
Furthermore, $\strat$ solves \eqref{eq:mirror} if and only if $\score - \subd\hreg(\strat) \ni 0$, \ie if and only if $\score\in\subd\hreg(\strat)$.
%The above also shows that $\proxdom = \dom\subd\hreg$;
%since $\relint\strats \subseteq \dom\subd\hreg \subseteq \strats$ \citep[Chap.~26]{Roc70}, our second claim follows.

    For our second claim, if $\strat\in\relint(\strats)$, the first-order stationarity conditions for the convex problem \eqref{eq:mirror} that defines $\mirror$ become
\begin{equation}
\label{eq:KKT-interior}
\score_{\pure} - \hker'(\strat_{\pure})
	= \mu
	\quad
	\text{for all $\pure\in\pures$,}
\end{equation}
    because the inequality constraints $\strat_{\pure}\geq0$ are all inactive (recall that $\strat\in\relint(\strats)$ by assumption).
Now, by the first part of the theorem we have $\strat = \mirror(\score)$ if and only if $\score\in\subd\hreg(\strat)$, so we conclude that $\subd\hreg(\strat) = \setdef{(\hker'(\strat_{\pure}) + \mu)_{\pure\in\pures}}{\mu\in\R}$, as claimed.

For the fourth item, the expression $\mirror = \nabla\hconj$ is an immediate consequence of Danskin's theorem,
while the Lipschitz continuity of $\mirror$ follows from standard results, see \eg \citep[Theorem 12.60(b)]{RW98}.

For our last claim, let $\curr[\score]$ be a sequence in $\scores$ such that $\score_{\pure,\run} - \score_{\purealt,\run} \to -\infty$ and let $\curr[\strat] = \mirror(\curr[\score])$.
Then, by descending to a subsequence if necessary, assume there exists some $\eps>0$ such that $\strat_{\pure,\run} \geq \eps > 0$ for all $\run$.
Then, by the defining relation $\mirror(\score) = \argmax\{\braket{\score}{\strat} - \hreg(\strat)\}$ of $\mirror$, we have:
\begin{flalign}
\label{eq:Qcomp1}
\braket{\curr[\score]}{\curr[\strat]} - \hreg(\curr[\strat])
	\geq \braket{\curr[\score]}{\stratalt} - \hreg(\stratalt)
\end{flalign}
for all $\stratalt\in\strats$.
Therefore, taking $\curr[\strat]' = \curr[\strat] + \eps(\bvec_{\purealt} - \bvec_{\pure})$, we readily obtain
\begin{equation}
\label{eq:Qcomp2}
\eps (\score_{\pure,\run} - \score_{\purealt,\run})
	\geq \hreg(\curr[\strat]) - \hreg(\curr[\strat]')
	\geq \min h - \max h
\end{equation}
which contradicts our original assumption that $y_{\pure,\run} - y_{\purealt,\run} \to -\infty$.
With $\strats$ compact, the above shows that $x_{\pure}^{\ast} = 0$ for any limit point $x^{\ast}$ of $\curr[\strat]$, i.e. $\mirror_{\pure}(\curr[\score])\to0$.
\end{proof}

The second collection of results concerns the \emph{Fenchel coupling}, an energy function that was first introduced in \cite{MS16,MZ19} and is defined as follows:
\begin{equation}
\label{eq:Fench}
\fench(\base,\score)
	= \hreg(\base) + \hconj(\score) - \braket{\score}{\base}
	\quad
	\text{for all $\base\in\strats$ and $\score\in\scores$}.
\end{equation}
This coupling will play a major role in the proofs of \cref{thm:resilience}, so we prove two of its most basic properties below.

\begin{lemma}
\label{lem:Fench}
For all $\base\in\strats$ and all $\score,\scorealt\in\scores$, we have:
\begin{subequations}
\setlength{\abovedisplayskip}{3pt}
\setlength{\belowdisplayskip}{3pt}
\begin{flalign}
\label{eq:Fench-norm}
\qquad
a)
	&\quad
	\fench(\base,\score)
	\geq \tfrac{1}{2} \hstr \, \norm{\mirror(\score) - \base}^{2}.
	&
	\\
\label{eq:Fench-bound}
\qquad
b)
	&\quad
	\fench(\base,\scorealt)
	\leq \fench(\base,\score) + \braket{\scorealt - \score}{\mirror(\score) - \base} + \tfrac{1}{2\hstr} \dnorm{\scorealt-\score}^{2}.
\end{flalign}
\end{subequations}
In particular, if $\hreg(0) = 0$, we have
\begin{equation}
(\hstr/2) \norm{\mirror(\score)}^{2}
	\leq \hconj(\score)
	\leq -\min\hreg
		+ \braket{\score}{\mirror(\score)}
		+ (2/\hstr) \dnorm{\score}^{2}
	\quad
	\text{for all $\score\in\scores$.}
\end{equation}
\end{lemma}

\begin{proof}[Proof of \cref{lem:Fench}]
By the strong convexity of $\hreg$ relative to $\norm{\cdot}$ (\cf \cref{lem:mirror}), we have
\begin{flalign}
\hreg(\strat) + t \braket{\score}{\base - x}
	&\leq \hreg(x + t(\base - x))
	\notag\\
	&\leq t\hreg(\base) + (1-t)\hreg(\strat) - \tfrac{1}{2} \hstr t(1-t) \norm{x - \base}^{2},
\end{flalign}
leading to the bound
\begin{equation}
\label{eq:divbound}
\tfrac{1}{2} \hstr(1-t) \norm{x - \base}^{2}
	\leq \hreg(\base) - \hreg(\strat) - \braket{\score}{\base - x}
	= \fench(\base,y)
\end{equation}
for all $t\in(0,1]$.
The bound \eqref{eq:Fench-norm} then follows by letting $t\to0^{+}$ in \eqref{eq:divbound}.

For our second claim, we have
\begin{flalign}
\fench(\base,\scorealt)
	&= \hreg(\base) + \hconj(\scorealt) - \braket{\scorealt}{\base}
	\notag\\
	&\leq \hreg(\base) + \hconj(\score) + \braket{\scorealt - \score}{\nabla \hconj(\score)} + \frac{1}{2\hstr} \dnorm{\scorealt - \score}^{2} - \braket{\scorealt}{\base}
	\notag\\
	&= \fench(\base,\score) + \braket{\scorealt - \score}{\mirror(\score) - \base} + \frac{1}{2\hstr} \dnorm{\scorealt - \score}^{2},
\end{flalign}
where the inequality in the second line follows from the fact that $\hconj$ is $(1/\hstr)$-strongly smooth \citep[Theorem 12.60(e)]{RW98}.
%and we used \cref{prop:choice} in the last line.
\end{proof}

%----------------------------------------------------------------------
%% Probability
%----------------------------------------------------------------------
\subsection{Basic results from probability theory}

We conclude this appendix with some useful results from probability theory that we will use freely throughout the sequel.
For a complete treatment, we refer the reader to \citet{HH80}.

\begin{lemma}
[Azuma-Hoeffding inequality]
\label{lem:Azuma}
Let $\curr[\mart] \in \R$, $\run=\running$, be a martingale with $\supnorm{\curr[\mart] - \prev[\mart]} \leq \curr[\martbd]$ \as.
Then, for all $\conf > 0$, we have
\begin{equation}
\probof*{
	\abs{\curr[\mart]}
		\leq \parens*{ 2\log(2\run^{2} / \conf) \sum\nolimits_{\runalt=1}^{\run} \curralt[\martbd]^{2}}^{1/2}
			\;\text{for all $\run$}
		}
	\geq 1 - \conf.
\end{equation}
\end{lemma}

\begin{lemma}
[Kolmogorov's inequality]
\label{lem:Kolmogorov}
Let $\curr[\err] \in \R$, $\run=\running$, be a martingale difference sequence that is bounded in $L^{2}$.
Then:
\begin{equation}
\probof*{\max_{\runalt \leq \run} \sum\nolimits_{\ell=1}^{\runalt} \err_\ell \ge \eps}
	\leq \frac 1{\eps^{2}} {\ex\bracks*{\parens*{\sum\nolimits_{\runalt=1}^{\run} \err_\runalt}^{2}}}
	\quad
	\text{for all $\eps>0$}.
\end{equation}
\end{lemma}

\begin{lemma}
[Doob's maximal inequality]
\label{lem:Doob}
Let $\curr[\err] \in \R$, $\run=\running$, be a martingale difference sequence that is bounded in $L^{\pexp}$ for some $\pexp\geq1$.
Then 
\begin{equation}
\probof*{ \max_{\runalt \leq \run} \abs{\err_{\runalt}} > \eps }
	\leq \frac{1}{\eps^\pexp} \ex \bracks[\big]{\abs{\curr[\err]}^\pexp}
		\quad
		\text{for all $\eps>0$}.
\end{equation}
\end{lemma}

\begin{lemma}
[\acl{BDG} inequality]
\label{lem:Burkholder}
Let $\curr[\err]$, $\run=\running$, be a martingale difference sequence in $\R^{\vdim}$.
Then, for all $\pexp>1$, there exist constants $\const_{\pexp}, \Const_{\pexp}$ that depend only on $\pexp$ and are such that 
\begin{equation}
\const_{\pexp} \exof*{\sum_{\runalt=1}^{\run} \twonorm{\err_\runalt}^{2}}^{\pexp/2}
	\leq \exof*{ \max_{\runalt\le\run} \twonorm*{\sum_{\ell=1}^{\runalt} \err_\ell}^{\pexp} }
	\leq \Const_{\pexp} \exof*{\sum_{\runalt=1}^{\run} \twonorm{\err_\runalt}^{2}}^{\pexp/2}.
	\end{equation}
%	where $\twonorm{\cdot}$ is the Euclidian norm.
\end{lemma}

\begin{lemma}
[Robbins\textendash Siegmund]
\label{lem:RS}
Let $\curr[\filter]$, $\run=\running$, be a filtration on a complete probability space $\probspace$, and suppose that the sequences $\curr$, $\curr[L]$ and $\curr[K]$ $\curr[\filter]$-measurable, nonnegative, and such that
%$\exof{\next \given \curr[\filter]}$ exists and the following inequality is satisfied
\begin{equation}
\exof{\next \given \curr[\filter]}
	\leq \curr (1 + \curr[L]) + \curr[K]
	\quad
	\text{\acl{wp1}.}
\end{equation}
Then, $\curr$ converges to some random variable $\state_{\infty}$ \acl{wp1} on the event
\begin{equation}
%\(
\braces*{\sum_{\run=1}^\infty \curr[L] < \infty \;\text{ and }\; \sum_{\run=1}^\infty \curr[K] < \infty}.
%\)
\end{equation}
%	the limit $\state_\infty \defeq \lim_{\run \to \infty} \curr$ exists. 
\end{lemma}

%----------------------------------------------------------------------
%%% APP: ALGORITHMS
%----------------------------------------------------------------------
\section{Specific algorithms and their properties}
\label{app:algorithms}
%----------------------------------------------------------------------
%%% APP: ALGORITHMS
%----------------------------------------------------------------------
% !TEX root = ../Main.tex

%----------------------------------------------------------------------
%%% Algorithms
%----------------------------------------------------------------------
\subsection{Known algorithms as special cases of \eqref{eq:method}}

To complement our analysis in the main part of our paper, we detail below how \crefrange{alg:FTRL}{alg:BFTRL} can be recast in the general framework of \eqref{eq:method}.
To lighten notation, we will assume that $\curr[\bias]$, $\curr[\noise]$ and $\curr[\signal]$ are respectively bounded as
\begin{equation}
\label{eq:errorbounds}
\dnorm{\curr[\bias]}
	\leq \curr[\bbound]
	\qquad
\dnorm{\curr[\noise]}
	\leq \curr[\sdev]
	\qquad
	\text{and}
	\qquad
\dnorm{\curr[\signal]}
	\leq \curr[\sbound]
\end{equation}
and we will set
\begin{equation}
\label{eq:vbound}
\vbound
	\defeq \max_{\play\in\players} \max_{\pure\in\pures} \abs{\payv_{\play}(\pure)}
\end{equation}
so we can take $\curr[\sbound] = \vbound + \curr[\bbound] + \curr[\sdev]$ in \eqref{eq:errorbounds}.
%\textendash\ meaning in particular that the bounds for $\curr[\sbound]$ follow immediately from the corresponding bounds for $\curr[\bbound]$ and $\curr[\sdev]$.
We will also make free use of the fact that $\payv$ is Lipschitz continuous on $\strats$, and we will write $\lips$ for its Lipschitz modulus in the $(L^{1},L^{\infty})$ pair of norms on $\strats$ and $\scores$ respectively, \viz
\begin{equation}
\label{eq:Lips}
\supnorm{\payv(\stratalt) - \payv(\strat)}
	\leq \lips \onenorm{\stratalt - \strat}
	\qquad
	\text{for all $\strat,\stratalt\in\strats$}.
\end{equation}

We now proceed to establish the required bounds for \crefrange{alg:FTRL}{alg:BFTRL}:

%----------------------------------------------------------------------
%%% EW
%----------------------------------------------------------------------
\para{\cref{alg:FTRL}}

Since $\curr[\signal] = \payv(\curr)$, we readily get $\curr[\bias] = \curr[\noise] = 0$ by definition, so \cref{alg:FTRL} fits the scheme \eqref{eq:method} for free with $\biasexp = \infty$, $\noisexp=0$.
\endenv
\smallskip
%----------------------------------------------------------------------

%----------------------------------------------------------------------
%%% OMW
%----------------------------------------------------------------------
\para{\cref{alg:OFTRL}}

For the case of \eqref{eq:OFTRL}, we have $\curr[\signal] = 2\payv(\curr) - \payv(\prev)$ so $\curr[\bias] = \payv(\curr) - \payv(\prev)$, which is $\curr[\filter]$-measurable.
We thus get
\begin{align}
\dnorm{\curr[\bias]}
	&= \dnorm{\exof{\curr[\signal] \given \curr[\filter]} - \payv(\curr)}
	\notag\\
	&\leq \exof{ \dnorm{\payv(\curr) - \payv(\prev)} \given \curr[\filter]}
	\notag\\
	&\leq \lips \exof{\norm{\curr - \prev} \given \curr[\filter]}
		\explain{by \eqref{eq:Lips}}
		\\
	&= \lips \exof{\dnorm{\mirror(\curr[\dstate]) - \mirror(\prev[\dstate])} \given \curr[\filter]}
		\explain{by \eqref{eq:OFTRL}}
		\\
	&\leq (\lips/\hstr) \exof{\dnorm{\curr[\dstate] - \prev[\dstate]} \given \curr[\filter]}
		\explain{by \cref{lem:mirror}}
		\\
	&\leq \curr[\step] (\lips/\hstr) \exof{ 2\payv(\curr) - \payv(\prev) \given \curr[\filter] }
		\explain{by \eqref{eq:OFTRL}}
		\\
	&\leq 3 \lips \vbound / \hstr \cdot \curr[\step]
		\explain{by \eqref{eq:vbound}}
		\\
	&= \bigoh(\curr[\step])
	= \bigoh(1/\run^{\stepexp})
\end{align}
Moreover, given that $\signal$ is $\curr[\filter]$-measurable, we readily get $\curr[\noise] = 0$. 
\hfill
\endenv
\smallskip
%----------------------------------------------------------------------

%----------------------------------------------------------------------
%%% EXP3
%----------------------------------------------------------------------
\para{\cref{alg:BFTRL}}

Since $\curr[\purequery]$ is sampled according to $\curr[\query] = (1-\curr[\mix]) \state_{\play,\run} + \curr[\mix] \unif_{\pures_{\play}}$ (\cf \cref{eq:explore} in \cref{sec:learning}), we readily obtain $\exof{\signal_{\play,\run} \given \curr[\filter]} = \payv_{\play}(\curr[\query])$, and hence, by \eqref{eq:Lips}, we get
\begin{equation}
\curr[\bbound]
	= \bigoh(\norm{\curr[\query] - \curr})
	= \bigoh(\curr[\mix])
	= \bigoh(1/\run^{\mixexp}).
\end{equation}
Moreover, since $\query_{\play\pure_{\play},\run} \geq \curr[\mix]/\nPures_{\play}$, it follows that $\dnorm{\curr[\signal]} = \bigoh(1/\curr[\mix]) = \bigoh(\run^{\mixexp})$.
\hfill
\endenv
\smallskip
%----------------------------------------------------------------------

For comparison purposes, we illustrate the algorithms' behavior in a simple $2\times2\times2$ game in \cref{fig:limit} in \cref{app:numerics}.

%%----------------------------------------------------------------------
%%% Convergence figure begins here
%
%\begin{figure}[t]
%\centering
%\footnotesize
%\begin{subfigure}{.48\textwidth}
%\includegraphics[height=.48\textwidth]{Figures/Hedge-Het.pdf}
%%\;\;
%\includegraphics[height=.48\textwidth]{Figures/OMW-Het.pdf}
%\caption{Full-information methods (\cref{alg:EW,alg:OMW}).}
%\end{subfigure}
%\hfill
%\begin{subfigure}{.48\textwidth}
%\includegraphics[height=.48\textwidth]{Figures/EXP3-Het.pdf}
%%\;\;
%\includegraphics[height=.48\textwidth]{Figures/Tsallis-Het.pdf}
%\caption{Payoff-based methods (\cref{alg:EXP3,alg:Tsallis}).}
%\end{subfigure}
%\caption{The long-run behavior of \cref{alg:EW,alg:OMW,alg:EXP3,alg:Tsallis} in a $2\times2\times2$ game with a boundary attractor (a heteroclinic cycle).}
%\label{fig:hetcycle}
%\end{figure}
%
%%% Convergence figure ends here
%%----------------------------------------------------------------------

%----------------------------------------------------------------------
%%% Further algos
%----------------------------------------------------------------------
\subsection{Further algorithms and illustrations}

To demonstrate the breadth of \eqref{eq:method} as an algorithmic template, we provide below some more examples of algorithms from the game-theoretic literature that can be recast as special cases thereof (see also \cref{tab:algorithms} for a recap).

%----------------------------------------------------------------------
\begin{algo}[\Acl{MP}]
\label{alg:MP}
A progenitor of \eqref{eq:OFTRL} is the so-called \acdef{MP} algorithm \cite{Nem04,JNT11}, which updates as:
\begin{equation}
\label{eq:MP}
\tag{\acs{MP}}
\begin{alignedat}{3}
\lead[\dstate]
	&= \curr[\dstate] + \curr[\step] \payv(\curr)
	&\qquad
\next[\dstate]
	&= \curr[\dstate] + \curr[\step] \payv(\lead)
	\\
\lead
	&= \mirror(\lead[\dstate])
	&\qquad
\next
	&= \mirror(\next[\dstate]).
\end{alignedat}
\end{equation}
The main difference between \eqref{eq:MP} and \eqref{eq:OFTRL} is that the former utilizes two surrogate gain vectors per iteration \textendash\ meaning in particular that the interim, leading state $\lead$ is generated with payoff information from $\curr$, not $\beforelead$.
This method has been used extensively in the literature for solving variational inequalities and two-player, zero-sum games, \cf \citet{JNT11} and references therein.

A calculation similar to that for \eqref{eq:OFTRL} shows that \cref{alg:MP} has $\curr[\bbound] = \bigoh(1/\run^{\stepexp})$ and $\curr[\sdev] = 0$ because the algorithm has no further randomization. 
\endenv
\end{algo}

\smallskip
%----------------------------------------------------------------------

%----------------------------------------------------------------------
\begin{algo}[\Acl{CMW}]
\label{alg:CMW}
A recent variant of \eqref{eq:EW} is the so-called \acdef{CMW} algorithm \citep{PSS21}
\begin{equation}
\label{eq:CMW}
\tag{\acs{CMW}}
\dstate_{\play,\run+1}
	= \dstate_{\play,\run} + \curr[\step] \payv_{\play}(\next)
	\qquad
\state_{\play,\run+1}
	= \logit_{\play}(\dstate_{\play,\run+1}).
\end{equation}
The main difference between \eqref{eq:CMW} and \eqref{eq:EW} is that the proxy payoff vector $\curr[\signal]$ in \eqref{eq:CMW} is based on the \emph{future} state $\next$ and \emph{not} the current state $\curr$.
To perform this ``clairvoyant'' update, the players of the game must coordinate to solve an implicit fixed point problem, so \eqref{eq:CMW} is only meaningful when one has access to the payoff function $\payv(\cdot)$.
In this regard, \eqref{eq:CMW} can be seen as a Bregman proximal point method in the general spirit of \citet{BBC03}.

To cast \eqref{eq:CMW} as an instance of the generalized template \eqref{eq:method}, simply note that the sequence of input signals is given by $\curr[\signal] = \payv(\next)$, so $\curr[\noise] = 0$ and $\curr[\bias] = \payv(\next) - \payv(\curr) = \bigoh(\curr[\step]) = \bigoh(1/\run^{\stepexp})$.
\endenv
\end{algo}
\smallskip
%----------------------------------------------------------------------

%----------------------------------------------------------------------
%% Table of algorithms begins here

\begin{table}[t]
\centering
\renewcommand{\arraystretch}{1.2}
\setlength{\tabcolsep}{.5em}
\small
%----------------------------------------------------------------------
%%% ALGORITHMS
%----------------------------------------------------------------------
% !TEX root = ../Main.tex

\begin{tabular}{rcccccc}
\toprule
	&\textbf{Representative}
	&\textbf{Regularizer ($\hker$)}
	&\textbf{Feedback}
	&\textbf{Bias ($\curr[\bbound]$)}
	&\textbf{Variance ($\curr[\sdev]$)}
	\\
\midrule
\cref{alg:FTRL}
	&\ac{EW}
	&$z\log z$
	&full info
	&$0$
	%&\eqref{eq:orcl-mixed}: $0$;\;
	% \eqref{eq:orcl-pure}: $\bigoh(1)$
    & $0$
	\\
\cref{alg:OFTRL}
	&\ac{OMW}
	&$z\log z$
	&full info
	&$\bigoh(1/\run^{\stepexp})$
	% &\eqref{eq:orcl-mixed}: $0$;\;
    % \eqref{eq:orcl-pure}: $\bigoh(1)$
    & $0$
	\\
\cref{alg:BFTRL}
	&\ac{EXP3}
	&$z\log z$
	&payoff
	&$\bigoh(1/\run^{\mixexp})$
	&$\bigoh(\run^{\mixexp})$
	\\
\cref{alg:BFTRL}
	&\ac{Tsallis}
	&$-4\sqrt{z}$
	&payoff
	&$\bigoh(1/\run^{\mixexp})$
	&$\bigoh(\run^{\mixexp})$
	\\
\cref{alg:MP}
	&\ac{MP}
	&general
	&full info
    &$\bigoh(1/\run^{\stepexp})$
	% &\eqref{eq:orcl-mixed}: $0$;\;
	% \eqref{eq:orcl-pure}: $\bigoh(1)$
    & $0$
	\\
\cref{alg:CMW}
	&\acs{CMW}
	&$z\log z$
	&full info
	&$\bigoh(1/\run^{\mixexp})$
	&$0$
	\\
\bottomrule
\end{tabular}

\medskip
\caption{A range of algorithms adhering to the general template \eqref{eq:method} and their bias and variance characteristics when run with
a step-size sequence of the form $\curr[\step] = \step/\run^{\stepexp}$, $\stepexp\in(0,1]$, and, where applicable, a sampling parameter $\curr[\mix] = \mix/\run^{\mixexp}$.}
\label{tab:algorithms}
\end{table}

%% Table of algorithms ends here
%----------------------------------------------------------------------

%For completeness, we present the characteristics of the entire range of algorithms considered so far in \cref{tab:algorithms}.
%For illustration purposes, we also provide an additional series of runs of \crefrange{alg:EW}{alg:Tsallis} in \cref{fig:hetcycle}.

%----------------------------------------------------------------------
%%% APP: RESILIENCE
%----------------------------------------------------------------------
\section{Proof of \cref{thm:resilience}}
\label{app:resilience}
%----------------------------------------------------------------------
%%% APP: LIMIT SETS
%----------------------------------------------------------------------
% !TEX root = ../Main.tex

\newmacro{\MyAmazingLipschitzMirrorConstantButItsReversed}{(1/\hstr)}

Our main goal in this appendix will be to prove \cref{thm:resilience} on the resilience properties of \eqref{eq:method}.
For convenience, we restate below the relevant result for ease of reference:

\RESILIENCE*

\begin{proof}
Our proof that $\limset(\state)$ is resilient hinges on an energy-based technique that we will employ repeatedly in other parts of our analysis.
To begin, introduce a player-strategy deviation pair $(\play, z_\play)$, and say that a set is resilient \emph{to} $(\play, z_\play)$ if there exists an element of the set, say $\eq$, which counters said deviation, \ie such that $\pay_\play(\eq) \geq \pay_\play(z_\play;\strat_{-\play}^{\ast})$. 
In this specific case, our proof proceeds by contradiction, namely by assuming that, with positive probability, $\limset(\state)$ is \emph{not} resilient to $(\play, z_\play)$.
The main steps of our proof unfold as follows:

\begin{proofstep}
\label{prf:better}
Assume that $\limset(\state)$ is not resilient to $(\play, z_\play)$ with positive probability.
Then there exists $c, \epsilon, \run_0 > 0$ such that 
\begin{equation}
\label{eq:thm1:step1}
\probof*{\pay_\play(z_\play;\state_{\run, -\play}) \geq \pay_\play(\curr) + c \;\; \text{for all $\run \geq \run_{0}$}}
	\geq \epsilon.
\end{equation}
\end{proofstep}

\begin{proof}[Proof of \cref{prf:better}]
The function $f : \strat \in \strats \mapsto \pay_\play(z_\play;\strat_{-\play})  - \pay_\play(\strat)$ is continuous and $\strats$ is compact, so there is a definite function $\eta \equiv \eta(\delta)$ such that if $\norm{\strat - \stratalt} \leq \eta(\delta)$, then $\abs{f(\strat) - f(\stratalt)} \leq \delta$. 
	Now, by assumption, $\braces{\forall \eq \in \limset(\state), \pay_\play(z_\play; \eq_{-\play}) > \pay_\play(\eq)}$ is of positive probability. 
	We thus get
	\begin{subequations}
	\begin{align}
		0 
		& < 
		\P\braces*{
			\forall \eq \in \limset(\state),
			\pay_\play(z_\play; \eq_{-\play}) > \pay_\play(\eq)
		}
		\notag\\
		& = \label{equation:resilience:proof:1}
		\P \braces*{
			\inf_{\eq \in \limset(\state)} \parens*{
				\pay_\play(z_\play;\eq_{-\play}) - \pay_\play(\eq)
			}
			> 0
		}
		\\
		& = \label{equation:resilience:proof:2}
		\P \parens*{
			\bigcup_{m > 0} \braces*{
				\inf_{\eq \in \limset(\state)} \parens*{
				\pay_\play(z_\play;\eq_{-\play}) - \pay_\play(\eq)
				}
				> 2^{-m}
			}
		} 
		\\
		& \leq \label{equation:resilience:proof:3}
		\frac12 \P \braces*{
			\forall {\eq \in \limset(\state)},
			\pay_\play(z_\play;\eq_{-\play}) - \pay_\play(\eq)
			> 2c
		}
	\end{align}
	for some $c > 0$ in \eqref{equation:resilience:proof:3}, and where \eqref{equation:resilience:proof:1} is because $\limset(\state)$ is closed -- hence compact -- almost surely. 
	Therefore, by definition of $\eta(\cdot)$,
	\begin{align}
		0 
		& < \label{equation:resilience:proof:4}
		\P \braces*{
			\forall \eq \in \strats,
			\dist(x^{\ast}, \limset(\state)) \leq \eta\parens*{c} \Rightarrow
			\pay_\play(z_\play;\eq_{-\play}) - \pay_\play(\eq)
			>  c
		} = 2\epsilon
	\end{align}
	\end{subequations}
	Now, let $\run_0$ such that $\prob\braces{\forall \run \geq \run_0, \dist(\state_\run, \limset(\state)) \leq \eta(c)} > 1 - \frac\epsilon2$. 
	Then by construction, we get
	\begin{equation}
		\prob\braces*{
			\forall \run \geq \run_0, 
			\pay_\play(z_\play;X_{\run, - \play}) > \pay_\play(\state_\run) + c
		} > \epsilon.
%		\qedhere
	\end{equation}
	and our proof is complete.
\end{proof}

Intuitively, the existence of an action that consistently outperforms $\curr$ runs contrary to the behavior that one would expect from any regularized learning algorithm.
We will proceed to make this intuition precise below by means of an energy argument.
To that end, consider the Fenchel coupling 
\begin{equation}
\fench_{\play,\run}
	= \hreg_{\play}(z_{\play})
		+ \hconj_{\play}(\dstate_{\play,\run})
%		- \dstate_{z_{\play},\run}
		- \braket{\dstate_{\play,\run}}{z_{\play}}
\end{equation}
Then, by \cref{lem:Fench} in \cref{app:aux}, we readily get that
%Relate $\fench_{\play,\run+1}$ to $\fench_{\play,\run}$ with a Taylor expansion of $\hconj_{\play}(\cdot)$:
\begin{align}
\label{eq:thm1:step2:fenchel rec}
\fench_{\play,\run+1}
%	&= 
%	\hreg_{\play}(z_{\play}) + \hconj_{\play}(\dstate_{\play,\run} + \curr[\step]\signal_{\play,\run}) + \braket{\dstate_{\play,\run} + \curr[\step]\signal_{\play,\run}}{z_{\play}} \\
%	&\leq \label{eq:thm1:step2:strong convexity}
%	\hreg_{\play}(z_{\play}) + \hconj_{\play}(\dstate_{\play,\run}) + \curr[\step]\braket{\score_{\play,\run}}{\state_{\play,\run}} 
%	- \curr[\step]\braket{\signal_{\play,\run}}{z_{\play} - \state_{\play,\run}} + \frac{\curr[\step]^{2}}{2 \kappa_\hreg} \supnorm{\signal_{\play,\run}}^{2} \\
	&\leq \fench_{\play,\run}
		- \curr[\step]\braket{\signal_{\play,\run}}{z_{\play} - \state_{\play,\run}} + \frac{\curr[\step]^{2}}{2 \kappa_\hreg} \supnorm{\signal_{\play,\run}}^{2}.
\end{align}
where, in obvious notation, we are identifying $z_{\play} \in \pures_{\play}$ with the corresponding vertex $\bvec_{z_{\play}}$ of $\strats_{\play} = \simplex(\pures_{\play})$.
%\cref{eq:thm1:step2:strong convexity} is obtained by strong convexity \VB{strong ... what?} of $\hconj_{\play}(\cdot)$, namely $\hconj_{\play}(\score_{\play} + \step \payv_{\play}) \leq \hconj_{\play}(\score_{\play}) + \step \braket{\mirrorMap_{\play}(\score_{\play})}{\payv_{\play}} + \tfrac{\step^{2}}{2\kappa_\hreg}\supnorm{\payv}^{2}$.
To proceed, the main idea will be to relate $\curr[\step] \braket{\signal_{\play,\run}}{z_{\play} - \state_{\play,\run}}$ to its ``perfect'' counterpart $\curr[\step]\braket{\payv_{\play}(\curr)}{z_{\play} - \state_{\play,\run}}$.
We formalize this below.
\medskip

\begin{proofstep}
\label{prf:drift}
	If $\limset(\state)$ is not resilient to $(\play, z_\play)$,  there exists $\run_1 \geq \run_0$ such that, with probability $\alt\eps / 2 > 0$, and for all $\run \geq \run_1$, we have
\begin{equation}
	\label{eq:thm1:step2}
	\fench_{\play,\run} \leq \fench_{\play,\run_{0}} - \frac {\const}2 \sum_{\runalt=\run_0}^{\run} \curralt[\step].
\end{equation}
\end{proofstep}

\begin{proof}[Proof of \cref{prf:drift}]
With probability $\alt\eps$ and for all $\run \geq \run_0$, we have
\begin{align}
	\curr[\step] \braket{\signal_{\play,\run}}{z_{\play} - \state_{\play,\run}}
	&=
	\curr[\step] \braket{\payv_{\play}(\curr)}{z_{\play} - \state_{\play,\run}} + \curr[\step] \braket{\noise_{\play,\run}}{z_{\play} - \state_{\play,\run}} + \curr[\step] \braket{\bias_{\play,\run}}{z_{\play} - \state_{\play,\run}}
	\notag\\
	& \geq \label{eq:thm1:step2:estimate with c}
	\bracks[\big]{\const + \braket{\noise_{\play,\run}}{z_{\play} - \state_{\play,\run}} + \braket{\bias_{\play,\run}}{z_{\play} - \state_{\play,\run}}} \curr[\step].
\end{align}
The combination of \cref{eq:thm1:step2:fenchel rec,eq:thm1:step2:estimate with c} then provides the following upper bound of $\fench_{\play,\run+1}$:
\begin{alignat}{2}
	\fench_{\play,\run+1}
	&\leq \fench_{\play,\run}
		&&- \const \curr[\step]
		+ \curr[\step] \braket{\noise_{\play,\run}}{z_{\play} - \state_{\play,\run}}
		+ \curr[\step] \braket{\bias_{\play,\run}}{z_{\play} - \state_{\play,\run}}
		+ \frac{\curr[\step]^{2}}{2\kappa_\hreg}\supnorm{\signal_{\play,\run}}^{2}
	\notag\\
%	&\leq \cdots \\
	&\label{eq:thm1:step2:main inequality}
	\leq
	\fench_{\play,\run_{0}} 
		&&- \const \sum_{\runalt=\run_{0}}^{\run} \curralt[\step] 
		+ \sum_{\runalt=\run_0}^{\run} \frac{\supnorm{\signal_{\runalt,\play}}^{2}}{2\kappa_\hreg} \curralt[\step]^{2}
	\notag\\
	&\quad
		&&+ \underbrace{\sum_{\runalt=\run_0}^{\run} \curralt[\step] \braket{\noise_{\runalt,\play}}{z_{\play} - \state_{\runalt,\play}}}_{E_{\noise, \run}}
		+ \underbrace{\sum_{\runalt=\run_0}^{\run} \curralt[\step] \braket{\bias_{\runalt, \play}}{z_{\play} - \state_{\runalt,\play}}}_{E_{\bias, \run}}.
\end{alignat}

We are thus left to show is that $\const \sum_{\runalt=\run_0}^{\run} \curralt[\step]$ is the dominant term above.
To do so, we proceed to examine each term individually:
\begin{itemize}
\addtolength{\itemsep}{1ex}
\item
\emph{Second-order term:}
We first deal with the second-order term $\sum_{\runalt=\run_0}^{\run} \frac{\supnorm{\signal_{\runalt,\play}}^{2}}{2\kappa_\hreg} \curralt[\step]^{2}$.
By expanding the $\supnorm{\signal_{\runalt,\play}}^{2}$, we readily get
%By the triangle inequality, we have $\supnorm{\signal_{\runalt,\play}} \leq \supnorm{\payv_{\play}} + \curralt[\bbound] + \curralt[\sdev] = \bigoh(1) + \curralt[\sdev]$, so 
\begin{equation}
\frac
	{\sum_{\runalt=\run_{0}}^{\run} \supnorm{\signal_{\runalt,\play}}^{2} \curralt[\step]^{2}}
	{\curr[\efftime]}
	= \bigoh\parens*{
		\frac
			{\sum_{\runalt=\start}^{\run} \iter[\step]^{2}(1 + \iter[\bbound]^{2} + \iter[\sdev]^{2})}
			{\sum_{\runalt=\start}^{\run} \iter[\step]}}.
%	\bigoh\parens*{\sum_{\runalt=\run_0}^{\run} \curralt[\step]^{2}} 
%	+ \bigoh\parens*{\sum_{\runalt=\run_0}^{\run} \curralt[\sdev] \curralt[\step]^{2}}
%	+ \frac{1}{2\kappa_\hreg}\sum_{\runalt=\run_0}^{\run} \curralt[\sdev]^{2} \curralt[\step]^{2}.
\end{equation}
However, by our assumptions on the parameters of \eqref{eq:method}, we readily get
\begin{equation}
%\lim_{\run\to\infty}
%	\frac
%		{\sum_{\runalt=\start}^{\run} \iter[\step]^{2}(1 + \iter[\bbound]^{2} + \iter[\sdev]^{2})}
%		{\sum_{\runalt=\start}^{\run} \iter[\step]}
\lim_{\run\to\infty}
	\frac{\curr[\step]^{2}(1 + \curr[\bbound]^{2} + \curr[\sdev]^{2})}{\curr[\step]}
	= 0
\end{equation}
so we conclude that
\begin{equation}
\label{eq:thm1:step2:second order bound}
\lim_{\run\to\infty}
	\frac
		{\sum_{\runalt=\start}^{\run} \iter[\step]^{2}(1 + \iter[\bbound]^{2} + \iter[\sdev]^{2})}
		{\sum_{\runalt=\start}^{\run} \iter[\step]}
	= 0
\end{equation}
by the Stolz-Cesàro theorem.

\item
\emph{Bias term:}
By far the most immediate, the bias term $E_{\bias, \run}$ is bounded as
\begin{equation}
	\label{eq:thm1:step2:bias bound}
	E_{\bias, \run} \leq 2 \sum_{\runalt=\run_0}^{\run} \supnorm{\bias_{\play,\run}} \curralt[\step] \leq 2 \sum_{\runalt=\run_0}^{\run} \bbound_\runalt \curralt[\step] = \littleoh\parens*{\sum_{\runalt=\run_0}^{\run} \curralt[\step]}
	\quad
	\text{as $\run \to \infty$}.
\end{equation}

\item 
\emph{Noise term:}
Finally, the noise term $E_{\noise, \run}$ is bounded by means of the Azuma-Hoeffding inequality, \cf \cref{lem:Azuma} in \cref{app:aux}.
Specifically, with probability at least $1 - \alt\eps/2$, we have 
\begin{align}
E_{\noise, \run}
	&\defeq \sum_{\runalt=\run_0}^{\run} \curralt[\step] \braket{\noise_{\runalt, \play}}{z_{\play} - \state_{\runalt,\play}}
	\notag\\
	&\leq 2~\parens*{\sum_{\runalt=\run_0}^{\run} \supnorm{\noise_{\runalt,\play}}^{2} \curralt[\step]^{2} }^{1/2} \sqrt{2\log\parens*{\tfrac{4\run^{2}}{\alt\eps}}}
	\notag\\
	&\leq 2~\parens*{\sum_{\runalt=\run_0}^{\run} \curralt[\sdev]^{2} \curralt[\step]^{2} }^{1/2} \sqrt{2\log\parens*{\tfrac{4\run^{2}}{\alt\eps}}}.
\end{align}
for all $\run \geq \run_0$.
To proceed, note that a second application of the Stolz-Cesàro theorem yields $\sum_{\runalt=\run_0}^{\run} \curralt[\sdev]^{2} \curralt[\step]^{2} = \littleoh\parens{\sum_{\runalt=\run_0}^{\run} \curralt[\step]}$
and, moreover, note that
$\log(4\run^{2}/\alt\eps) = \bigoh\parens{\sum_{\runalt=\run_0}^{\run} \curralt[\step]}$.
Taking square roots and multiplying then yields that
\begin{equation}
\label{eq:thm1:step2:noise bound}
E_{\noise, \run}
	=
	\littleoh \parens*{\sum_{\runalt=\run_0}^{\run} \curralt[\step]}
\end{equation}
with probability at least $1 - \alt\eps/2$.
\end{itemize}

We are now in a position to establish the bound \cref{eq:thm1:step2}.
Indeed, putting \cref{eq:thm1:step2:second order bound,eq:thm1:step2:bias bound,eq:thm1:step2:noise bound} together, we readily infer that there exists $\run_1 \geq \run_0$ such that, with probability at least $1 - \alt\eps/2$, we have
\begin{equation}
\sum_{\runalt=\run_0}^{\run}
	\curralt[\step] \braket{\noise_{\runalt,\play}}{z_{\play} - \state_{\runalt,\play}}
		+ \sum_{\runalt=\run_0}^{\run} \curralt[\step] \braket{\bias_{\runalt, \play}}{z_{\play} - \state_{\runalt,\play}}
		+ \sum_{\runalt=\run_0}^{\run} \frac{\supnorm{\signal_{\runalt,\play}}^{2}}{2\kappa_\hreg} \curralt[\step]^{2} \leq \frac{\const}2 \sum_{\runalt=\run_0}^{\run} \curralt[\step]
\end{equation}
for all $\run \geq \run_1$.
This proves \cref{eq:thm1:step2} and concludes our proof.
\end{proof}

Summarizing the above, we have shown that, with probability at least $1 -  \alt\eps/2$, we have
\begin{equation}
\fench_{\play,\run+1} \leq F_{\run_0} - \frac {\const}2 \sum_{\runalt=\run_0}^{\run} \step_\runalt
	\to -\infty
	\quad
	\text{as $\run\to\infty$}.
\end{equation}
Since $\fench$ is nonnegative (by \cref{lem:Fench}), we have established that the event where $\limset(\state)$ is not resilient to $(\play, z_\play)$ is an event of probability zero.
However, since there are uncountably many strategic deviations, the proof is not yet complete;
the last step involves an approximation by deviations with \emph{rational} entries.

\begin{proofstep}
	\label{proof:resilience:step3}
	$\limset(\state)$ is almost-surely resilient.
\end{proofstep}
\begin{proof}[Proof of \cref{proof:resilience:step3}]
The key point of the proof is the observation that a closed set is resilient if and only if it is \emph{rationally} resilient, \ie it nullifies all \emph{rational} deviations $z_\play \in \strats_\play \cap \Q^{\pures_\play}$ (which are countably many).
Indeed, if $\limset(\state)$ is not resilient with positive probability, then, likewise, $\limset(\state)$ will not be rationally resilient with positive probability either.
Because there are countably many rational deviations, there must be a rational strategic deviation $(\play, z_\play)$ (with $z_\play \in \strats_\play \cap \Q^{\pures_\play}$) to which $\limset(\state)$ is not resilient. 
This comes in contradiction with the conclusions of \cref{prf:drift}.
\end{proof}

This concludes the last required step, so the proof of \cref{thm:resilience} is now complete.
\end{proof}

%----------------------------------------------------------------------
%%% APP: CLOSED
%----------------------------------------------------------------------
\section{Proof of \cref{thm:closed,thm:rate}}
\label{app:closed}
%----------------------------------------------------------------------
%%% APP: SHARP
%----------------------------------------------------------------------
% !TEX root = ../Main.tex

In this last appendix, our goal is to prove our characterization of \ac{closed} sets, namely:

\CLOSED*

\RATE*

Our proof strategy will be to construct a sheaf of ``linearized'' energy functions which, when bundled together, yield a suitable Lyapunov-like function for $\set$.
To do so, let $\pureset = \prod_{\play} \pureset_{\play}$ denote the support of $\set$ (\cf the definition of \ac{closed} sets), and let
\begin{equation}
\tvecs_{\play}
	= \setdef
		{\bvec_{\play\purealt_{\play}} - \bvec_{\play\pure_{\play}}}
		{\pure_{\play}\in\pureset_{\play}, \purealt_{\play}\in\pures_{\play}\setminus\pureset_{\play}}
\end{equation}
and
\begin{equation}
\tvecs
	= \union\nolimits_{\play\in\players} \tvecs_{\play}
%	= \setdef{\bvec_{\pure} - \bvec_{\pureq}}{\pureq\in\pureset, \pure\in\pures\setminus\pureset}
\end{equation}
denote the set of all pure strategic deviations from $\set$.
Then, our ensemble of candidate energy functions will be given by
\begin{equation}
\label{eq:energy-sharp}
\energy_{\tvec}(\score)
	= \braket{\score}{\tvec}
	\qquad
	\text{for $\tvec\in\tvecs$, $\score\in\dspace$}.
\end{equation}
The motivation for this definition is given by the following lemma.

\begin{lemma}
\label{lem:energy-bundle}
Suppose that the sequence $\curr[\score]\in\dspace$, $\run=\running$, has $\energy_{\tvec}(\curr[\score]) \to -\infty$ for all $\tvec\in\tvecs$ as $\run\to\infty$.
Then the sequence $\curr[\strat] = \mirror(\curr[\score])$ converges to $\set$ as $\run\to\infty$.
\end{lemma}

\begin{proof}
Let $\tvec = \bvec_{\play\purealt_{\play}} - \bvec_{\play\pure_{\play}}$ for some $\play\in\players$, $\pure_{\play}\in\pureset_{\play}$, and $\purealt_{\play} \in \pures_{\play}\setminus\pureset_{\play}$.
Since $\energy_{\tvec}(\curr[\score]) \to -\infty$ by assumption, we get $\score_{\play\purealt_{\play},\run} - \score_{\play\pure_{\play},\run} \to -\infty$ and hence, by \cref{lem:mirror}, we conclude that $\mirror_{\play\purealt_{\play}}(\curr[\strat]) \to 0$ as $\run\to\infty$.
In turn, given that this holds for all $\play\in\players$ and all $\purealt_{\play} \in \pures_{\play}\setminus\pureset_{\play}$, we conclude that $\curr[\strat] = \mirror(\curr[\score])$ converges to $\set$.
\end{proof}

In view of the above, we will focus on showing that $\energy_{\tvec}(\curr[\dstate]) \to -\infty$ for all $\tvec\in\tvecs$.
As a first step, we establish a basic template inequality for the evolution of $\energy_{\tvec}$ under \eqref{eq:method}.
\smallskip

\begin{lemma}
\label{lem:energy-sharp}
Fix some $\tvec\in\tvecs$ and let $\curr[\energy] \defeq \energy_{\tvec}(\curr[\dstate])$.
Then, for all $\run=\running$, we have
\begin{equation}
\label{eq:template-sharp}
\next[\energy]
	\leq \curr[\energy]
		+ \curr[\step] \braket{\vecfield(\curr)}{\tvec}
		+ \curr[\step] \curr[\snoise]
		+ \curr[\step] \curr[\sbias]
%		+ \curr[\step]^{2} \curr[\second]^{2}
\end{equation}
where the error terms $\curr[\snoise]$ and $\curr[\sbias]$ are given by
\begin{equation}
\label{eq:errors}
\curr[\snoise]
	= \braket{\curr[\noise]}{\tvec}
%	\\
	\qquad
	\text{and}
	\qquad
\curr[\sbias]
	= \ebound \curr[\bbound].
\end{equation}
\end{lemma}

\begin{proof}
Simply set $\score \gets \next[\dstate]$ in $\energy_{\tvec}(\score)$,
invoke the definition of the update $\curr[\dstate] \gets \next[\dstate]$ in \eqref{eq:method},
and
note that $\abs{\braket{\curr[\bias]}{\tvec}} \leq \norm{\tvec} \dnorm{\curr[\bias]} \leq 2 \curr[\bbound]$ by the definition of $\tvecs$.
\end{proof}

The key take-away from \eqref{eq:template-sharp} is that, if $\curr$ is close to $\set$ and $\pure_{\play}\in\pureset_{\play}$, $\purealt_{\play} \in \pures_{\play} \setminus \pureset_{\play}$, we will have
\begin{equation}
\braket{\payv(\curr)}{\tvec}
	= \payv_{\play\purealt_{\play}}(\curr)
		- \payv_{\play\pure_{\play}}(\curr)
	= \pay_{\play}(\purealt_{\play};\state_{-\play,\run})
		- \pay_{\play}(\pure_{\play};\state_{-\play,\run})
	< 0
\end{equation}
by the continuity of $\pay_{\play}$ and the assumption that $\set$ is a \ac{closed} set.
More concretely, by the definition of the better-reply correspondence, we have
\begin{equation}
\braket{\payv(\eq)}{\tvec}
	< 0
	\quad
	\text{for all $\eq\in\set$ and all $\tvec\in\tvecs$}
\end{equation}
and hence, by continuity, there exists a neighborhood $\basin$ of $\set$ such that
\begin{equation}
\braket{\payv(\strat)}{\tvec}
	< 0
	\quad
	\text{for all $\strat\in\basin$ and all $\tvec\in\tvecs$}.
\end{equation}
In other words, as long as $\curr$ is sufficiently close to $\set$, \eqref{eq:template-sharp} exhibits a consistent negative drift pushing $\curr[\energy]$ towards $-\infty$.

To exploit this ``dynamic consistency'' property of $\set$, it will be convenient to introduce the family of sets
\begin{equation}
\label{eq:dset}
\dbasin(\thres)
	= \setdef
		{\score\in\dspace}
		{\braket{\score}{\tvec} < -\thres \text{ for all } \tvec\in\tvecs}
%\dbasin_{\play}(\thres)
%	= \setdef
%		{\score_{\play}\in\R^{\pures_{\play}}}
%		{\score_{\play\pureq_{\play}} - \score_{\play\pure_{\play}} > \thres}
\end{equation}
As we show below, these sets are mapped under $\mirror$ to neighborhoods of $\set$, so they are particularly well-suited to serve as initialization domains for \eqref{eq:method}.
This is encoded in the following properties:

\begin{lemma}
\label{lem:score2strat}
Let $\strat = \mirror(\score)$ for some $\score\in\dspace$.
Then, for all $\pure_{\play},\purealt_{\play}$, $\play\in\players$, we have
\begin{equation}
\label{eq:score2strat}
\strat_{\play\pure_{\play}}
	\leq \rate\parens*{\hker(1^{-}) + \score_{\play\purealt_{\play}} - \score_{\play\pure_{\play}}}
\end{equation}
with $\rate$ defined as per \cref{thm:rate}, \ie
\begin{equation}
\label{eq:hrate}
\rate(z)
	= \begin{dcases*}
		0
			&if $z \leq \hker'(0^{+})$,
			\\
		(\hker')^{-1}(z)
			&if $\hker'(0^{+}) < z < \hker'(1^{-})$,
			\\
		1
			&if $z \geq \hker'(1^{-})$.
	\end{dcases*}
\end{equation}
\end{lemma}

\begin{corollary}
\label{cor:score2strat}
For all $\delta>0$ there exists some $\thres_{\delta}\in\R$ such that, for all $\thres > \thres_{\delta}$ and all $\score \in \dbasin_{\thres}$, we have
\begin{equation}
\mirror_{\play\purealt_{\play}}(\score_{\play})
	< \delta
	\quad
	\text{for all $\purealt_{\play}\in\pures_{\play} \setminus \pureset_{\play}$ and all $\play\in\players$}.
\end{equation}
\end{corollary}

\begin{proof}[Proof of \cref{lem:score2strat}]
Suppressing the player index for simplicity, the first-order stationarity conditions for the convex problem \eqref{eq:mirror} readily give
\begin{equation}
\label{eq:KKT}
\score_{\pure} - \hker'(\strat_{\pure})
	= \mu - \nu_{\pure},
\end{equation}
where $\mu$ is the Lagrange multiplier for the equality constraint $\sum_{\pure}\strat_{\pure}=1$, and $\nu_{\pure}$ is the complementary slackness multiplier of the inquality constraint $\strat_{\pure} \geq 0$ (so $\nu_{\pure} = 0$ whenever $\strat_{\pure} > 0$).
Thus, rewriting \eqref{eq:KKT} for some $\pure\in\pures$, we get
\begin{equation}
\label{eq:KKT-diff1}
\score_{\purealt} - \score_{\pure}
	= \hker'(\strat_{\purealt}) - \hker'(\strat_{\pure})
		+ \nu_{\pure} - \nu_{\purealt}
\end{equation}
and hence
%if $\strat_{\pure} > 0$ (which implies that $\nu_{\pure} = 0$), we obtain
\begin{equation}
\label{eq:KKT-diff2}
\hker'(\strat_{\purealt})
	= \hker'(\strat_{\pure})
		+ \nu_{\purealt} - \nu_{\pure}
		+ \score_{\purealt} - \score_{\pure}
	\leq \hker'(1^{-})
		+ \nu_{\purealt}
		+ \score_{\purealt} - \score_{\pure},
\end{equation}
where we used the fact that $\nu_{\pure} \geq 0$.
Now, if $\hker'(1^{-}) + \score_{\purealt} - \score_{\pure} < \hker'(0^{+})$ and $\strat_{\purealt} > 0$ (so $\nu_{\purealt} = 0$), we will have $\hker'(\strat_{\purealt}) < \hker'(0^{+})$, a contradiction.
This shows that $\strat_{\purealt} = 0$ if $\hker'(1^{-}) + \score_{\purealt} - \score_{\pure} < \hker'(0^{+})$, so \eqref{eq:score2strat} is satisfied in this case.
Otherwise, if $\strat_{\purealt} > 0$, we must have $\nu_{\purealt} = 0$ by complementary slackness, so \eqref{eq:score2strat} follows by applying the second branch of \eqref{eq:hrate} to \eqref{eq:KKT-diff2}.
\end{proof}

The above provides us with a fairly good handle on the local geometric and dynamic properties of $\set$.
On the flip side however, the various error terms in \eqref{eq:errors} may be positive, so $\curr[\energy]$ may fail to be decreasing and $\curr$ may drift away from $\set$.
On that account, it will be convenient to introduce the aggregate error processes
\begin{alignat}{4}
\label{eq:errors-agg}
\curr[\aggnoise]
	&= \sum_{\runalt=\start}^{\run} \iter[\step]\iter[\snoise]
	&\qquad
	\text{and}
	&\qquad
\curr[\aggbias]
	&= \sum_{\runalt=\start}^{\run} \iter[\step]\iter[\sbias].
%\end{equation}
%\shortintertext{and their corresponding ``worst-case'' counterparts}
%%\begin{equation}
%\label{eq:errors-sup}
%\supnoise
%	&= \sup\nolimits_{\run} \curr[\aggnoise]
%	&\qquad
%	\text{and}
%	&\qquad
%\supbias
%	&= \sup\nolimits_{\run} \curr[\aggbias].
\end{alignat}
Intuitively, the aggregates \eqref{eq:errors-agg} measure the total effect of each error term in \eqref{eq:template-sharp},
%while the maximal processes \eqref{eq:errors-sup} control the probability that a single ``bad'' realization of the noise may cause $\curr$ to exit the basin of attraction of $\set$, possibly never to return.
%In light of this,
so we will establish a first series of results under the following general requirements:
\begin{enumerate}
\addtolength{\itemsep}{1ex}
\item
\emph{Subleading error growth:}
\begin{subequations}
\makeatletter
\def\@currentlabel{Sub}
\makeatother
\label{eq:sub}
\renewcommand{\theequation}{Sub.\Roman{equation}}
\begin{align}
\label{eq:sub-noise}
\lim_{\run\to\infty} \curr[\aggnoise]/\curr[\efftime]
	&= 0
	\\[\smallskipamount]
\label{eq:sub-bias}
\lim_{\run\to\infty} \curr[\aggbias]/\curr[\efftime]
	&= 0
\end{align}
\end{subequations}
where $\curr[\efftime] = \sum_{\runalt=\start}^{\run} \iter[\step]$ and both limits are to be interpreted in the almost sure sense.

\item
\emph{Drift dominance:}
\begin{subequations}
\makeatletter
\def\@currentlabel{Dom}
\makeatother
\label{eq:dom}
\renewcommand{\theequation}{Dom.\Roman{equation}}
\label{eq:dom}
\begin{align}
\label{eq:dom-noise}
\probof{\curr[\aggnoise] \leq \Const\curr[\efftime]^{\texp}/2\;\;\text{for all $\run$}}
	&\geq 1-\conf
	\\[\smallskipamount]
\label{eq:dom-bias}
\probof{\curr[\aggbias] \leq \Const\curr[\efftime]^{\texp}/2\;\;\text{for all $\run$}}
	&\geq 1-\conf
\end{align}
\end{subequations}
for some $\Const>0$ and $\texp\in[0,1)$.
\end{enumerate}

In a nutshell, \eqref{eq:sub} posits that the aggregate error processes $\curr[\aggnoise]$ and $\curr[\aggbias]$ of \eqref{eq:errors-agg} are subleading relative to the long-run drift of \eqref{eq:template-sharp}, while \eqref{eq:dom} goes a step further and asks that said errors are asymptotically dominated by the drift in \eqref{eq:template-sharp}.
Accordingly, under these implicit error control conditions, we obtain the interim convergence result below:

\begin{proposition}
\label{lem:strict-local}
Let $\set$ be a \ac{closed} set,
fix some confidence threshold $\conf>0$,
and
let $\curr = \mirror(\curr[\dstate])$ be the sequence of play generated by \eqref{eq:method}.
If \eqref{eq:sub} and \eqref{eq:dom} hold, there exists an unbounded initialization domain $\dbasin\subseteq\dspace$ such that
\begin{equation}
\probof{\text{$\curr$ converges to $\set$} \given \init[\dstate]\in\dbasin}
	\geq 1-2\conf.
\end{equation}
\end{proposition}

\begin{proof}[Proof of \cref{lem:strict-local}]
Fix some $\tvec\in\tvecs$, let $\curr[\energy] = \energy_{\tvec}(\curr[\dstate])$, and pick $\texp\in[0,1)$ so that \eqref{eq:dom} holds for some $\Const > 0$.
In addition, set
$\const = -\sup_{\strat\in\basin} \braket{\payv(\strat)}{\tvec} > 0$,
let $\run_{0} = \inf\setdef{\run}{\const\curr[\efftime] > \Const\curr[\efftime]^{\texp}}$,
and
write $\ediff = \max_{\run}\braces{\Const\curr[\efftime]^{\texp} - \const\curr[\efftime]}$.
Then, if $\init[\dstate]$ is initialized in $\dbasin \gets \dbasin(\thres + \ediff)$ where $\thres$ is such that $\dbasin(\thres) \subseteq \basin$, we will have $\curr[\dstate]\in\dbasin(\thres)$ for all $\run$.
Indeed, this being trivially the case for $\run=1$, assume it to be the case for all $\runalt = \running,\run$.
Then, by \eqref{eq:template-sharp} and our inductive hypothesis, we get
\begin{align}
\label{eq:energy-agg2}
\next[\energy]
	&\leq \init[\energy]
		- \sum_{\runalt=\start}^{\run} \iter[\step] \braket{\payv(\iter)}{\tvec}
		+ \curr[\aggnoise]
		+ \curr[\aggbias]
	\notag\\
	&\leq - \thres
		- \ediff
		- \const \curr[\efftime]
		+ \Const \curr[\efftime]^{\texp}/2
		+ \Const \curr[\efftime]^{\texp}/2
	\notag\\
	&\leq -\thres
		- \ediff
		+ \ediff
	= -\thres
\end{align}
\ie $\next[\energy] \in \dbasin(\thres)$, as claimed.

Now, since $\curr[\energy] \in \dbasin(\thres)$ for all $\run$, we conclude that
\begin{equation}
\label{eq:energy-agg}
\next[\energy]
	\leq \init[\energy]
		- \const \curr[\efftime]
		+ \curr[\aggnoise]
		+ \curr[\aggbias]
		\quad
		\text{for all $\run = \running$}
\end{equation}
Thus, if \eqref{eq:sub} holds, we readily get $\curr[\energy] \to -\infty$ \acl{wp1} on the event that \eqref{eq:dom-noise} and \eqref{eq:dom-bias} both hold.
This implies that $\curr[\energy] \to -\infty$, and since $\tvec\in\tvecs$ above is arbitrary, we conclude that $\curr\to\set$ with probability at least $1-2\conf$, as claimed.
\end{proof}

We are now in a position to prove \cref{thm:closed}.

\begin{proof}[Proof of \cref{thm:closed}]
Our proof will hinge on showing that \eqref{eq:sub} and \eqref{eq:dom} hold under the stated step-size and sampling parameter schedules.
Our claim will then follow by a direct application of \cref{lem:strict-local} and a reduction to a suitable subface of $\strats$.

First, regarding \eqref{eq:sub}, the law of large numbers for martingale difference sequences \citep[Theorem 2.18]{HH80} shows that $\curr[\aggnoise] / \curr[\efftime] \to 0$ \acl{wp1} on the event
\begin{equation}
\braces*{\sum_{\run} \curr[\step]^{2} \exof{\curr[\snoise]^{2} \given \curr[\filter]} / \curr[\efftime]^{2} < \infty}.
\end{equation}
However
\begin{equation}
\exof{\curr[\snoise]^{2} \given \curr[\filter]}
	\leq \ebound^{2} \exof{\dnorm{\curr[\noise]}^{2} \given \curr[\filter]}
	\leq \ebound^{2} \curr[\sdev]^{2}
	= \bigoh(\run^{2\noisexp})
\end{equation}
so, in turn, we get
\begin{equation}
\sum_{\run} \frac{\curr[\step]^{2} \exof{\curr[\snoise]^{2} \given \curr[\filter]}}{\curr[\efftime]^{2}}
	= \bigoh\parens*{ \sum_{\run} \frac{\curr[\step]^{2} \curr[\sdev]^{2}}{\curr[\efftime]^{2}} }
	= \bigoh\parens*{ \sum_{\run} \frac{\run^{-2\stepexp} \run^{2\noisexp}}{\run^{2(1-\stepexp)}} }
	= \bigoh\parens*{ \sum_{\run} \frac{1}{\run^{2 - 2\noisexp}} }
	< \infty
\end{equation}
given that $\noisexp < 1/2$.
This establishes \eqref{eq:sub-noise};
the remaining requirement \eqref{eq:sub-bias} follows trivially by noting that $\sum_{\runalt=\start}^{\run} \iter[\step] \iter[\bbound] \big/ \sum_{\runalt=\start}^{\run} \iter[\step] \to 0$ if and only if $\curr[\bbound] \to 0$, which is immediate from the theorem's assumptions.
%This shows that \eqref{eq:sub} holds, so we are left to show that \eqref{eq:dom}.

Second, regarding \eqref{eq:dom}, since $\curr[\bbound]$ is deterministic and $\curr[\bbound] = \bigoh(1/\run^{\biasexp})$ for some $\biasexp>0$, it is always possible to find $\Const>0$ and $\texp\in(0,1)$ so that \eqref{eq:dom-bias} holds.
We are thus left to establish \eqref{eq:dom-noise}.
To that end, let $\curr[\maxnoise] = \sup_{1\leq\runalt\leq\run} \abs{\curr[\aggnoise]}$ and set $\curr[P] \defeq \probof*{\curr[\maxnoise] > \Const\curr[\efftime]^{\texp}/2}$ so
\begin{equation}
\label{eq:dom1}
%\probof*{\curr[\maxnoise] \leq \Const\curr[\efftime]^{\texp}/2}
\curr[P]
	\leq \frac{\exof{\abs{\curr[\aggnoise]}^{\qexp}}}{(\Const/2)^{\qexp}\curr[\efftime]^{\texp\qexp}}
	\leq \const_{\qexp}
		\frac{\exof{\parens*{\sum_{\runalt=\start}^{\run} \iter[\step]^{2} \dnorm{\iter[\noise]}^{2}}^{\qexp/2}}}{\curr[\efftime]^{\texp\qexp}}
\end{equation}
where $c_{\qexp}$ is a positive constant depending only on $\Const$ and $\qexp$, and we used Kolmogorov's inequality (\cref{lem:Kolmogorov}) in the first step and the \acl{BDG} inequality (\cref{lem:Burkholder}) in the second.

To proceed, we will require the following variant of Hölder's inequality \citep[p.~15]{Ben99}:
\begin{equation}
\label{eq:Holder}
\parens*{\sum_{\runalt=\start}^{\run} \iter[a]\iter[b]}^{\rho}
	\leq \parens*{\sum_{\runalt=\start}^{\run} \iter[a]^{\frac{\lambda\rho}{\rho-1}}}^{\rho-1}
		\sum_{\runalt=\start}^{\run} \iter[a]^{(1-\lambda)\rho} \iter[b]^{\rho}
\end{equation}
valid for all $\iter[a],\iter[b] \geq 0$ and all $\rho>1$, $\lambda\in[0,1)$.
Then, substituting $\iter[a] \gets \iter[\step]^{2}$, $\iter[b] \gets \dnorm{\iter[\noise]}^{2}$, $\rho \gets \qexp/2$ and $\lambda \gets 1/2 - 1/\qexp$, \eqref{eq:dom1} gives
\begin{align}
\label{eq:dom2}
%\probof*{\curr[\maxnoise] \leq \Const\curr[\efftime]^{\texp}}
\curr[P]
	\leq \const_{\qexp}
		\frac{
		\parens*{\sum_{\runalt=\start}^{\run} \iter[\step]}^{\qexp/2 - 1}
			\sum_{\runalt=\start}^{\run} \iter[\step]^{1+\qexp/2} \exof{\dnorm{\iter[\noise]}^{\qexp}}}
			{\curr[\efftime]^{\texp\qexp}}
%	\notag\\
	\leq \const_{\qexp}
		\frac{\sum_{\runalt=\start}^{\run} \iter[\step]^{1+\qexp/2} \iter[\sdev]^{\qexp}}
			{\curr[\efftime]^{1+(\texp-1/2)\qexp}}
\end{align}

We now consider two cases, depending on whether the numerator of \eqref{eq:dom2} is summable or not.
\begin{enumerate}
[left=0pt,label={\bfseries Case \arabic*:}]
\item
$\stepexp(1+\qexp/2) \geq 1 + \qexp\noisexp$.
In this case, the numerator of \eqref{eq:dom2} is summable under the theorem's assumptions, so the fraction in \eqref{eq:dom2} behaves as $\bigoh(1/\run^{(1-\stepexp) (1 + (\texp-1/2)\qexp)})$.
\item
$\stepexp(1+\qexp/2) < 1 + \qexp\noisexp$.
In this case, the numerator of \eqref{eq:dom2} is not summable under the theorem's assumptions, so the fraction in \eqref{eq:dom2} behaves as $\bigoh\parens*{\run^{1-\stepexp(1+\qexp/2) + \qexp\noisexp} \big/ \run^{(1-\stepexp) (1 + (\texp-1/2)\qexp)}}$.
\end{enumerate}
Thus, working out the various exponents, a tedious \textendash\ but otherwise straightforward \textendash\ calculation shows that there exists some $\texp\in(0,1)$ such that $\curr[P]$ is summable as long as $\noisexp < 1/2 - 1/\qexp$ and $0 \leq \stepexp < \qexp/(2+\qexp)$.
Hence, if $\step$ is sufficiently small relative to $\conf$, we conclude that
\begin{equation}
\txs
\probof{\curr[\aggnoise] \leq \Const\curr[\efftime]^{\texp}/2 \; \text{for all $\run$}}
	\geq 1 - \sum_{\run} \curr[P]
	\geq 1 - \conf/2.
\end{equation}
Finally, if $\stepexp>1/2+\noisexp$, \eqref{eq:dom-noise} is a straightforward consequence of \eqref{eq:dom1} for $\qexp=2$.

With all this in hand, the final steps of our proof proceed as follows:

\para{Closedness $\implies$ Stability}

Our assertion follows by invoking \cref{lem:strict-local}.
%\hfill
%\qedsymbol

\para{Stability $\implies$ Closedness}

Suppose that $\set$ is not \ac{closed}.
Then there exists some pure strategy $\pure\in\pureset$ and some deviation $\purealt\not\in\pureset$ such that the deviation from $\pure$ to $\purealt$ is not costly to the deviating player.
Thus, if we consider the restriction of the game to the face spanned by $\pure$ and $\purealt$ (a single-player game with two strategies), the corresponding score difference will be
\begin{equation}
\score_{\purealt,\run} - \score_{\pure,\run}
	\geq \sum_{\runalt=\start} \iter[\step] \iter[\bias]
		+ \sum_{\runalt=\start} \iter[\step] \iter[\noise]
\end{equation}
By our standing assumptions for $\curr[\bias]$ and $\curr[\noise]$ (and Doob's martingale convergence theorem for the latter), both $\sum_{\runalt=\start} \iter[\step] \iter[\bias]$ and $\sum_{\runalt=\start} \iter[\step] \iter[\noise]$ will be bounded from below by some \as finite random variable $A_{0}$.
Since $\hker$ is steep, it follows that, \acl{wp1}, $\liminf_{\run\to\infty}(\score_{\pure,\run}) > 0$, so $\pureset$ cannot be stable.
%\hfill
%\qedsymbol

\para{Minimality $\implies$ Irreducible Stability}

Suppose that $\set$ is \ac{minimal}.
Then, by our previous claim, $\set$ is stochastically asymptotically stable.
If $\set$ contains a proper subface $\alt\set \subsetneq \set$ that is also stochastically asymptotically stable, $\alt\set$ must be \ac{closed} by the converse implication of the first part of the theorem.
However, in that case, $\set$ would not be \ac{minimal}, a contradiction which proves our claim.
%\hfill
%\qedsymbol

\para{Irreducible Stability $\implies$ Minimality}

For our last claim, assume that $\set$ is irreducibly stable.
By the first part of our theorem, this implies that $\set$ is \ac{closed}.
Then, if it so happens that $\set$ is not \ac{minimal}, it would contain a proper \ac{closed} subface $\alt\set \subsetneq \set$;
by the first part of our theorem, this set would be itself stochastically asymptotically stable, in contradiction to the irreducibility assumption.
This shows that $\set$ is \ac{minimal} and concludes our proof.
\end{proof}

We are only left to establish the convergence rate estimate of \cref{thm:rate}.

\begin{proof}[Proof of \cref{thm:rate}]
Going back to \eqref{eq:energy-agg} and invoking \cref{lem:score2strat} shows that there exist constants $\const_{1} > 0$ and $\const_{2} \in \R$ such that, for all $\pure_{\play} \in \pures_{\play} \setminus \pureset_{\play}$, $\play\in\players$, we have
\begin{equation}
\state_{\play\pure_{\play},\run}
	\leq \rate\parens*{\hker(1^{-}) + \curr[\energy]}
	\leq \rate(\const_{2} - \const_{1} \curr[\tau])
\end{equation}
\acl{wp1} on the events of \eqref{eq:dom}.
We thus get
\begin{equation}
\dist_{1}(\curr,\set)
	\leq \sum_{\play\in\players} \sum_{\pure_{\play} \in \pures_{\play} \setminus \pureset_{\play}} \rate(\const_{2} - \const_{1} \curr[\efftime]),
\end{equation}
and our proof is complete.
\end{proof}

As for the rate estimates of \cref{cor:rate}, the proof boils down to a simple derivation of the corresponding rate functions:

\begin{proof}[Proof of \cref{cor:rate}]
By a straightforward calculation, we have:
\begin{enumerate}
\item
If $\hker(z) = z\log z$ then $\rate(z) = \exp(1+z)$.
\item
If $\hker(z) = -4\sqrt{z}$ then $\rate(z) = 4/z^{2}$.
\item
If $\hker(z) = z^{2}/2$ then $\rate(z) = \clip{z}_{0}^{1}$.
\end{enumerate}
Our claims then follow immediatly from the rate estimate \eqref{eq:rate} of \cref{thm:closed}.
\end{proof}

%----------------------------------------------------------------------
%%% APP: NUMERICS
%----------------------------------------------------------------------
\section{Numerical experiments}
\label{app:numerics}
%----------------------------------------------------------------------
%%% APP: NUMERICS
%----------------------------------------------------------------------
% !TEX root = ../Main.tex

In all our experiments, we ran the \ac{EXP3} variant of \ac{BFTRL} (\cf \cref{alg:BFTRL}) with step-size and sampling radius parameters $\curr[\step] = 0.2 \times \run^{-1/2}$ and $\curr[\mix] = 0.1 \times \run^{-0.15}$ respectively.
The algorithm was run for $\nRuns = 10^{4}$ iterations and, to reduce graphical clutter, we plotted only every third point of each trajectory.
Trajectories have been colored throughout with darker hues indicating later times (\eg light blue indicates that the trajectory is closer in time to its starting point, darker shades of blue indicate proximity to the termination time).
The algorithm's initial conditions were taken from a uniform initialization grid of the form $\init[\score] \in \{-1,0,1\}^{3}$ and perturbed by a uniform random number in $[-0.1,-0.1]$ to avoid non-generic initializations.

%----------------------------------------------------------------------
%% Convergence figure begins here

\begin{figure}[t]
\centering
\footnotesize
\begin{subfigure}{.48\textwidth}
\includegraphics[height=.48\textwidth]{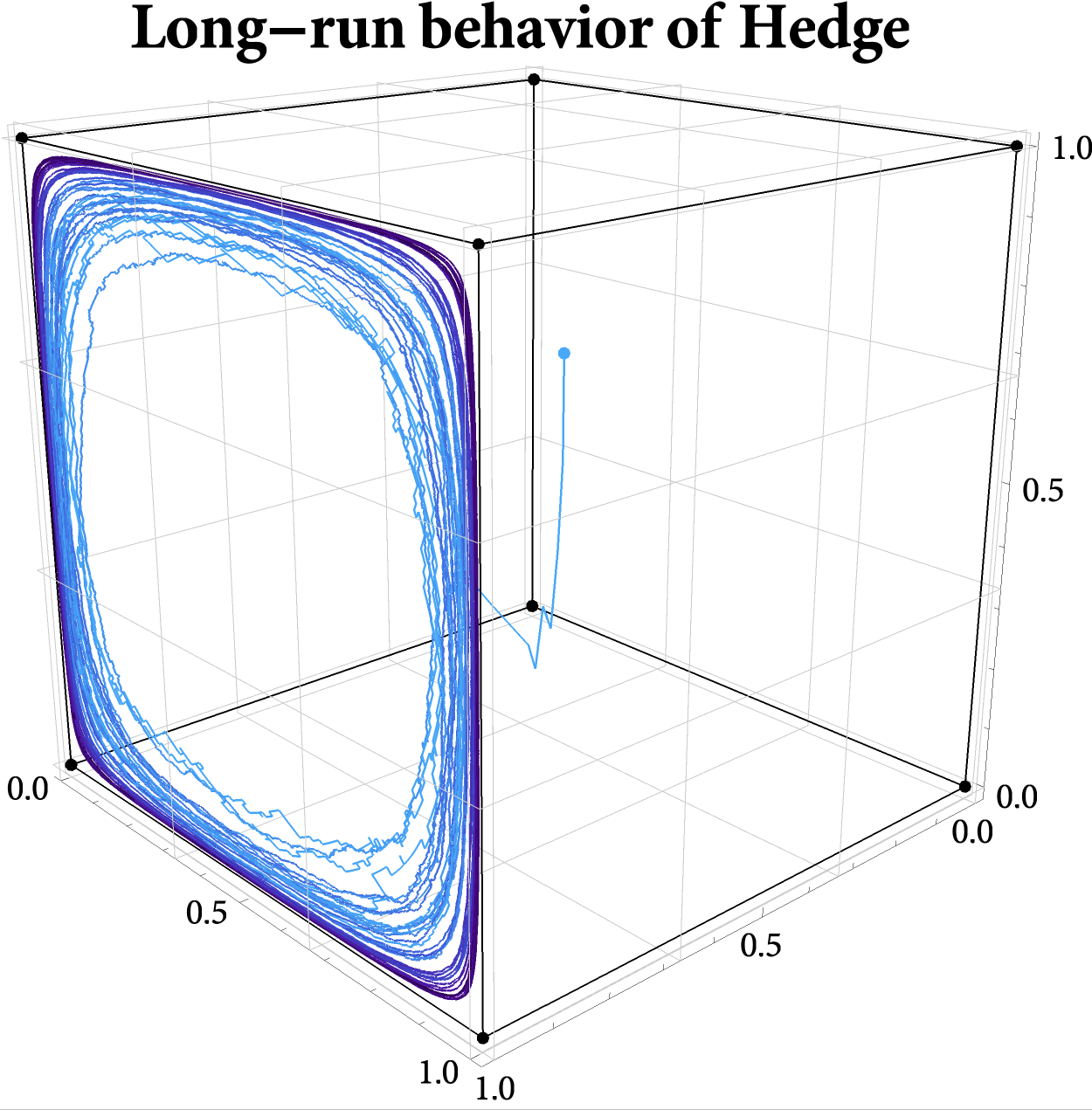}
%\;\;
\includegraphics[height=.48\textwidth]{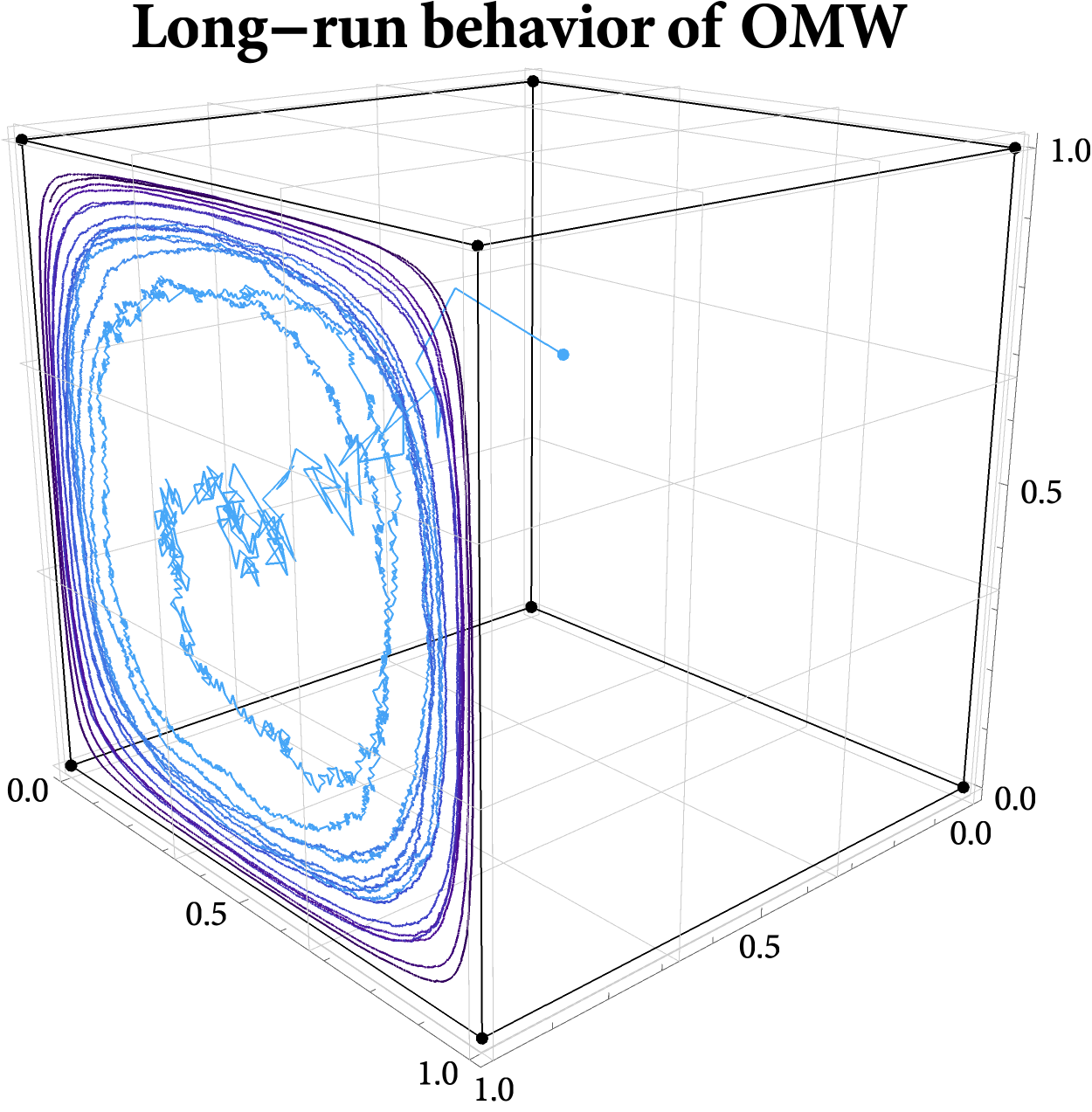}
%\caption{Oracle-based methods (\cref{alg:EW,alg:OMW}).}
\end{subfigure}
\hfill
\begin{subfigure}{.48\textwidth}
\includegraphics[height=.48\textwidth]{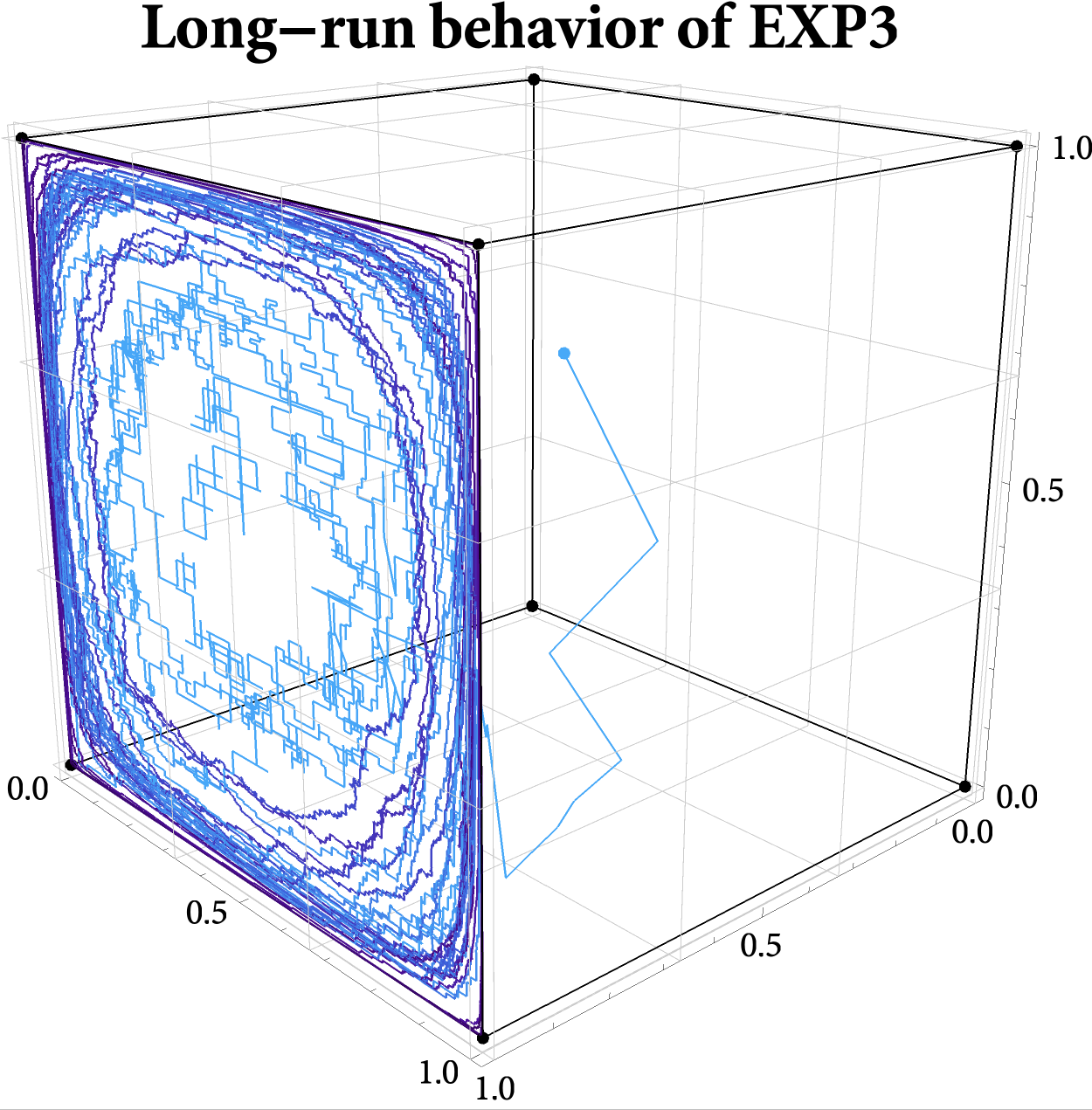}
%\;\;
\includegraphics[height=.48\textwidth]{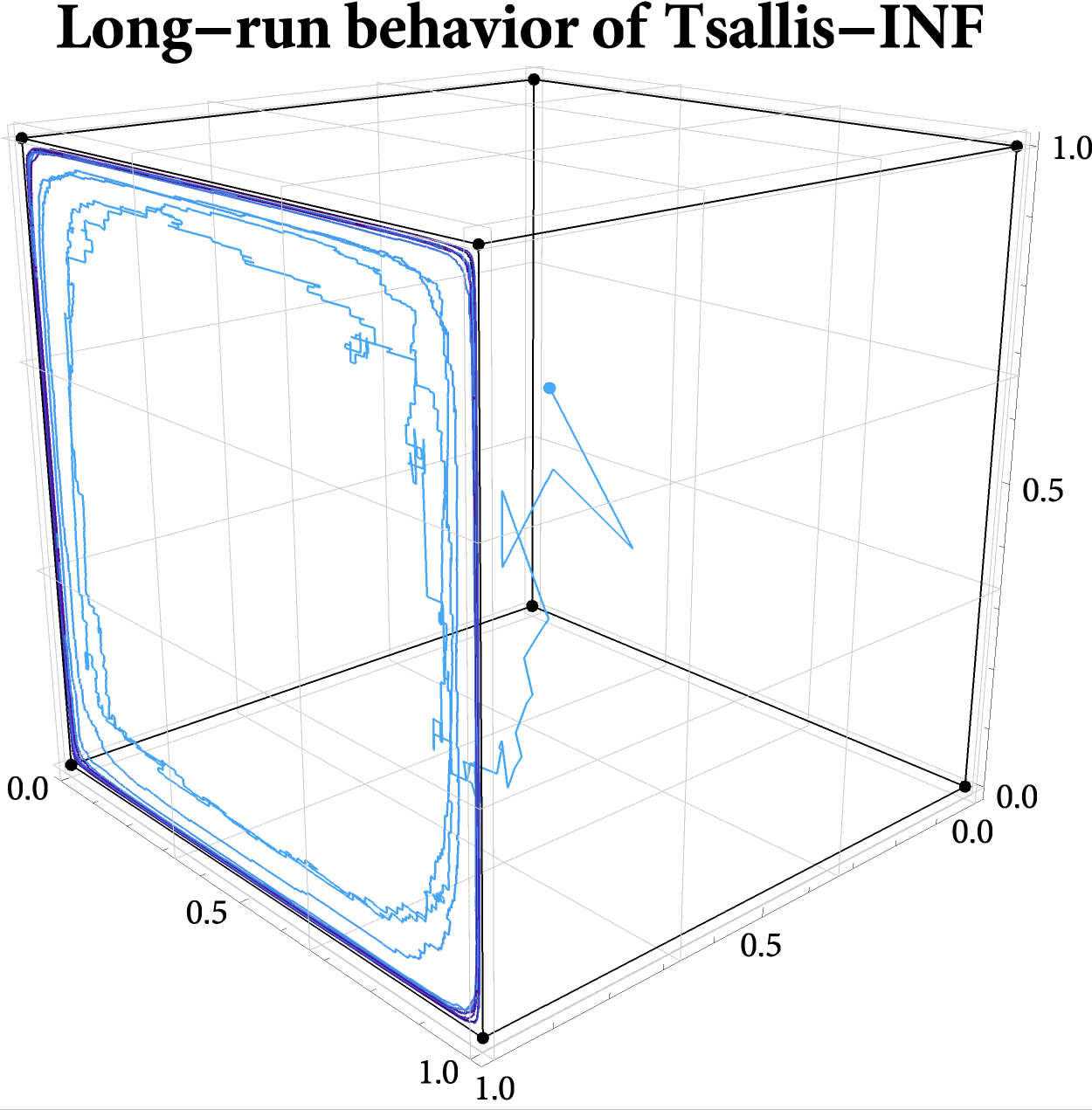}
%\caption{Payoff-based methods (\cref{alg:EXP3,alg:Tsallis}).}
\end{subfigure}
%\\[\medskipamount]
%\begin{subfigure}{.48\textwidth}
%\includegraphics[height=.48\textwidth]{Figures/Hedge-Het.pdf}
%%\;\;
%\includegraphics[height=.48\textwidth]{Figures/OMW-Het.pdf}
%\caption{Full-information methods (\cref{alg:EW,alg:OMW}).}
%\end{subfigure}
%\hfill
%\begin{subfigure}{.48\textwidth}
%\includegraphics[height=.48\textwidth]{Figures/EXP3-Het.pdf}
%%\;\;
%\includegraphics[height=.48\textwidth]{Figures/Tsallis-Het.pdf}
%\caption{Payoff-based methods (\cref{alg:EXP3,alg:Tsallis}).}
%\end{subfigure}
\caption{The long-run behavior of \crefrange{alg:FTRL}{alg:BFTRL} in a $2\times2\times2$ game.
\cref{alg:FTRL,alg:OFTRL} were run with a logit choice map as per \eqref{eq:EW};
\cref{alg:BFTRL} was run with both variants, \ac{EXP3} and \ac{Tsallis}.
All algorithms were run for $5\times10^{5}$ iterations with $\curr[\step] = 1/\run^{0.4}$ and $\curr[\mix] = 0.1/\run^{0.15}$;
%and they all converge to the same internally resilient set.
color indicates time, with darker hues indicating later iterations.
%\beginrev
The face to the left is \acl{closed}, so $\curr$ converges quickly to said face (as per \cref{thm:closed,thm:rate}).
}
\label{fig:limit}
\end{figure}

%% Convergence figure ends here
%----------------------------------------------------------------------

In general, the two defining elements of \eqref{eq:method} are
\begin{enumerate*}
[\itshape a\upshape)]
\item
the regularizer of the method;
and
\item
the feedback available to the players.
\end{enumerate*}
From our experiments, we conclude that methods with Euclidean regularization tend to have faster identification rates (\ie converge to the support of an equilibrium / \ac{closed} set faster), but they are more ``extreme'' than methods with an ``entropy-like'' regularizer (in the sense that players tend to play pure strategies more often).
As for the feedback available to the players, payoff-based methods tend to have higher variance (and hence a slower rate of convergence) relative to methods with full information;
otherwise, from a qualitative viewpoint, there are no perceptible differences in their limiting behavior.

Finally, optimistic / extra-oracle methods with full information exhibit better convergence properties in two-player zero-sum games (relative to standard \acs{FTRL} policies);
however, this is a fragile advantage that evaporates in the presence of noise and/or uncertainty (in which case "vanilla" and "optimistic" methods are essentially indistinguishable).
We illustrate these findings in \cref{fig:limit}.

Regarding \cref{fig:games}, the payoffs of the chosen games were normalized to $[-1,1]$ and players are assumed to choose between two actions labeled ``$O$'' and ``$1$''.
The specific tableaus are shown in the table below, next to the respective portrait (all taken from \cref{fig:games}.

\noindent
\begin{tikzpicture}
	\footnotesize
	\draw (-5, -0.5) to (-5, 2.0);
	\node[rotate=90, anchor=south] at (-5, 0.75) {\sc Dynamics};
	\node[anchor=north east] at (-2.2, 2.2)
	{\includegraphics[width=2.5cm]{Figures/StrictNash.pdf}};

	\draw (-1.75, -0.5) to (-1.75, 2.0);
	\node[rotate=90, anchor=south] at (-1.75, 0.75) {\sc Utilities};
	\node[anchor=east] at (-0.5,1.0) {P.~I};
	\node[anchor=east] at (-0.5,0.5) {P.~II};
	\node[anchor=east] at (-0.5,0.0) {P.~III};
	\node[rotate=30, anchor=west] at (-1.5, 1.45) {\small I, II, III};
	\node[rotate=30, anchor=west] at (0,1.5) {$(0,0,0)$};
	\node[rotate=30, anchor=west] at (1,1.5) {$(0,0,1)$};
	\node[rotate=30, anchor=west] at (2,1.5) {$(0,1,0)$};
	\node[rotate=30, anchor=west] at (3,1.5) {$(0,1,1)$};
	\node[rotate=30, anchor=west] at (4,1.5) {$(1,0,0)$};
	\node[rotate=30, anchor=west] at (5,1.5) {$(1,0,1)$};
	\node[rotate=30, anchor=west] at (6,1.5) {$(1,1,0)$};
	\node[rotate=30, anchor=west] at (7,1.5) {$(1,1,1)$};
	
	\draw (-0.25,-0.33) to (-0.25,1.33) to (7.75,1.33) to (7.75,-0.33) to (-0.25,-0.33);

	\node at (0.25,1.0) {1}; 
	\node at (1.25,1.0) {0}; 
	\node at (2.25,1.0) {0}; 
	\node at (3.25,1.0) {1}; 
	\node at (4.25,1.0) {0}; 
	\node at (5.25,1.0) {1}; 
	\node at (6.25,1.0) {1}; 
	\node at (7.25,1.0) {0}; 

	\node at (0.25, 0.5) {1}; 
	\node at (1.25, 0.5) {0}; 
	\node at (2.25, 0.5) {0}; 
	\node at (3.25, 0.5) {1}; 
	\node at (4.25, 0.5) {0}; 
	\node at (5.25, 0.5) {1}; 
	\node at (6.25, 0.5) {1}; 
	\node at (7.25, 0.5) {0}; 

	\node at (0.25, 0.0) {1}; 
	\node at (1.25, 0.0) {0}; 
	\node at (2.25, 0.0) {0}; 
	\node at (3.25, 0.0) {1}; 
	\node at (4.25, 0.0) {0}; 
	\node at (5.25, 0.0) {1}; 
	\node at (6.25, 0.0) {1}; 
	\node at (7.25, 0.0) {0}; 

\end{tikzpicture}

\noindent
\begin{tikzpicture}
	\footnotesize
	\draw (-5, -0.5) to (-5, 2.0);
	\node[rotate=90, anchor=south] at (-5, 0.75) {\sc Dynamics};
	\node[anchor=north east] at (-2.2, 2.2)
	{\includegraphics[width=2.5cm]{Figures/Spectator.pdf}};

	\draw (-1.75, -0.5) to (-1.75, 2.0);
	\node[rotate=90, anchor=south] at (-1.75, 0.75) {\sc Utilities};
	\node[anchor=east] at (-0.5, 1) {P.~I};
	\node[anchor=east] at (-0.5, 0.5) {P.~II};
	\node[anchor=east] at (-0.5, 0.0) {P.~III};
	\node[rotate=30, anchor=west] at (-1.5, 1.45) {\small I, II, III};
	\node[rotate=30, anchor=west] at (0,1.5) {$(0,0,0)$};
	\node[rotate=30, anchor=west] at (1,1.5) {$(0,0,1)$};
	\node[rotate=30, anchor=west] at (2,1.5) {$(0,1,0)$};
	\node[rotate=30, anchor=west] at (3,1.5) {$(0,1,1)$};
	\node[rotate=30, anchor=west] at (4,1.5) {$(1,0,0)$};
	\node[rotate=30, anchor=west] at (5,1.5) {$(1,0,1)$};
	\node[rotate=30, anchor=west] at (6,1.5) {$(1,1,0)$};
	\node[rotate=30, anchor=west] at (7,1.5) {$(1,1,1)$};
	
	\draw (-0.25,-0.33) to (-0.25,1.33) to (7.75,1.33) to (7.75,-0.33) to (-0.25,-0.33);

	\node at (0.25,1.0) {1}; 
	\node at (1.25,1.0) {1}; 
	\node at (2.25,1.0) {0}; 
	\node at (3.25,1.0) {0}; 
	\node at (4.25,1.0) {0}; 
	\node at (5.25,1.0) {0}; 
	\node at (6.25,1.0) {1}; 
	\node at (7.25,1.0) {1}; 

	\node at (0.25,0.5) {1}; 
	\node at (1.25,0.5) {1}; 
	\node at (2.25,0.5) {0}; 
	\node at (3.25,0.5) {0}; 
	\node at (4.25,0.5) {0}; 
	\node at (5.25,0.5) {0}; 
	\node at (6.25,0.5) {1}; 
	\node at (7.25,0.5) {1}; 

	\node at (0.25,0.0) {0}; 
	\node at (1.25,0.0) {0}; 
	\node at (2.25,0.0) {0}; 
	\node at (3.25,0.0) {0}; 
	\node at (4.25,0.0) {0}; 
	\node at (5.25,0.0) {0}; 
	\node at (6.25,0.0) {0}; 
	\node at (7.25,0.0) {0}; 

\end{tikzpicture}

\noindent
\begin{tikzpicture}
	\footnotesize
	\draw (-5, -0.5) to (-5, 2.0);
	\node[rotate=90, anchor=south] at (-5, 0.75) {\sc Dynamics};
	\node[anchor=north east] at (-2.2, 2.2)
	{\includegraphics[width=2.5cm]{Figures/TwistedMP.pdf}};

	\draw (-1.75, -0.5) to (-1.75, 2.0);
	\node[rotate=90, anchor=south] at (-1.75, 0.75) {\sc Utilities};
	\node[anchor=east] at (-0.5, 1) {P.~I};
	\node[anchor=east] at (-0.5, 0.5) {P.~II};
	\node[anchor=east] at (-0.5, 0.0) {P.~III};
	\node[rotate=30, anchor=west] at (-1.5, 1.45) {\small I, II, III};
	\node[rotate=30, anchor=west] at (0,1.5) {$(0,0,0)$};
	\node[rotate=30, anchor=west] at (1,1.5) {$(0,0,1)$};
	\node[rotate=30, anchor=west] at (2,1.5) {$(0,1,0)$};
	\node[rotate=30, anchor=west] at (3,1.5) {$(0,1,1)$};
	\node[rotate=30, anchor=west] at (4,1.5) {$(1,0,0)$};
	\node[rotate=30, anchor=west] at (5,1.5) {$(1,0,1)$};
	\node[rotate=30, anchor=west] at (6,1.5) {$(1,1,0)$};
	\node[rotate=30, anchor=west] at (7,1.5) {$(1,1,1)$};
	
	\draw (-0.25,-0.33) to (-0.25,1.33) to (7.75,1.33) to (7.75,-0.33) to (-0.25,-0.33);

	\node at (0.25,1.0) {0}; 
	\node at (1.25,1.0) {0}; 
	\node at (2.25,1.0) {0}; 
	\node at (3.25,1.0) {0}; 
	\node at (4.25,1.0) {0.1}; 
	\node at (5.25,1.0) {0.1}; 
	\node at (6.25,1.0) {0.1}; 
	\node at (7.25,1.0) {0.1}; 

	\node at (0.25,0.5) {1}; 
	\node at (1.25,0.5) {0}; 
	\node at (2.25,0.5) {0}; 
	\node at (3.25,0.5) {1}; 
	\node at (4.25,0.5) {0}; 
	\node at (5.25,0.5) {1}; 
	\node at (6.25,0.5) {1}; 
	\node at (7.25,0.5) {0}; 

	\node at (0.25,0.0) {0}; 
	\node at (1.25,0.0) {1}; 
	\node at (2.25,0.0) {1}; 
	\node at (3.25,0.0) {0}; 
	\node at (4.25,0.0) {1}; 
	\node at (5.25,0.0) {0}; 
	\node at (6.25,0.0) {0}; 
	\node at (7.25,0.0) {1}; 

\end{tikzpicture}

\noindent
\begin{tikzpicture}
	\footnotesize
	\draw (-5, -0.5) to (-5, 2.0);
	\node[rotate=90, anchor=south] at (-5, 0.75) {\sc Dynamics};
	\node[anchor=north east] at (-2.2, 2.2)
	{\includegraphics[width=2.5cm]{Figures/OutsideMP.pdf}};

	\draw (-1.75, -0.5) to (-1.75, 2.0);
	\node[rotate=90, anchor=south] at (-1.75, 0.75) {\sc Utilities};
	\node[anchor=east] at (-0.5, 1) {P.~I};
	\node[anchor=east] at (-0.5, 0.5) {P.~II};
	\node[anchor=east] at (-0.5, 0.0) {P.~III};
	\node[rotate=30, anchor=west] at (-1.5, 1.45) {\small I, II, III};
	\node[rotate=30, anchor=west] at (0,1.5) {$(0,0,0)$};
	\node[rotate=30, anchor=west] at (1,1.5) {$(0,0,1)$};
	\node[rotate=30, anchor=west] at (2,1.5) {$(0,1,0)$};
	\node[rotate=30, anchor=west] at (3,1.5) {$(0,1,1)$};
	\node[rotate=30, anchor=west] at (4,1.5) {$(1,0,0)$};
	\node[rotate=30, anchor=west] at (5,1.5) {$(1,0,1)$};
	\node[rotate=30, anchor=west] at (6,1.5) {$(1,1,0)$};
	\node[rotate=30, anchor=west] at (7,1.5) {$(1,1,1)$};
	
	\draw (-0.25,-0.33) to (-0.25,1.33) to (7.75,1.33) to (7.75,-0.33) to (-0.25,-0.33);

	\node at (0.25,1.0) {-1}; 
	\node at (1.25,1.0) {1}; 
	\node at (2.25,1.0) {-1}; 
	\node at (3.25,1.0) {1}; 
	\node at (4.25,1.0) {1}; 
	\node at (5.25,1.0) {-1}; 
	\node at (6.25,1.0) {1}; 
	\node at (7.25,1.0) {-1}; 

	\node at (0.25,0.5) {1}; 
	\node at (1.25,0.5) {1}; 
	\node at (2.25,0.5) {-1}; 
	\node at (3.25,0.5) {-1}; 
	\node at (4.25,0.5) {-1}; 
	\node at (5.25,0.5) {1}; 
	\node at (6.25,0.5) {1}; 
	\node at (7.25,0.5) {-1}; 

	\node at (0.25,0.0) {1}; 
	\node at (1.25,0.0) {-1}; 
	\node at (2.25,0.0) {-1}; 
	\node at (3.25,0.0) {1}; 
	\node at (4.25,0.0) {-1}; 
	\node at (5.25,0.0) {1}; 
	\node at (6.25,0.0) {1}; 
	\node at (7.25,0.0) {-1}; 

\end{tikzpicture}

%**********************************************************************
%***    BIBLIOGRAPHY
%**********************************************************************
\bibliographystyle{icml}
\bibliography{bibtex/IEEEabrv,bibtex/Bibliography-PM}

\end{document}